\documentclass[11pt,a4paper]{article}
\usepackage[a4paper]{geometry}
\usepackage{amsthm,amsfonts,amstext,amssymb,mathrsfs,amsmath,latexsym,mathtools}
\usepackage{mathtools}
\usepackage{fullpage}

\usepackage{suffix}
\usepackage[title,titletoc]{appendix}
\usepackage{comment}
\usepackage{overpic}
\usepackage{mathrsfs,enumerate}
\usepackage{amsmath}
\usepackage{amsthm}
\usepackage{amssymb}
\usepackage{amsfonts}
\usepackage{epsf}
\usepackage{hyperref}
\usepackage{indentfirst}
\usepackage{seqsplit}

\usepackage{xcolor}

\DeclareMathOperator{\capac}{cap}
\DeclareMathOperator{\diag}{diag}

\DeclareMathOperator{\supp}{supp}
\DeclareMathOperator{\bal}{bal}
\DeclareMathOperator{\re}{Re}
\DeclareMathOperator{\im}{Im}

\newcommand{\R}{\mathbb{R}}

\newcommand{\C}{\mathbb{C}}

\newcommand{\ud}{\,\mathrm{d}}
\newcommand{\Boh}{\mathcal{O}}
\newcommand{\boh}{\mathit{o}}

\usepackage{tikz}
\usepackage{pgflibraryshapes}
\usetikzlibrary{arrows,decorations.markings}

\tikzset{->-/.style={decoration={
  markings,
  mark=at position #1 with {\arrow[thin,scale=2.5]{>}}},postaction={decorate,line width=1pt}}}
\tikzset{-<-/.style={decoration={
  markings,
  mark=at position #1 with {\arrow[thin,scale=2.5]{<}}},postaction={decorate,line width=1pt}}}

\tikzstyle{vecArrow} = [thick, decoration={markings,mark=at position
   1 with {\arrow[semithick]{open triangle 60}}},
   double distance=1.4pt, shorten >= 5.5pt,
   preaction = {decorate},
   postaction = {draw,line width=1.4pt, white,shorten >= 4.5pt}]
\tikzstyle{innerWhite} = [semithick, white,line width=1.4pt, shorten >= 4.5pt]

\usetikzlibrary{decorations.pathreplacing,decorations.markings,shapes.geometric,matrix,calc}
\tikzset{naming/.style={align=center,font=\small}}
\tikzset{antenna/.style={insert path={-- coordinate (ant#1) ++(0,0.25) -- +(135:0.25) + (0,0) -- +(45:0.25)}}}
\tikzset{station/.style={naming,draw,shape=dart,shape border rotate=90, minimum width=10mm, minimum height=10mm,outer sep=0pt,inner sep=3pt}}
\tikzset{mobile/.style={naming,draw,shape=rectangle,minimum width=12mm,minimum height=6mm, outer sep=0pt,inner sep=3pt}}
\tikzset{radiation/.style={{decorate,decoration={expanding waves,angle=90,segment length=4pt}}}}

\usetikzlibrary{arrows}

\makeatletter
\def\grd@save@target#1{%
  \def\grd@target{#1}}
\def\grd@save@start#1{%
  \def\grd@start{#1}}
\tikzset{
  grid with coordinates/.style={
    to path={%
      \pgfextra{%
        \edef\grd@@target{(\tikztotarget)}%
        \tikz@scan@one@point\grd@save@target\grd@@target\relax
        \edef\grd@@start{(\tikztostart)}%
        \tikz@scan@one@point\grd@save@start\grd@@start\relax
        \draw[minor help lines] (\tikztostart) grid (\tikztotarget);
        \draw[major help lines] (\tikztostart) grid (\tikztotarget);
        \grd@start
        \pgfmathsetmacro{\grd@xa}{\the\pgf@x/1cm}
        \pgfmathsetmacro{\grd@ya}{\the\pgf@y/1cm}
        \grd@target
        \pgfmathsetmacro{\grd@xb}{\the\pgf@x/1cm}
        \pgfmathsetmacro{\grd@yb}{\the\pgf@y/1cm}
        \pgfmathsetmacro{\grd@xc}{\grd@xa + \pgfkeysvalueof{/tikz/grid with coordinates/major step}}
        \pgfmathsetmacro{\grd@yc}{\grd@ya + \pgfkeysvalueof{/tikz/grid with coordinates/major step}}
        \foreach \x in {\grd@xa,\grd@xc,...,\grd@xb}
        \node[anchor=north] at (\x,\grd@ya) {\pgfmathprintnumber{\x}};
        \foreach \y in {\grd@ya,\grd@yc,...,\grd@yb}
        \node[anchor=east] at (\grd@xa,\y) {\pgfmathprintnumber{\y}};
      }
    }
  },
  minor help lines/.style={
    help lines,
    step=\pgfkeysvalueof{/tikz/grid with coordinates/minor step}
  },
  major help lines/.style={
    help lines,
    line width=\pgfkeysvalueof{/tikz/grid with coordinates/major line width},
    step=\pgfkeysvalueof{/tikz/grid with coordinates/major step}
  },
  grid with coordinates/.cd,
  minor step/.initial=.2,
  major step/.initial=1,
  major line width/.initial=1.2pt,
}
\makeatother
\tikzset{every state/.style={minimum size=0pt}}

\newtheorem{thm}{Theorem}[section]

\newtheorem{lem}[thm]{Lemma}
\newtheorem{prop}[thm]{Proposition}

\theoremstyle{definition}

\newtheorem{rhp}[thm]{RH Problem}
\numberwithin{equation}{section}

\newcommand\restr[2]{{#1}\raisebox{-.5ex}{$|$}_{#2}}

\theoremstyle{remark}
\newtheorem{remark}[thm]{Remark}
\numberwithin{equation}{section}

\newcommand{\eq}{\begin{equation}}
\newcommand{\nq}{\end{equation}}
\newcommand{\eqa}{\begin{eqnarray}}
\newcommand{\nqa}{\end{eqnarray}}

\begin{document}

\title{Large $n$ limit for the product of two coupled random matrices}
\author{Guilherme L. F. Silva\footnote{{Department of Mathematics, University of Michigan, Ann Arbor, MI, 48109, USA.
E-mail: \texttt{silvag@umich.edu}}} ~and Lun Zhang \footnote{School of Mathematical Sciences
and Shanghai Key Laboratory for Contemporary Applied Mathematics,
Fudan University, Shanghai 200433, People's Republic of China. E-mail: \texttt{lunzhang@fudan.edu.cn}}}
\date{}
\maketitle

\begin{abstract}
For a pair of coupled rectangular random matrices we consider the squared singular values of their product, which form a determinantal point process. We show that the limiting mean distribution of these squared singular values is described by the second component of the solution to a vector equilibrium problem. This vector equilibrium problem is defined for three measures with an upper constraint on the first measure and an external field on the second measure. We carry out the steepest descent analysis for a 4 $\times$ 4 matrix-valued Riemann-Hilbert problem, which characterizes the correlation kernel and is related to mixed type multiple orthogonal polynomials associated with the modified Bessel functions. A careful study of the vector equilibrium problem, combined with this asymptotic analysis, ultimately leads to the aforementioned convergence result for the limiting mean distribution, an explicit form of the associated spectral curve, as well as local Sine, Meijer-G and Airy universality results for the squared singular values considered.
\end{abstract}

\renewcommand{\thefootnote}{\arabic{footnote}}

\tableofcontents

\section{Introduction}
The studies of products of random matrices might be traced back to the work of Furstenberg and Kesten \cite{FK60} in the context of random Schr\"{o}dinger operators \cite{BL85} and statistical physics relating to disordered and chaotic dynamical systems \cite{CPV93} in the 1960s. The emphasis at that time was put on the statistical behavior of individual entries as the number of factors in the product tends to infinity. The recent rapid developments, however, are focused on the eigenvalue or singular value distributions, at various scales, as the sizes of the matrices tend to infinity while the number of matrices in the product is kept fixed.

Among various progresses of this aspect, significant contributions are due to the works of Akemann, Ipsen, Kieburg and Wei \cite{AIK13,AKW}, in which they showed that the squared singular values of products of independent complex Gaussian matrices (i.e., the matrices whose entries are independent with a complex Gaussian distribution, also known as Ginibre random matrices) form a determinantal point process over the positive real axis. The various local limits of the correlation kernel then reveal an interesting mathematical structure behind the products of independent random matrices, and various scaling limits can be predicted once one knows properties of the global distribution of the squared singular values \cite{BJLNS}. On one hand, after proper centering and scaling, the correlation kernel tends to the sine kernel for points in the bulk, and to the Airy kernel for the right endpoint of the limiting spectrum \cite{LWZ16}, which obey the principle of universality in random matrix theory \cite{KuiBookChapter}. One the other hand, a new family of kernels, namely, the Meijer G-kernels, are found to describe the scaling limit of the correlation kernel near the origin \cite{KZ}. The Meijer G-kernels generalize the classical Bessel kernel and represent a new universality class in random matrix theory as evidenced by their later appearances in many other random matrix models including Cauchy-chain matrix models \cite{BB15,BGS14}, products of Ginibre matrices with inverse ones \cite{Fore14}, Muttalib-Borodin ensembles \cite{Bor,KS14}, a matrix model with Bures measure \cite{ForKie16}, etc. For more information about recent results for products of independent random matrices, we refer to the review article \cite{AI15} and references therein.

In view of these interesting results obtained for products of independent complex Gaussian random matrices, a natural question to ask is how far such results remain valid, or yet if different ones arise, if some of the conditions on the models are relaxed. One attempt towards this direction is to drop the requirement of independence of the matrices in the product, as initiated by Akemann and Strahov \cite{AS15} and further explored by them and Liu \cite{AS18,AS16b,Liu16}. Following \cite{AS15,Liu16}, let us consider a coupled two-matrix model defined by the probability distribution
\begin{equation}\label{def:coupledmatrix}
\frac{1}{\widehat Z_n} \exp \left(-\beta \textrm{Tr}(X_1X_1^*+X_2^*X_2)+\textrm{Tr}(\Omega X_1 X_2 +(\Omega X_1 X_2)^*)\right)\ud X_1\ud X_2,
\end{equation}
over pairs of rectangular complex matrices $(X_1,X_2)$ of sizes $L\times M$ and $M\times n$ respectively,
where the superscript $^*$ stands for the Hermitian transpose, $\ud X_1$ and $\ud X_2$ are the flat complex Lebesgue measures on the entries of $X_1$ and $X_2$, and $\widehat Z_n$ is a normalization constant. Here, $\beta>0$ and $\Omega$ is a fixed $n\times L$ complex matrix playing the role of coupling between $X_1$ and $X_2$, which should satisfy
\begin{equation}\label{eq:coupling_strength}
\Omega\Omega^* < \beta^2 I_n
\end{equation}
to make sure the model is well defined, where $I_n$ is the $n\times n$ identity matrix,. The interest lies in the singular values of the product matrix
\begin{equation}\label{def:Ymatrix}
\widehat Y:=X_1X_2,
\end{equation}
where the matrices $X_1$ and $X_2$ are drawn from \eqref{def:coupledmatrix}.

There are several motivations to study the product \eqref{def:Ymatrix}. First, if $L=n$ and $\Omega$ is a scalar matrix, the model \eqref{def:coupledmatrix} can be interpreted as the chiral two-matrix model \cite{ADOS07,OSborn04}, which was introduced in the context of quantum chromodynamics (QCD). In this case, an alternative formulation of the model is the following (see \cite{AS15,Liu16}). Let $A$ and $B$ be two independent matrices of size $n\times M (M\geq n)$ with independent and identically distributed standard complex Gaussian entries. Define two random matrices
      \begin{equation}\label{def:X12}
      X_1:=\frac{1}{\sqrt{2}}(A-i\sqrt{\tau}B),\qquad X_2:=\frac{1}{\sqrt{2}}(A^*-i\sqrt{\tau}B^*), \qquad 0<\tau<1.
      \end{equation}
 Then the pair $(X_1,X_2)$ is distributed according to \eqref{def:coupledmatrix} with
\begin{equation}\label{def:interpolating_parameters}
 L=n, \quad \beta=\frac{1+\tau}{2\tau}\quad \mbox{and}\quad \Omega=\frac{1-\tau}{2\tau}I_n,
\end{equation}
 and one can see $\tau$ as an interpolation parameter between a model for singular values of the Ginibre matrix $A$ (corresponding to $\tau=0$) and a correlated product (for $\tau>0$).

Also, in the context of QCD with a baryon chemical potential \cite{OSborn04}, the Dirac operator is realized as a block matrix whose diagonal entries are null matrices and the off-diagonal entries are matrices of the form \eqref{def:X12}. The singular values of $\widehat Y$ can be viewed as the correlations of complex eigenvalues of the QCD dirac operator.

In addition, as observed in \cite{AS15}, the product of $X_1$ and $X_2$ defined in \eqref{def:X12} provides a new interpolating ensemble, in a sense extending \eqref{def:X12}--\eqref{def:interpolating_parameters} to a rectangular coupling matrix $\Omega$. It interpolates between the classical Laguerre ensemble \cite{Mehta} (for $\tau=0$) and the product of two independent Ginibre random matrices (for $\tau=1$).

A striking feature is that the squared singular values of $\widehat Y$  are distributed according to a determinantal point process over the positive real axis \cite{AS15,Liu16}. This determinantal point process is a biorthogonal ensemble \cite{Bor} with joint probability density function (see \cite[Proposition 1.1]{Liu16})
\begin{equation}\label{def:jpdf}
\frac{1}{Z_n}\det \left[I_{\kappa}(2\alpha_i\sqrt{x_j})\right]_{i,j=1}^n
\det\left[x_j^{\frac{\nu+i-1}{2}}K_{\nu-\kappa+i-1}(2\beta\sqrt{x_j})\right]_{i,j=1}^n,
\end{equation}
with $I_\mu$ and $K_\nu$ being the modified Bessel functions of first kind and second kind, respectively, where
\begin{equation}\label{def:kappanu}
\kappa:=L-n, \qquad \nu:=M-n,
\end{equation}
$\alpha_1,\hdots, \alpha_n$ are the singular values of the coupling matrix $\Omega$ and $Z_n$ is a normalization constant explicitly known.
The correlation kernel describing the point process \eqref{def:jpdf} admits a double contour integral representation, which can be used to establish various limits near the origin if one further couples the $\alpha_i$'s and $\beta$ on one parameter; see \cite{AS16b,AS15,Liu16} for details. In particular, the universal Meijer G-kernel also appears in one of these limits.

An interesting yet open question posed in \cite{AS15} is to find the limiting mean distribution of the singular values for $\widehat Y$ and the local limits of the correlation kernel beyond the origin. Due to the missing of
independence of the matrices, the challenge we encounter is the fact that the approaches developed for the products of independent matrices are not applicable directly. The main contribution of this paper to fully resolve this problem, and along the way obtain several other asymptotic results when, in contrast with the mentioned previous works, the parameters $\alpha$ and $\beta$ are not coupled together.

\section{Statement of results}

\subsection{The confluent case}
We will focus on the confluent case that all the singular values of $\Omega$ are the same, i.e.,
\begin{equation}\label{eq:alphailimit}
\alpha_i \to \alpha>0.
\end{equation}
In virtue of \eqref{eq:coupling_strength}, we stress that
\begin{equation}\label{eq:inequality_alpha_beta}
\alpha<\beta,
\end{equation}
condition which is not a restriction but only ensures the model \eqref{def:coupledmatrix} is well defined.
In addition, it is assumed that
\begin{equation*}
M \geq L \geq n,
\end{equation*}
so that
\begin{equation*}
\nu \geq \kappa \geq 0.
\end{equation*}
The condition $ M,L \geq n$ assures us that, almost surely, $X_1$ and $X_2$ do not have $0$ as a singular value, and the case $L<M$ can be handled by swapping the roles of $X_1$ and $X_2$.

Under the condition \eqref{eq:alphailimit}, the vector space spanned by the functions $x \mapsto I_{\kappa}(2\alpha_j\sqrt{x})$, $j=1,\ldots,n$, becomes
the linear space spanned by
\begin{equation*}
x \mapsto \frac{\partial^{j-1}}{\partial y^{j-1}}I_{\kappa}(2y \sqrt{x})|_{y=\alpha}, \quad j=1,\ldots,n.
\end{equation*}
Using the recurrence relations (see \cite[Equation 10.29.1]{DLMF})
\begin{equation*}
I_{\mu-1}(z)-I_{\mu+1}(z) =\frac{2\mu}{z}I_{\mu}(z), \qquad
I_{\mu-1}(z)+I_{\mu+1}(z) = 2I_{\mu}'(z),
\end{equation*}
satisfied by the modified Bessel functions of the first kind, it is readily seen that the resulting space is spanned by the functions $x \mapsto x^{\frac{j-1}{2}}I_{\kappa+j-1}(2\alpha\sqrt{x})$, $j=1,\ldots,n$. Thus, a further algebraic calculation implies that the joint probability density function for the squared singular values of $\widehat Y$ is given by
\begin{equation}\label{eq:jpdfcon}
\frac{1}{Z_n}\det \left[x_k^{\frac{\kappa+j-1}{2}}I_{\kappa+j-1}(2\alpha\sqrt{x_k})\right]_{j,k=1}^n
\det\left[x_k^{\frac{\nu-\kappa+j-1}{2}}K_{\nu-\kappa+j-1}(2\beta\sqrt{x_k})\right]_{j,k=1}^n,
\end{equation}
under the condition that the coupling matrix $\Omega$ has a single singular value $\alpha$. For the case $\kappa=0$, this result was first obtained by Akemann and Strahov \cite{AS15}.

From general properties of biorthogonal ensembles \cite{Bor}, it is known that \eqref{eq:jpdfcon}
is a determinantal point process with correlation kernel
\begin{equation} \label{def:Kn}
    K_n(x,y) = \sum_{k=0}^{n-1} \mathcal{Q}_k(x) \mathcal{P}_k(y),
    \end{equation}
where for each $k = 0, 1, \ldots$, $\mathcal{Q}_k$ belongs to the linear span of $x^{\frac{\kappa+j}{2}}I_{\kappa+j}(2\alpha\sqrt{x})$, $j=0,\ldots,k$, while $\mathcal{P}_k$ belongs to the linear span of $x^{\frac{\nu-\kappa+j}{2}}K_{\nu-\kappa+j}(2\beta \sqrt{x})$, $j=0,\ldots,k$, in such a way that
\begin{equation*}
    \int_0^{\infty} \mathcal{Q}_k(x) \mathcal{P}_j(x) \ud x = \delta_{j,k},
\end{equation*}
with $\delta_{j,k}$ being the Kronecker delta.

This characterization of $K_n$ will be the starting point of our work. To describe the large $n$ limit of the correlation kernel $K_n$, we introduce next a vector equilibrium problem.

\subsection{A vector equilibrium problem}\label{subsec:vep}
Given any two finite measures $\mu$ and $\nu$ on $\C$, we denote by, as usual (cf. \cite{SaffTotik}),
\begin{equation*}
I(\mu,\nu)=\iint \log\frac{1}{|x-y|}\ud \mu(x) \ud \nu(y)
\end{equation*}
their mutual logarithmic interaction, and by
\begin{equation}\label{eq:log_energy}
I(\mu)=I(\mu,\mu)=\iint \log \frac{1}{|x-y|} \ud\mu(x) \ud \mu(y)
\end{equation}
the logarithmic energy of the measure $\mu$.

The vector equilibrium problem relevant to the present work asks for minimizing the energy functional
\begin{equation}\label{definition_vector_energy}
E(\nu_1,\nu_2,\nu_3)= I(\nu_1)+I(\nu_2)+I(\nu_3)-I(\nu_1,\nu_2)-I(\nu_2,\nu_3)+2(\beta-\alpha)\int \sqrt{x}\ud \nu_2(x),
\end{equation}
over the set $\mathcal M$ of admissible measures, which is defined to be the set of triples of measures
$\pmb{\nu}=(\nu_1,\nu_2,\nu_3)$ satisfying the following conditions.
\begin{enumerate}
\item[(E1)] All three measures $\nu_1$, $\nu_2$ and $\nu_3$ have finite logarithmic energy.

\item[(E2)] $\nu_1$ is a measure on $\R_-:=(-\infty, 0]$ with total mass $1/2$, i.e.,
$2|\nu_1|=1$, and satisfies the upper constraint
    $$\nu_1 \leq \sigma,$$
where $\sigma$ is the absolutely continuous measure on $\R_-$ with density
\begin{equation}\label{def:constraint_measure}
\frac{\ud \sigma}{\ud x}(x)=\frac{\alpha}{\pi \sqrt{|x|}},\qquad x<0.
\end{equation}

\item[(E3)] $\nu_2$ is a measure on $\R_+:=[0,\infty)$ with total mass $1$, i.e., $|\nu_2|=1$,

\item[(E4)]  $\nu_3$ is a measure on $\R_-$ with total mass $1/2$, i.e.,
$2|\nu_3|=1$.
\end{enumerate}

At first sight, the exact form of $E(\cdot)$ and the conditions on the measures might look mysterious. In the Appendix~\ref{section:heuristics_equil_probl} we present the calculations that led us to this exact form. Similar vector equilibrium problems have appeared before in the literature \cite{DKRZ12,DK09,duits_kuijlaars_mo,KMW09}, and existence and uniqueness of solution are known under very mild conditions \cite{BKMW13,hardy_kuijlaars} which include ours.

Our first result concerns the structure of the minimizer of the above equilibrium problem.

\begin{thm}\label{thm_equilibrium_problem}
There exists a unique vector of measures $\pmb{\mu}=(\mu_1,\mu_2,\mu_3)\in \mathcal M$ that minimizes the energy functional \eqref{definition_vector_energy} over $\mathcal M$. In addition, the components $\mu_1$, $\mu_2$ and $\mu_3$ have the following properties.
\begin{enumerate}[$($$a$$)$]

\item The support of $\mu_1$ is the negative real axis, and
   \begin{equation*}
   \supp(\sigma-\mu_1)=(-\infty,-q]
   \end{equation*}
with
\begin{equation}\label{eq:formula_q_endpoint}
q=\frac{\alpha ^6+\beta ^6-33 \alpha ^4 \beta ^2-33 \alpha ^2 \beta ^4+\sqrt{\left(\alpha ^4+14 \alpha ^2 \beta ^2+\beta ^4\right)^3}}{8 \alpha ^2 \beta ^2 \left(\beta ^2-\alpha^2\right)^2}>0.
\end{equation}

Furthermore, $\mu_1$ is absolutely continuous with respect to the Lebesgue measure and satisfies
\begin{equation}\label{eq:local_behavior_mu1}
\frac{\ud \sigma}{\ud x}(x)-\frac{\ud \mu_1}{\ud x}(x) = c_1(-q-x)^\frac{1}{2}(1+o(1)),\qquad x\to (-q)^-,
\end{equation}
for some positive constant $c_1$.

\item The support of $\mu_2$ is
\begin{equation*}
\supp\mu_2=[0,p]
\end{equation*}
with
\begin{equation}\label{eq:formula_p_endpoint}
p=\frac{-\alpha ^6-\beta ^6+33 \alpha ^4 \beta ^2+33 \alpha ^2 \beta ^4+\sqrt{\left(\alpha ^4+14 \alpha ^2 \beta ^2+\beta ^4\right)^3}}{8 \alpha ^2 \beta ^2 \left(\beta ^2-\alpha^2\right)^2}.
\end{equation}

Furthermore, $\mu_2$ is absolutely continuous with respect to the Lebesgue measure on $[0,p]$ and
\begin{equation}\label{eq:local_behavior_mu2}
\frac{\ud \mu_2}{\ud x}(x)=
\begin{cases}
c_2 x^{-\frac{2}{3}}(1+o(1)), & x\to 0^+, \\
\tilde c_2 (p-x)^{\frac{1}{2}}(1+o(1)), & x\to p^-, \\
\end{cases}
\end{equation}
for some positive constants $c_2$ and $\tilde c_2$.
\item The support of $\mu_3$ is the negative real axis and $\mu_3$ is absolutely continuous with respect to the
Lebesgue measure with density
\begin{equation}\label{eq:density_mu3}
\frac{\ud \mu_3}{\ud x}(x)=\frac{1}{2\pi \sqrt{|x|}}\int \frac{\sqrt{s}}{s-x}\ud \mu_2(s).
\end{equation}
In particular,
\begin{equation}\label{eq:local_behavior_mu3}
\frac{\ud \mu_3}{\ud x}(x)=c_3|x|^{-\frac{2}{3}}(1+o(1)),\quad x \to 0^-,
\end{equation}
for some positive constant $c_3$.
\end{enumerate}
\end{thm}

\subsection{The spectral curve}
One of the fundamental objects for a matrix model is its associated  {\it spectral curve} that has been explored for various other matrix models \cite{balogh_bertola_09,BEH03,BEH06,KT15,Neuschel14,MFS19}. To describe the spectral curve for the model  \eqref{def:Ymatrix}, denote by
\begin{equation}\label{def:Cauchytransform}
C^{\mu}(z)=\int \frac{\ud \mu(x)}{x-z}, \qquad z\in \mathbb{C}\setminus \supp \mu,
\end{equation}
the Cauchy transform of a measure $\mu$, let $\pmb{\mu}=(\mu_1,\mu_2,\mu_3)$ be the unique minimizer given in Theorem~\ref{thm_equilibrium_problem} and set
\begin{equation}\label{eq:definition_xi_functions}
\begin{aligned}
&\xi_1(z) =C^{\mu_1}(z)+\frac{\alpha} {\sqrt{z}}, && \quad z\in \C\setminus \R_-, \\
&\xi_2(z) =C^{\mu_2}(z)- C^{\mu_1}(z)-\frac{\alpha} {\sqrt{z}},  && \quad z\in \C\setminus (-\infty,p],\\
&\xi_3(z) =C^{\mu_3}(z)- C^{\mu_2}(z)-\frac{\beta} {\sqrt{z}}, && \quad z\in \C \setminus (-\infty,p], \\
&\xi_4(z) =-C^{\mu_3}(z)+\frac{\beta} {\sqrt{z}},  && \quad z\in \C\setminus \R_-,
\end{aligned}
\end{equation}
where the branch cut of the square root function $\sqrt{z}$ is taken along the negative real axis.
The spectral curve for \eqref{def:Ymatrix} takes the form of an algebraic equation and is given by our next theorem.
\begin{thm}\label{thm:spectral_curve}
The functions $\xi_1,\xi_2,\xi_3$ and $\xi_4$ are the four solutions to the algebraic equation
\begin{equation}\label{eq:spectral_curve}
\xi^4-\frac{\alpha^2+\beta^2}{z}\xi^2+\frac{\alpha^2-\beta^2}{z^2}\xi+\frac{\alpha^2\beta^2}{z^2}=0,
\end{equation}
and \eqref{eq:spectral_curve} is parametrized by
$$
(z,\xi)=\left(\frac{t}{h(t)^2},h(t)\right),\qquad t\in \overline \C,
$$
where
\begin{equation}\label{def:h}
h(t)=\frac{t^2-(\alpha^2+\beta^2)t+\alpha^2\beta^2}{\beta^2-\alpha^2}.
\end{equation}
\end{thm}

Using the parametrization of \eqref{eq:spectral_curve}, one can describe the densities of the components of $\pmb{\mu}$. For instance, the graph of the density of $\mu_2$ takes the form
$$
\left(x,\frac{\ud\mu_2}{\ud x}(x)\right)=\left(\frac{t}{h(t)^2},\pi i h(t)\right),\qquad t\in \gamma_2^-,
$$
where $\gamma_2^-$ is a specific contour on $\C$ along which $h$ becomes purely imaginary. We refer the reader to Section~\ref{sec:rational_parametrization} for details, in particular Figure~\ref{fig:uniformization} where $\gamma_2^-$ is evaluated numerically.

\subsection{Limiting mean distribution and hard edge scaling limit of the correlation kernel}
Our main result is the following theorem relating the large $n$ limit of the correlation kernel $K_n(x,y)$ to the unique minimizer of the vector equilibrium problem introduced in Section \ref{subsec:vep}.
\begin{thm}\label{thm:limitingmeandistri}
Let $K_n(x,y)$  be the correlation kernel defined in \eqref{def:Kn} for the squared singular values of $\widehat{Y}$ \eqref{def:Ymatrix} in the confluent case. With $\nu$ and $\kappa$ being fixed, we have
\begin{equation}\label{eq:convergence_kernels_global_scaling}
\lim_{n\to\infty}nK_{n}\left(n^2 x, n^2 x\right)=\frac{\ud \mu_2}{\ud x}(x),
\qquad x>0,
\end{equation}
where $\pmb{\mu}=(\mu_1,\mu_2,\mu_3)\in \mathcal M$ is the unique minimizer of the energy functional \eqref{definition_vector_energy} over $\mathcal M$ stated in Theorem \ref{thm_equilibrium_problem} and the limit above is uniform for $x$ in any compact subset of $(0,\infty)$.
\end{thm}

According to \eqref{eq:local_behavior_mu2}, the density of $\mu_2$ blows up at $x=0$, so it does not make sense to talk about the convergence \eqref{eq:convergence_kernels_global_scaling} when $x=0$. But, alternatively, the vector equilibrium problem stated in Section \ref{subsec:vep} is directly related to the
matrix model \eqref{def:Ymatrix} in the way we now explain, which then provides \eqref{eq:convergence_kernels_global_scaling} also for $x=0$ in a weaker sense. 

Let us denote by $y_1,\hdots,y_n$ the
squared singular values of the matrix $\widehat{Y}$ in \eqref{def:Ymatrix} and set
$$
P_n(z)=\mathbb E\left(\prod_{j=1}^n(z-y_j)\right),
$$
where the expectation is over the $y_j$'s with respect to the density in \eqref{eq:jpdfcon}.
That is, $P_n$ is the average characteristic polynomial for $\widehat Y^* \widehat Y$. If we denote by
$x_1,\hdots,x_n$ the zeros of $P_n$ and construct the sequence of zero counting measures
$$
\mu(P_n)=\frac{1}{n}\sum_{j=1}^n\delta_{x_j/n^2}
$$
with $\delta_a$ being the Dirac delta measure with mass at $a$, then the sequence $\{\mu(P_n)\}$ converges weakly to the second component $\mu_2$ of
the minimizer given in Theorem \ref{thm_equilibrium_problem}. This claim follows from the uniform convergence above, or also from the Riemann-Hilbert (shortly RH) asymptotic analysis that we perform.

We next come to the hard edge scaling limit of the correlation kernel. As aforementioned, if the parameters $\alpha$  and $\beta$ are coupled in a specific way, it was shown in \cite{AS15,Liu16} that the hard edge scaling limit of $K_n$ is given by the universal Meijer G-kernel, which in a format appropriate for us takes the form \cite{BGS14,KZ}
\begin{equation}\label{def:MeijerKer}
K_{\nu_1,\nu_2}(x,y)=\int_0^1 G^{1,0}_{0,3}\left({- \atop 0,-\nu_1,-\nu_2} \; \Big{|} \; ux \right)G^{2,0}_{0,3}\left({- \atop \nu_1, \nu_2, 0} \; \Big{|} \; uy \right) \ud u,
\end{equation}
where $G^{m,n}_{p,q}\left({a_1,\ldots,a_p \atop b_1,\ldots,b_q} \; \Big{|} \; z \right)$ is the Meijer G-function (see \eqref{def:Meijer} below for the definition). We extend the results just mentioned to any fixed $\alpha$ and $\beta$.

\begin{thm}\label{thm:hardedge}
Let $K_n$  be the correlation kernel defined in \eqref{def:Kn} for the squared singular values of $\widehat{Y}$ \eqref{def:Ymatrix} in the confluent case. With $\nu$ and $\kappa$ being fixed, we have
\begin{equation*}
\lim_{n\to\infty}\;\frac{1}{n(\beta^2-\alpha^2)}K_{n}\left(\frac{x}{n(\beta^2-\alpha^2)}, \frac{y}{n(\beta^2-\alpha^2)}\right)=\left(\frac y x\right)^{\kappa/2}K_{\nu,\kappa}(y,x),
\end{equation*}
uniformly for $x,y$ in compact subsets of $(0,\infty)$, where the limiting kernel $K_{\nu,\kappa}$ is given in \eqref{def:MeijerKer} and the parameters $\alpha,\beta$ satisfying \eqref{eq:inequality_alpha_beta} are fixed.
\end{thm}

Our asymptotic analysis, leading to the proofs of the theorems above, also allows us to obtain the expected universality results for the local statistics of the squared singular values of $\widehat{Y}$ beyond the origin. This means that the scaling limits of $K_n$ tend to the sine kernel when centered around a point $x_0\in (0,p)$ (bulk universality), and to the Airy kernel for $x_0=p$ (soft edge universality). All the ingredients for obtaining such results are presented, but we will not write the details down neither comment them any further; instead, we refer to \cite{ABK,BK04,DK09} for a more detailed analysis in similar situations.

\subsection{About the proofs and organization of the rest of the paper}
The proofs of our asymptotic results rely on the fact that the biorthogonal functions $\mathcal{P}_k$ and $\mathcal{Q}_k$ in \eqref{def:Kn} can be interpreted as multiple orthogonal polynomials of mixed type \cite{DK}, which is first observed by the second-named author in \cite{Zhang17}. This in particular implies an RH problem characterization \cite{DK} of the correlation kernel, which extends the classical results in \cite{Fokas,VGJK2001}, and of relevance to us here takes the following form.

\begin{rhp} \label{rhp:Y}
We look for a $4\times 4$ matrix-valued function
$Y : \mathbb C \setminus \R_+ \to \mathbb C^{4 \times 4}$ satisfying the following properties:
\begin{enumerate}
\item[\rm (1)] $Y$ is defined and analytic in $ \mathbb{C} \setminus \R_+$.
\item[\rm (2)] $Y$ has limiting values $Y_{\pm}$ on $(0,\infty)$,
where $Y_+$ ($Y_-$) denotes the limiting value from the upper
(lower) half-plane, and
\begin{equation}\label{defjumpmatrixY}
Y_{+}(x) = Y_{-}(x)
\begin{pmatrix} I_2 & W(x)\\
0 & I_2
\end{pmatrix}, \qquad  x \in \mathbb (0,+\infty),
\end{equation}
where $W(x)$ is the rank-one matrix
\begin{align}\label{eq:matrixW}
W(x) &=
\begin{pmatrix}
\omega_{\kappa,\alpha}(x) \\
\omega_{\kappa+1,\alpha}(x)
\end{pmatrix}
\begin{pmatrix}
\rho_{\nu-\kappa,\beta}(x) &
\rho_{\nu-\kappa+1,\beta}(x)
\end{pmatrix} \nonumber \\
&=\begin{pmatrix}
\omega_{\kappa,\alpha}(x)\rho_{\nu-\kappa,\beta}(x) & \omega_{\kappa,\alpha}(x)\rho_{\nu-\kappa+1,\beta}(x) \\
\omega_{\kappa+1,\alpha}(x)\rho_{\nu-\kappa,\beta}(x) & \omega_{\kappa+1,\alpha}(x)\rho_{\nu-\kappa+1,\beta}(x)
\end{pmatrix},
\end{align}
with
$$\omega_{\mu,a}(x)=x^{\frac{\mu}{2}}I_{\mu}(2an\sqrt{x}), \qquad\mu>-1, \qquad a>0,$$
and
$$\rho_{\nu,b}(x)=x^{\frac{\nu}{2}}K_{\nu}(2bn\sqrt{x}), \qquad \nu \geq 0, \qquad b>0.$$
In \eqref{eq:matrixW}, the parameters $\kappa$ and $\nu$ are given in \eqref{def:kappanu}.

\item[\rm (3)] As $z\to\infty$ and $z\in \mathbb{C}\setminus \R_+$, we have
\begin{equation*}
    Y(z) = \left(I_4+\frac{Y_1}{z}+\mathcal{O}\left(\frac{1}{z^2}\right)\right)
    \diag \left(z^{n_1},z^{n_2},z^{-n_1},z^{-n_2} \right).
\end{equation*}
with $n_1=\lfloor \frac{n-1}{2} \rfloor+1$ and $n_2=\lfloor \frac{n-2}{2} \rfloor+1$, where $\lfloor x \rfloor= \max\{n\in\mathbb{Z}:n \leq x\}$ stands for the integer part of $x$.

\item[\rm (4)]As $z\to 0$, $z\in \C\setminus \R_+$, the matrix $Y(z)$ has the following behavior:
\begin{equation}\label{eq:Yzero}
Y(z)=
\Boh
\begin{pmatrix}
1 & 1 & h(z) & \widetilde h(z) \\
1 & 1 & h(z) & \widetilde h(z) \\
1 & 1 & h(z) & \widetilde h(z)\\
1 & 1 & h(z) & \widetilde h(z)
\end{pmatrix},
\quad
Y^{-1}{}(z)=
\Boh
\begin{pmatrix}
h(z) & h(z) & h(z) & h(z) \\
\widetilde h(z) & \widetilde h(z) & \widetilde h(z) & \widetilde h(z) \\
1 & 1 & 1 & 1\\
1 & 1 & 1 & 1
\end{pmatrix},
\end{equation}
where
\begin{equation*}
h(z):=
\begin{cases}
1,& \kappa>0, \\
\log z, & \kappa=0, \; \nu>0, \\
(\log z)^2, & \kappa=\nu=0,
\end{cases}
\qquad
\widetilde h(z):=
\begin{cases}
1,& \kappa>0, \\
\log z, & \kappa=0,
\end{cases}
\end{equation*}
and the $\Boh$ condition in \eqref{eq:Yzero} is understood in an entry-wise manner.
\end{enumerate}
\end{rhp}

The above RH problem can be uniquely solved with the aid of mixed type multiple orthogonal polynomials associated with the modified Bessel functions; see \cite{Zhang17}. Moreover, a general result in \cite{DK} shows that the correlation kernel \eqref{def:Kn} admits the following representation in terms of the solution of the RH problem \ref{rhp:Y} for $Y$:
\begin{multline}\label{kernel representation}
n^2K_{n}(n^2x,n^2y)
\\
=\frac{1}{2\pi i(x-y)}\begin{pmatrix}0 &0 & \rho_{\nu-\kappa,\beta}(y)&
\rho_{\nu-\kappa+1,\beta}(y)\end{pmatrix} Y_{+}(y)^{-1}Y_{+}(x)
\begin{pmatrix}
\omega_{\kappa,\alpha}(x) \\ \omega_{\kappa+1,\alpha}(x) \\ 0 \\ 0
\end{pmatrix}.
\end{multline}

We will then perform a Deift/Zhou steepest descent analysis \cite{deift_book,DKMVZ99} for the RH problem for $Y$. The analysis consists of a series of explicit and invertible transformations
\begin{equation}\label{eq:transformations}
Y \to X \to T \to S \to R,
\end{equation}
which leads to a RH problem for $R$ tending to the identity matrix as $n\to \infty$. Analyzing the effect of the transformations \eqref{eq:transformations} gives us the large $n$ limits of the correlation kernel in various regimes.

The rest of this paper is organized as follows. In Section \ref{sec:anaVEP}, we analyze the equilibrium problem, along the way also extending some classical results from potential theory, introducing a four-sheeted Riemann surface built from the solution to the vector equilibrium problem and describing its uniformization in detail. Theorems \ref{thm_equilibrium_problem} and \ref{thm:spectral_curve} are finally established in Sections~\ref{section:conclusion_proof} and ~\ref{sec:thm_rational_parametrization}, respectively.

Some auxiliary functions, constructed using objects from Section~\ref{sec:anaVEP}, are then introduced in Section \ref{sec:auxifunc} as a preparation for the asymptotic analysis.

Sections \ref{sec:firstrans}--\ref{sec:StoR} are devoted to the steepest descent analysis of the RH Problem \ref{rhp:Y} for $Y$ described above. In particular, we construct a local parametrix near the origin with the aid of the Meijer-G parametrix introduced by Bertola and Bothner in \cite{BB15}, using a recently introduced matching technique by Kuijlaars and Molag \cite{KM19}.

After the RH asymptotic analysis is finished, the conclusion of our main asymptotic results, i.e., Theorems \ref{thm:limitingmeandistri} and \ref{thm:hardedge}, are presented in Section \ref{sec:proofasy}.

We conclude this paper with an Appendix to give some heuristic arguments on how to obtain the precise formulation of the vector equilibrium problem introduced in Section \ref{subsec:vep}, which plays an important role in this paper.

\paragraph{Assumptions and notations}
Throughout this paper, it is assumed that $n$ is an even number so that
$$n_1=n_2=\frac{n}{2}.$$
This assumption is not essential and is only made to simplify the proof.

Since the asymptotic analysis of $4\times 4$ RH problems takes a substantial part of this work, it is notationally convenient to denote by $E_{jk}$ the $4\times 4$ elementary matrix
whose entries are all $0$, except for the $(j,k)$-entry, which is $1$, that is,
\begin{equation}\label{def:Eij}
E_{jk}=\left( \delta_{l,j}\delta_{k,m} \right)_{l,m=1}^4.
\end{equation}
A fact of simple verification that comes in handy is the identity
\begin{equation*}
E_{jk}E_{lm}=\delta_{k,l}E_{jm}.
\end{equation*}

Finally, we adopt the notations
\begin{equation}\label{def:deltai}
\Delta_1=(-\infty,-q),\qquad \Delta_2=(0,p), \qquad \Delta_3=(-\infty, 0),
\end{equation}
i.e., $\Delta_1$ is the interior of $\supp(\sigma-\mu_1)$, $\Delta_2$ is the interior of $\supp \mu_2$, $\Delta_3$ is the interior of $\supp \mu_3$,
and also set
\begin{equation}\label{def:calpha}
\mathfrak c_\alpha=e^{ \alpha \pi i }.
\end{equation}
It is worthwhile to point out that for integer $\alpha$ the symmetry relation
$$\mathfrak c_\alpha=\mathfrak c_{-\alpha}$$ takes place.

\section{Analysis of the vector equilibrium problem}\label{sec:anaVEP}

The goal of this section is to analyze the equilibrium problem associated to the energy functional \eqref{definition_vector_energy}, which will ultimately lead to the proofs of Theorems \ref{thm_equilibrium_problem}, \ref{thm:spectral_curve}, and Proposition \ref{prop:varicondition} about the relevant Euler-Lagrange variational conditions.

\subsection{Preliminaries from potential theory}
In this subsection, we will review some basic concepts and their properties from potential theory, which will be needed in what follows. For more details, we refer to the standard references \cite{Landkofbook,SaffTotik,ransford_book,Tsuji75}.

\paragraph{Logarithmic potential and Cauchy transform of a measure}
Given a measure $\mu$ on $\mathbb{C}$, recall that its Cauchy transform $C^\mu$ was
previously defined in \eqref{def:Cauchytransform}. Closely connected is its logarithmic potential, which is defined by
\begin{equation*}
U^{\mu}(x)=\int \log \frac{1}{|x-y|}\ud \mu(y),\qquad x\in \C,
\end{equation*}
whenever the integral makes sense as a finite real number.

By expanding the integrands into powers of $z$ around infinity, it immediately follows that, as $z\to\infty$,
\begin{equation}\label{eq:asymptotics_cauchy_transf_log_pot}
C^{\mu}(z)=-\frac{|\mu|}{z}(1+\boh(1)), \qquad U^{\mu}(z)=-|\mu|\log |z|(1+\boh(1)).
\end{equation}
If $\mu$ is compactly supported, the terms $\boh(1)$ in \eqref{eq:asymptotics_cauchy_transf_log_pot} can
be replaced by $\Boh(z^{-1})$. Furthermore, these functions are related through
$$
U^{\mu}(z)=\re\int^z C^{\mu}(s)\ud s +c,
$$
where the constant $c$ is chosen so as to have the same asymptotic behavior as $z\to \infty$ on
both sides of the identity above. This last relation implies that
$$
\frac{\partial U^{\mu}}{\partial z}(z)=\frac{1}{2}C^{\mu}(z), \qquad z\in \C\setminus \supp\mu,
$$
where $\frac{\partial }{\partial z}=\frac{1}{2}(\frac{\partial }{\partial x}-i\frac{\partial }{\partial y})$.
This identity also extends to the $\pm$-boundary values on smooth arcs of $\supp\mu$. In this sense,
for a measure $\mu$ on $\R$ with real-differentiable potential, we have
\begin{equation}\label{eq:derivative_potential_boundary_value}
\frac{\ud U^{\mu}}{\ud x}(x)=\re C^{\mu}_{\pm}(x),\qquad x\in \supp\mu,
\end{equation}
so
\begin{equation}\label{eq:derivative_real_potential}
\frac{\ud U^{\mu}}{\ud x}(x)=\int\frac{\ud \mu(s)}{s-x}\left\{
                             \begin{array}{ll}
                               >0, & \hbox{if $x<\inf\supp\mu$,} \\
                               <0, & \hbox{if $x>\sup\supp\mu$.}
                             \end{array}
                           \right.
\end{equation}
In addition, for $z_0\in \supp\mu$ and $\delta>0$ for which $\supp\mu\cap \{  |z-z_0|<\delta\}=\gamma$ is an analytic arc with complex line element $\ud s$,
the Sokhotski-Plemelj relations
\begin{equation}\label{eq:plemelj_relations}
C^{\mu}_+(z)-C^{\mu}_-(z)=2\pi i\frac{\ud \mu}{\ud s}(z),\qquad
C^{\mu}_+(z)+C^{\mu}_-(z)=2\textrm{PV}\int\frac{\ud \mu(x)}{x-z},
\end{equation}
hold for $z\in\gamma$, where $\textrm{PV}$ denotes the Cauchy principal value.

Given a function $\omega(x)$ on $K$, $K=\R_-$ or $K=\R_+$, with
\begin{equation}\label{eq:general_local_behavior}
\omega(x)=c |x|^{a}(1+\boh(1)), \qquad x\to 0 \mbox{ along } K, \quad a>-1,
\end{equation}
its Cauchy transform
$$
C^\omega(z):=C^{\mu_\omega}(z),\qquad \ud \mu_\omega(x):=\omega(x)\ud x
$$
satisfies \cite[Section~29]{Muskhelishvili}
\begin{equation}\label{eq:behavior_cauchy_transform_endpoint}
C^\omega(z)=
\begin{cases}
\Boh(1), & a>0, \\
\Boh(\log z), & a=0, \\
C_Kz^a(1+\boh(1)), & -1<a<0,
\end{cases}
\qquad
\mbox{ as } z\to 0,
\end{equation}
where, for $-1<a<0$, the branch of $z^a$ is chosen so that
$$
\lim_{\delta\to 0+}(x+i \delta)^a=|x|^a, \qquad x\in K,
$$
and
$$
C_K=
\begin{cases}
\dfrac{c \pi e^{a\pi i}}{\sin(a\pi)}, & \textrm{if $K=\R_-$,}    \\
-\dfrac{c \pi e^{-a\pi i}}{\sin(a\pi)}, & \textrm{if $K=\R_+$.}
\end{cases}
$$
Obviously, the behavior near the origin in \eqref{eq:general_local_behavior} and \eqref{eq:behavior_cauchy_transform_endpoint} could be replaced by $x-x_0\to 0$ for any finite point $x_0\in \R$.

\paragraph{Logarithmic capacity}
The {\it logarithmic capacity} $\capac{K}$ of a compact set $K\subset \C$ is defined by
$$
\capac{K}=\sup_{\substack{ |\mu|=1 \\ \supp\mu\subset K }} e^{-I(\mu)}=\exp\left({-\inf_{\substack{ |\mu|=1 \\ \supp\mu\subset K }} I(\mu)}\right),
$$
where we emphasize that the inf/sup is taken over probability measures supported on $K$ and $I(\mu)$ is the logarithmic energy of $\mu$ previously defined in \eqref{eq:log_energy}. In particular, if $\capac{K}=0$, then there is no probability measure on $K$ with finite logarithmic energy.

If $G\subset \C$ is an arbitrary Borel set, its capacity is defined by
$$
\capac G=\sup \{ \capac K \; \mid \; K\subset G,\; K \mbox{ compact} \}.
$$
A property is said to hold {\it quasi-everywhere} (shortly {\it q.e.}), if it holds everywhere except on a set of capacity zero. For a general treatise on capacity and its relation to complex analysis, we refer the reader to \cite{pommerenke_book,ransford_book}.

\paragraph{Balayage measure}
Given a closed set $K\subset \C$ with positive capacity and a finite measure $\mu$ on $\C$,
the {\it balayage measure} of $\mu$ associated with $K$ is the unique measure $\widehat \mu$ such that
$|\mu|=|\widehat \mu|$ and
\begin{equation}\label{eq:balayage_identity_potentials}
U^{\widehat \mu}(z)=U^{\mu}(z)+c,\quad \mbox{q.e. } z\in K,
\end{equation}
where $c$ is a constant.
In particular, if $K$ has an unbounded connected component,
then comparing the behavior of both sides of \eqref{eq:balayage_identity_potentials} as $z\to \infty$
tells us that $c=0$. When needed, we write $$\widehat \mu=\bal (\mu,K)$$ to emphasize the underlying set $K$. A direct relation between the measures $\mu$ and $\bal(\mu,K)$ is given by the formula
\begin{equation}\label{eq:balK}
\bal(\mu,K)=\int \bal(\delta_z,K)\ud \mu(z),
\end{equation}
where $\delta_z$ denotes the Dirac measure at the point $z$.

A choice of our particular interest is
$$K=K_c=(-\infty,-c], \qquad c\geq 0.$$
In this case,
if $z>-c$, then $\bal(\delta_z,K_c)$ is absolutely continuous with respect to the Lebesgue measure, and
\begin{equation}\label{eq:baldelta}
\frac{\ud \bal(\delta_z,K_c)}{\ud x}(x)=\frac{1}{\pi}\frac{\sqrt{z+c}}{\sqrt{|x+c|}(z-x)},\qquad x\in K_c.
\end{equation}

For a measure $\mu$ with $\supp\mu\subset [-c,+\infty)$, for simplicity we denote
$$\widehat \mu_c=\bal(\mu,K_c).$$
Assuming that $\mu(\{-c\})=0$, it is easily seen from \eqref{eq:balK}
and \eqref{eq:baldelta} that
\begin{equation}\label{eq:balayage_formula}
\frac{\ud \widehat \mu_c}{\ud x}(x)=\frac{1}{\pi\sqrt{|x+c|}}\int\frac{\sqrt{z+c}}{z-x} \ud \mu(z),\qquad x\in K_c.
\end{equation}

As an application of \eqref{eq:balayage_formula}, we have the following two simple lemmas which will be essential in establishing the characterization of $\supp\mu_1$ given by Theorem~\ref{thm_equilibrium_problem}.

\begin{lem}\label{lem:density_balayage}
If $\mu$ is a finite measure on $[-c,+\infty)$ with $\mu(\{-c\})=0$, then the function
$$
x\mapsto \sqrt{|x|} \frac{\ud \widehat \mu_c}{\ud x}(x)
$$
is increasing on $K_c$.
\end{lem}
\begin{proof}
By \eqref{eq:balayage_formula}, it follows that
\begin{equation*}
\sqrt{|x|} \frac{\ud \widehat \mu_c}{\ud x}(x)=\frac1\pi\sqrt{\frac{|x|}{|x+c|}}\int\frac{\sqrt{z+c}}{z-x} \ud \mu(z),\qquad x\in K_c.
\end{equation*}
Since both $\sqrt{\frac{|x|}{|x+c|}}$ and the integrand on the right-hand side of the above formula are increasing functions of $x$ on $K_c$, the lemma follows immediately.
\end{proof}

With the measure $\sigma$ introduced in \eqref{def:constraint_measure}, we have

\begin{lem}\label{lem:support_saturation}
If $\mu$ is an absolutely continuous finite measure on $K_c$ for which $\sqrt{|x|}\frac{\ud \mu}{\ud x}(x)$ is increasing on $K_c$, then the positive part $(\mu-\sigma)^+$ of the signed measure $\mu-\sigma$ is either zero or satisfies $$\supp((\mu-\sigma)^+)=[-\widetilde c,-c],$$
for some $\widetilde c>c$.
\end{lem}
\begin{proof}
Because $\mu$ is finite but $\sigma$ is not, we are sure that $\frac{\ud \mu}{\ud x}-\frac{\ud \sigma}{\ud x}$ is negative for $x$ large.
By \eqref{def:constraint_measure}, we can write
$$
\frac{\frac{\ud \mu}{\ud x}(x)}{\frac{\ud\sigma}{\ud x}(x)}=
 \frac{\pi}{\alpha}
\left( \sqrt{|x|} \; \frac{\ud \mu}{\ud x}(x)\right).
$$
Thus, the previous Lemma tells us that the quotient on the left-hand side above is strictly increasing, so there exists at most one point in which this quotient changes from smaller to bigger than $1$. That is, there is at most one point for which the difference $\frac{\ud \mu}{\ud x}-\frac{\ud \sigma}{\ud x}$ changes from negative to positive, and the result follows.
\end{proof}

\subsection{An extension of the Lower Envelope Theorem}
In this subsection, we will extend the so-called Lower Envelope Theorem. The results presented here are well-known under the stronger assumption that the underlying measures are supported in a fixed compact set of $\C$, but later we will need these results for measures with unbounded support.

\begin{prop}\label{prop:lower_envelope}
Let $\{\mu_n\}$ be a sequence of probability measures on $\C$ that converges weakly to a probability measure $\mu$ on $\C$
and satisfies the following conditions:
\begin{enumerate}[$($$i$$)$]
\item The quantities
$$
\int \log(1+|z|^2)\ud \mu_n(z)
$$
are finite and uniformly bounded in $n$.
\item As $R\to \infty$, the quantities
$$
\int_{|z|\geq R} \log(1+|z|^2)\ud \mu_n(z)
$$
converge to zero uniformly in $n$.
\end{enumerate}

Then, we have
$$
U^{\mu}(z)\leq \liminf_{n\to\infty}  U^{\mu_n}(z), \qquad z\in \C,
$$
and
$$
\liminf_{n\to\infty}  U^{\mu_n}(z)=U^{\mu}(z)
$$
for quasi-every $z \in \C$.
\end{prop}
\begin{proof}
We follow an idea in \cite{hardy_kuijlaars} and map the Riemann sphere $\overline \C$ to
the sphere $\mathcal S\subset \R^3$ centered at $(0,0,1/2)$ with radius $1/2$ through the stereographic projection
$$
T(z)=
\begin{cases}
\left( \frac{\re(z)}{1+|z|^2},\frac{\im(z)}{1+|z|^2},\frac{|z|^2}{1+|z|^2} \right),& \quad z\in \C, \\
(0,0,1), & \quad  z=\infty.
\end{cases}
$$
It is straightforward to check that the mapping $T$ satisfies
\begin{equation}\label{eq:distance_spherical}
\|T(z)-T(w)\|=\frac{|z-w|}{\sqrt{1+|z|^2}\sqrt{1+|w|^2}},\qquad z,w\in \C,
\end{equation}
where $\|\cdot\|$ stands for the standard Euclidean norm in $\R^3$.

For a measure $\nu$ on $\overline\C$, denote by $\nu^T$ its pushforward measure induced by $T$. That is, $\nu^T$ is a measure on $\mathcal S$ determined by the condition that
\begin{equation}\label{def:pushforward}
\int f(x) \ud \nu^T(x) = \int f(T(z))\ud \nu(z).
\end{equation}
With
$$
V^{\lambda}(x)=\int \log\frac{1}{\|x-y\|}\ud \lambda(y)
$$
denoted by the potential of a measure $\lambda$ on $\mathcal S$,
it follows from \eqref{eq:distance_spherical} that if a measure $\nu$ on $\C$ satisfies
$$
\int \log(1+|z|^2)\ud\nu(z)<\infty
$$
and $U^{\nu}$ is finite at $z$, then $V^{\nu^T}$ is finite at $T(z)$ and
\begin{equation}\label{eq:Vpushforward}
V^{\nu^T}(T(z))=V^{\nu}(z)+\frac{1}{2}\int \log(1+|w|^2)\ud\nu(w)+\frac{1}{2}|\nu|\log(1+|z|^2).
\end{equation}

Since $\mu_n\stackrel{*}{\to} \mu$, it then follows from \eqref{def:pushforward} that $\mu^T_n\stackrel{*}{\to} \mu^T$. Thus, by
replacing the measure $\nu$ in \eqref{eq:Vpushforward} by $\mu_n$, it is readily seen that our proposition follows
if we can show that
\begin{enumerate}[(a)]
\item The following limit holds:
\begin{equation}\label{eq:logconver}
 \lim_{n\to\infty} \int \log(1+|w|^2)\ud\mu_n(w) = \int \log(1+|w|^2)\ud \mu(w).
 \end{equation}
\item Whenever $\{\nu_n \}$ is a sequence of probability measures on $\mathcal S$ converging weakly to $\nu$, then
\begin{equation}\label{eq:Vnules}
V^{\nu}(x)\leq \liminf_{n\to\infty}  V^{\nu_n}(x),
\end{equation}
for every $x\in \mathcal S$
\item For the measures $\{\nu_n\}$ and $\nu$ as in (b), there exists a log-polar set $E \subset \mathcal S$ such that
\begin{equation}\label{eq:Vnuequal}
V^{\nu}(x)= \liminf_{n\to\infty}  V^{\nu_n}(x),
\end{equation}
for $x\in \mathcal S\setminus E$.
\end{enumerate}

In part (c), by a log-polar set we mean that
$$
\int V^{\nu}(x)\ud \nu(x)=+\infty,
$$
for any probability measure $\nu$ supported on $E$. We note that $E$ is log-polar if, and only if, $T^{-1}(E)$ has zero capacity in $\C$.

The proof of \eqref{eq:Vnules} follows immediately from the weak convergence of $\nu_n$ to $\nu$,
and the fact that the function
$$
y\mapsto \log\frac{1}{\| x-y \|}
$$
is lower semi-continuous on any compact subset of $\R^3$, whereas the proof of \eqref{eq:Vnuequal} follows in the same steps as its analogue for measures supported in a fixed compact set of the plane \cite[Theorem I.6.9]{SaffTotik}.

We finally provide a proof of \eqref{eq:logconver}. Since the non-negative function $\log(1+|z|^2)$ is lower semi-continuous on $\C$,
the weak convergence $\mu_n\stackrel{*}{\to}\mu$ immediately implies that
\begin{equation}\label{eq:inequality_normalized_log}
0\leq \int \log(1+|z|^2)\ud\mu(z)\leq \liminf_{n\to \infty } \int \log(1+|z|^2) \ud \mu_n(z).
\end{equation}
By the condition (i), the right-hand side of the above inequality is finite.

Let $\{\lambda_n\}$ be a sequence of probability measures on $\C$ defined by
$$
\ud \lambda_n(z)=\frac{1}{c_n}\log(1+|z|^2) \ud \mu_n(z),
$$
where
$$
c_n=\int \log(1+|z|^2)\ud \mu_n(z).
$$
If $\limsup_{n \to \infty} c_n=0$, then the proof is over. Hence, we may assume that, without loss of generality,
$$
c_n\to \limsup_{n \to \infty} c_n:=c>0.
$$
From the condition (i), the limsup above is finite and thus $\{\lambda_n\}$ is
a well-defined sequence of probability measures on $\C$. Furthermore, from the condition (ii),
we see that this sequence is tight. By Prohorov's theorem, we can assume that, after extracting a subsequence, it
converges weakly to a probability measure $\lambda$ on $\C$. Thus, if $f$ is any bounded continuous function on $\C$, we have
$$
\lim_{n\to\infty} \int f(z)\ud \lambda_n(z) = \int f(z) \ud \lambda(z).
$$
If, in addition, the function $f$ has compact support, then function $f(z)\log(1+|z|^2)$ is continuous and bounded on $\C$.
The weak convergence $\mu_n\stackrel{*}{\to} \mu$ then implies that
$$
\int f(z)\ud \lambda_n(z)=\frac{1}{c_n}\int f(z)\log(1+|z|^2)\ud \mu_n(z)\stackrel{n\to\infty}{\to} \frac{1}{c}\int f(z)\log(1+|z|^2)\ud \mu(z),
$$
and consequently
$$
\int f(z)\log(1+|z|^2)\ud \mu(z) = c\int f(z) \ud\lambda(z),
$$
for every compactly supported continuous function $f$. Considering a sequence $\{f_m\}$ of
such functions with the extra conditions that $f_m\geq 0$ and $f_m\nearrow 1$ pointwise, it follows from
the Monotone Convergence Theorem that
$$
\int \log(1+|z|^2)\ud \mu(z)=c \int \ud \lambda(z)=c = \limsup_{n\to \infty} \int \log(1+|z|^2)\ud \mu_n(z).
$$
This, together with \eqref{eq:inequality_normalized_log}, gives us \eqref{eq:logconver}.

This completes the proof of Proposition \ref{prop:lower_envelope}.
\end{proof}

\subsection{A scalar constrained equilibrium problem}

Let $\rho$ be a probability measure on $\R_+$. The so-called $\sigma$-constrained equilibrium measure $\mu_\rho^\sigma$
of $\R_-$, if it exists, is the measure that minimizes the functional
$$
I(\mu)-2\int U^\rho(z) \ud \mu(z)
$$
over all probability measures $\mu$ on $\R_-$ subject to the condition $\mu\leq \sigma$, where $\sigma$ is a given measure on
$\R_-$.

The characterization of the measure $\mu_1$ in Theorem \ref{thm_equilibrium_problem} that we are looking for will follow from the following proposition.

\begin{prop}\label{prop:support_constrained_equilibrium_measure}
With the measure $\sigma$ given in \eqref{def:constraint_measure}, the $\sigma$-constrained
measure $\mu_\rho^\sigma$ exists uniquely. Furthermore, there exists a constant $c \geq 0$ such that
\begin{equation}\label{eq:suppkc}
\supp(\sigma-\mu_\rho^\sigma)=K_c=(-\infty,-c],
\end{equation}
and the following Euler-Lagrange variational conditions hold:
\begin{align}
U^{\mu_\rho^\sigma}(z)-U^{\rho}(z)=0,&\qquad z\in \supp(\sigma-\mu_\rho^\sigma), \label{eq:euler_lagrange_constrained_scalar}
\\
U^{\mu_\rho^\sigma}(z)-U^{\rho}(z)\leq 0, & \qquad z\in \R_- \label{eq:euler_lagrange_constrained_scalar2}.
\end{align}
\end{prop}
\begin{proof}
Existence, uniqueness and characterization through the variational conditions of the minimizer, with possibly a nonzero constant $\ell$ on the right-hand side of \eqref{eq:euler_lagrange_constrained_scalar} and \eqref{eq:euler_lagrange_constrained_scalar2}, follow from the standard theory, we refer the reader to \cite{Dragnev_Saff_1997} for details. To see that $\ell=0$ is the correct constant, we first observe that $\sigma((-\infty,-a])=+\infty$ for any $a>0$. This, together with the fact that $\mu_\rho^\sigma$ is a probability measure, implies that $\supp(\sigma-\mu_\rho^\sigma)$ is unbounded. Thus, we can take the limit $z\to -\infty$ in \eqref{eq:euler_lagrange_constrained_scalar} and use the behavior of $U^{\mu_\rho^\sigma}(z)-U^{\rho}(z)$ near $\infty$ (see \eqref{eq:asymptotics_cauchy_transf_log_pot}) to conclude that $\ell=0$.

To show \eqref{eq:suppkc}, we follow the ideas similar to the ones in \cite{Dragnev_Kuijlaars_1999,DK09,duits_kuijlaars_mo}, which are based on
the iterative balayage algorithm introduced by Dragnev \cite{Dragnev_thesis}. To proceed, we set
\begin{equation}\label{def:nu1}
\nu_1=\bal (\rho,\R_-).
\end{equation}
The measure $\nu_1$ then has the following properties:
\begin{enumerate}[(a)]
\item For $z\in\R_-$, we have
      \begin{equation*}
      U^{\nu_1}(z)=U^{\rho}(z),
      \end{equation*}
      that is, $\nu_1$ is the unconstrained equilibrium measure of $\R_-$ with the external field $-2U^{\rho}$. This property follows from the definition of the balayage measure.
\item From Lemmas \ref{lem:density_balayage} and \ref{lem:support_saturation}, we have that $$\supp((\nu_1-\sigma)^+)=[-c_1,0],$$
    for some $c_1\geq 0$.
\item For $c_1$ as above, we have
$$\restr{\mu_\rho^\sigma}{[-c_1,0]}=\restr{\sigma}{[-c_1,0]}.$$
This follows from property (a) and the Saturation Principle \cite[Theorem~2.6]{Dragnev_Saff_1997}.
\end{enumerate}

We now define inductively
\begin{equation}\label{eq:recursion_nu_k}
\nu_{k+1}=\restr{\nu_{k}}{K_{c_{k}}}+\restr{\sigma}{[-c_{k},0]}+\widetilde \nu_k,\qquad k\geq 1,
\end{equation}
with
\begin{equation}\label{eq:definition_tilde_nu}
\widetilde \nu_k=\bal((\nu_{k}-\sigma)^+,K_{c_{k}}).
\end{equation}
In \eqref{eq:definition_tilde_nu}, if $k \geq 2$, the constant $c_{k} \geq 0$ is, as we will show in a moment,
uniquely defined through the condition
\begin{equation}\label{eq:definition_endpoint_iterative_balayage}
\begin{cases}
c_{k}=c_{k-1},& \mbox{if } (\nu_{k}-\sigma)^+=0,\\
\supp((\nu_{k}-\sigma)^+)=[-c_k,-c_{k-1}],& \mbox{if }(\nu_{k}-\sigma)^+\neq 0.
\end{cases}
\end{equation}
In words, we swap out the part of $\nu_{k}$ that saturates $\sigma$ to the set $K_{c_k}$.
From \eqref{eq:recursion_nu_k}--\eqref{eq:definition_endpoint_iterative_balayage}, we also observe that
\begin{equation}\label{eq:identity_sigma_nu_k}
\restr{\nu_{k}}{[-c_{k-1},0]}=\restr{\sigma}{[-c_{k-1},0]}
\end{equation}
and that $\nu_k$ has no mass points. This particularly implies that
$$
\begin{aligned}
|\nu_{k+1}| & =\nu_k(K_{c_k})+\sigma([-c_k,0])+|(\nu_k-\sigma)^+| \\
            & = \nu_k(K_{c_k})+\sigma([-c_k,-c_{k-1}])+\sigma([-c_{k-1},0])+\nu_k([-c_k,-c_{k-1}])-\sigma([-c_k,-c_{k-1}]) \\
			& = \nu_k(K_{c_k})+\sigma([-c_k,-c_{k-1}])+\nu_{k}([-c_{k-1},0])+\nu_k([-c_k,-c_{k-1}])-\sigma([-c_k,-c_{k-1}]) \\
			& = |\nu_k|,
\end{aligned}
$$
and because $|\nu_1|=1$ we get that $|\nu_k|=1$ for every $k$.

To see that \eqref{eq:definition_endpoint_iterative_balayage} indeed uniquely defines $c_k$, we will proceed inductively. We start with the observation that the function
\begin{equation}\label{eq:density_transform}
K_{c_k}\ni x\mapsto \sqrt{|x|}\frac{\ud \nu_{k+1}}{\ud x}(x)
\end{equation}
is increasing, once $\nu_1,\hdots,\nu_{k+1}$ are all well defined. In fact, because $\nu_1$ is absolutely continuous and $\nu_{k+1}$ is obtained from $\nu_k$ and $\sigma$ by sums, balayages and restrictions,
which are operations that preserve the absolutely continuity, it follows that $\nu_{k+1}$ is always absolutely continuous.
By \eqref{eq:recursion_nu_k}, it is readily seen that for $x\in K_{c_k}$,
$$
\sqrt{|x|}\frac{\ud \nu_{k+1}}{\ud x}(x)= \sqrt{|x|}\frac{\ud \nu_{k}}{\ud x}(x) +
\sqrt{|x|}\frac{\ud \widetilde \nu_k}{\ud x}(x).
$$
Because of \eqref{eq:definition_tilde_nu}, Lemma \ref{lem:density_balayage} tells us that the second term in the sum on the right-hand side above is increasing. Under induction hypothesis for \eqref{eq:density_transform}, the first term in this sum is increasing as well. Hence, by induction it follows that \eqref{eq:density_transform} is always increasing.

Thus, once we know that $c_k$ as in \eqref{eq:definition_endpoint_iterative_balayage} exists, the corresponding measure $\nu_{k+1}$ in \eqref{eq:recursion_nu_k} is well defined. Since $\sqrt{|x|}\frac{\ud \nu_{k+1}}{\ud x}(x)$ is increasing on $K_{c_k}$, we conclude \eqref{eq:definition_endpoint_iterative_balayage} for $k+1$ with the aid of Lemma \ref{lem:support_saturation}, showing that the recursions \eqref{eq:recursion_nu_k}--\eqref{eq:definition_endpoint_iterative_balayage} are well defined.

We also remark that, for $x\in K_{c_k}$,
\begin{align}
U^{\nu_{k+1}}(x) & =U^{\restr{\nu_{k}}{K_{c_{k}}}}(x) + U^{\restr{\sigma}{[-c_k,0]}}(x)+U^{\widetilde \nu_{k}}(x)
\nonumber
\\
			 & = U^{\restr{\nu_{k}}{K_{c_k}}}(x) + U^{\restr{\sigma}{[-c_k,0]}}(x)+U^{\restr{\nu_k}{[-c_k,-c_{k-1}]}}(x) - U^{\restr{\sigma}{[-c_k,-c_{k-1}]}}(x)
\nonumber
\\
			 & =U^{\restr{\nu_{k}}{K_{c_{k}}}}(x)+U^{\restr{\sigma}{[-c_{k-1},0]}}(x)+U^{\restr{\nu_k}{[-c_k,-c_{k-1}]}}(x)
\nonumber
\\
			 & = U^{\nu_{k}}(x),
\label{eq:iterative_potential}
\end{align}
where the first equality simply follows from the definition of $\nu_k$ in \eqref{eq:recursion_nu_k},
the second equality is a consequence of the definition \eqref{eq:definition_tilde_nu} of $\widetilde \nu_k$ as a balayage
measure and the assumption that $x\in K_{c_k}$, and for the final equality we have made use of
\eqref{eq:identity_sigma_nu_k}. Furthermore, from the Principle of Domination \cite{SaffTotik}, we also know that
$$
U^{\widetilde \nu_k}(x) \leq U^{\restr{\nu_k}{[-c_k,-c_{k-1}]}}(x) - U^{\restr{\sigma}{[-c_k,-c_{k-1}]}}(x),\qquad x\in \C.
$$
Thus, by performing similar calculations as in \eqref{eq:iterative_potential} but
replacing the second equality by an inequality, we conclude that
\begin{equation}\label{eq:iterative_potential_2}
U^{\nu_{k+1}}(x) \leq U^{\nu_{k}}(x),\qquad x \in \C\setminus K_{c_k}.
\end{equation}

We claim that the sequence $\{c_k\}$ is convergent. Indeed,
from its construction, it is readily seen that $c_k\geq c_{k-1}$,
so this sequence is increasing. It is also bounded, because by \eqref{eq:identity_sigma_nu_k}, we have
\begin{equation*}
\sigma([-c_k,0])=\nu_{k}([-c_k,0])\leq 1,
\end{equation*}
but $\sigma([x,0])\to +\infty$ when $x\to -\infty$. Hence,
\begin{equation}\label{eq:cklimit}
\lim_{k\to \infty}c_k = c
\end{equation}
for some $c\geq 0$.

Our next goal is to show that the measures $\{\nu_k\}$ has a weakly convergent subsequence. 
To see this, we observe from \eqref{eq:balayage_formula}, \eqref{def:nu1} and \eqref{eq:recursion_nu_k} that
\begin{equation}\label{eq:nukbound}
\frac{\ud \nu_k}{\ud x}(x)=\Boh(|x|^{-3/2}),\qquad x\to -\infty,
\end{equation}
where the bound is uniform in $k$. Thus, given any $\varepsilon>0$, we can find $M=M(\varepsilon)>c$ such that
\begin{equation*}
\nu_k((-\infty,-M]) < \varepsilon.
\end{equation*}
Since $\varepsilon>0$ is arbitrary, this shows that $\{\nu_k\}$ is a tight sequence of probability measures.
By Prokhorov's theorem, there is a subsequence $\{\nu_{k_j}\}$ converging weakly to a probability measure $\nu$ on $\R_-$.

Let $G \subset \R_-$ be any bounded open subset. From the weak convergence, we have
\begin{equation*}
\nu(G)\leq \liminf_{j\to \infty}\nu_{k_j}(G).
\end{equation*}
If $G \subset [-c,0]$, by \eqref{eq:identity_sigma_nu_k}, it is easily seen that $\nu(G)=\sigma(G)$. If $G\subset K_c$, note that $\frac{\ud \nu_k/\ud x}{\ud \sigma /\ud x}$ is strictly increasing on $K_{c_k}$, it then follows from \eqref{eq:definition_endpoint_iterative_balayage} and \eqref{eq:cklimit} that
$$\nu(G)<\sigma(G).$$

Moreover, the bound \eqref{eq:nukbound} implies that the requirements, and thus, the conclusions of Proposition \ref{prop:lower_envelope} are applicable to the sequence $\{\nu_{k_j}\}$. This, together with \eqref{eq:iterative_potential}, tells us that
$$
U^{\nu}(c)\leq \liminf_{j\to\infty} U^{\nu_{k_j}}(c)= U^{\nu_1}(c)=U^{\rho}(c)<+\infty,
$$
hence, $\nu$ cannot have a point mass at $z=c$. A combination of all these results then shows that
$$\nu \leq \sigma, \qquad \textrm {on $\R_-$},$$
and
$$\supp(\sigma-\nu)=K_c.$$

Finally, using Proposition~\ref{prop:lower_envelope} and equations \eqref{eq:iterative_potential}--\eqref{eq:iterative_potential_2}, we have that $\nu$ also satisfies the two conditions in \eqref{eq:euler_lagrange_constrained_scalar}. Hence, by uniqueness of the minimizer, it follows that $\nu=\mu_\rho^{\sigma}$ and $\supp(\sigma-\mu_\rho^\sigma)=\supp(\sigma-\nu)=K_c$.

This completes the proof of Proposition \ref{prop:support_constrained_equilibrium_measure}.
\end{proof}

\subsection{Qualitative properties for the vector equilibrium measure}\label{sec:qualitative}

To obtain qualitative properties for the vector of measures $\pmb{\mu}=(\mu_1,\mu_2,\mu_3)\in\mathcal M$ that minimizes
\eqref{definition_vector_energy}, we recall the Euler-Lagrange conditions of the problem, which here take the form of the following set of equalities and inequalities:
\begin{align}
& 2U^{\mu_1}(x)-U^{\mu_2}(x)=\ell_1 , && \mbox{q.e. } x\in \supp(\sigma-\mu_1),
\label{eq:vector_equilibrium_conditions1}\\
& 2U^{\mu_1}(x)-U^{\mu_2}(x)\leq \ell_1 , && x\in  \R_-\setminus\supp(\sigma-\mu_1), \label{eq:vector_equilibrium_conditions2}\\
& 2U^{\mu_2}(x)-U^{\mu_1}(x)-U^{\mu_3}(x)+2(\beta-\alpha)\sqrt{x}=\ell_2 , && \mbox{q.e. }x\in \supp\mu_2,
\label{eq:vector_equilibrium_conditions3}\\
& 2U^{\mu_2}(x)-U^{\mu_1}(x)-U^{\mu_3}(x)+2(\beta-\alpha)\sqrt{x}\geq \ell_2 , && x\in \R_+\setminus \supp\mu_2,
\label{eq:vector_equilibrium_conditions4}\\
& 2U^{\mu_3}(x)-U^{\mu_2}(x)=\ell_3 , && \mbox{q.e. }x\in \supp\mu_3,
\label{eq:vector_equilibrium_conditions5}\\
& 2U^{\mu_3}(x)-U^{\mu_2}(x)\geq \ell_3 , && x\in \R_-\setminus \supp\mu_3,
\label{eq:vector_equilibrium_conditions6}
\end{align}
where $\ell_1,\ell_2$ and $\ell_3$ are three constants. These equations actually follow from the Euler-Lagrange conditions for the usual equilibrium problem for scalar measures. For instance, the equilibrium problem for $\nu_2$ is to minimize, with fixed $\nu_1$ and $\nu_3$ satisfying conditions (E1), (E2) and (E4), the functional
\begin{equation*}
I(\nu)+\int Q_2(x)\ud \nu(x)
\end{equation*}
among all the probability measures $\nu$ on $\R_+$, where
\begin{equation}\label{eq:external_field_mu2}
Q_2(x)=-U^{\mu_1}(x)-U^{\mu_3}(x)+2(\beta-\alpha)\sqrt{x},\qquad x\in \R_+,
\end{equation}
is interpreted as the external field. Thus, given $\mu_1$ and $\mu_3$, the component $\mu_2$ is characterized by equations \eqref{eq:vector_equilibrium_conditions3} and \eqref{eq:vector_equilibrium_conditions4}. The other variational conditions can be derived in a similar manner.

In our setup, \eqref{eq:vector_equilibrium_conditions1}--\eqref{eq:vector_equilibrium_conditions6} are improved with the next result, which also provides some of the statements claimed in Theorem~\ref{thm_equilibrium_problem}.
\begin{prop}\label{prop:varicondition}
There exists a unique minimizer $\pmb{\mu}=(\mu_1,\mu_2,\mu_3)\in \mathcal M$ of the energy functional \eqref{definition_vector_energy} over $\mathcal M$ stated in Theorem \ref{thm_equilibrium_problem}.  Moreover, $\supp\mu_1=\R_-$, $\supp\mu_3=\R_+$, and for some positive numbers $p$ and $q$,
\begin{equation}\label{eq:supports_mu1_mu2}
\supp(\sigma-\mu_1)=(-\infty,-q],\qquad \supp\mu_2=[0,p].
\end{equation}
In addition, the three measures $\mu_1,\mu_2$ and $\mu_3$ are absolutely continuous with respect to the Lebesgue measure, and their densities are bounded except possibly at the origin.

Furthermore, there exists a constant $\ell \in \mathbb{R}$ such that
\begin{align}
& 2U^{\mu_2}(x)-U^{\mu_1}(x)-U^{\mu_3}(x)+2(\beta-\alpha)\sqrt{x}=\ell , &&  x\in \supp\mu_2=[0,p], \label{eq:variation123}\\
& 2U^{\mu_2}(x)-U^{\mu_1}(x)-U^{\mu_3}(x)+2(\beta-\alpha)\sqrt{x} > \ell , &&  x\in (p,+\infty). \label{eq:variation123ineq}
\end{align}
Finally, we have
\begin{align}
2U^{\mu_1}(x)-U^{\mu_2}(x)& =0, \qquad x\in \supp(\sigma-\mu_1)=(-\infty,-q],\label{eq:variation12}\\
2U^{\mu_1}(x)-U^{\mu_2}(x)&< 0,  \qquad  x\in  (-q,0], \label{eq:variation23}
\end{align}
and
\begin{equation}\label{eq:variation232}
2U^{\mu_3}(x)-U^{\mu_2}(x)=0 ,  \qquad x\in \supp\mu_3=\R_-.
\end{equation}
\end{prop}

\begin{proof}
The existence and uniqueness of the minimizer $\pmb{\mu}=(\mu_1,\mu_2,\mu_3)$ claimed by Theorem \ref{thm_equilibrium_problem} follows from the standard theory, we refer the reader to \cite{hardy_kuijlaars} for details, and also \cite{DKRZ12,DK09,duits_kuijlaars_mo} where similar equilibrium problems appeared.

Also, observe that once two among the measures $\mu_1$, $\mu_2$ and $\mu_3$ are fixed, the total potential acting on the third measure is real analytic on the set supporting it, except possibly at the origin. This immediately implies that the three measures are absolutely continuous, and also that their densities are bounded except possibly at the origin. For the same reason, the q.e. conditions on \eqref{eq:vector_equilibrium_conditions1}, \eqref{eq:vector_equilibrium_conditions3} and \eqref{eq:vector_equilibrium_conditions5} are actually valid everywhere on the corresponding supports, so yielding \eqref{eq:variation123}, \eqref{eq:variation12} and \eqref{eq:variation232}, where for the latter two the fact that the variational constants $\ell_1$ and $\ell_3$ are zero will follow from the unboundedness of the supports of $\sigma-\mu_1$ and $\mu_3$, to be shown in a moment.

We first show the properties of $\mu_1$. Since $\sigma((-\infty,-x))=+\infty$ for any $x\geq 0$ and $\mu_1$ is finite, we get that
$\supp(\sigma-\mu_1)$ is unbounded. In addition, it is readily seen from  \eqref{eq:vector_equilibrium_conditions1} and \eqref{eq:vector_equilibrium_conditions2} that
\begin{equation*}
\mu_1=\frac{1}{2}\mu_\rho^{\sigma} \qquad \textrm{with} \qquad \rho=\mu_2,
\end{equation*}
where the measure $\mu_\rho^{\sigma}$ is defined in Proposition \ref{prop:support_constrained_equilibrium_measure}. Hence, it follows that $\supp(\sigma-\mu_1)=(-\infty,-q]$ for some $q\geq 0$, as well as \eqref{eq:variation23} with possibly weak inequality and also \eqref{eq:variation232}. To see that the inequality is indeed strict, we start with the functions $\xi_1$ and $\xi_2$ in \eqref{eq:definition_xi_functions} that are at this point already defined off the real axis, and compute from \eqref{eq:derivative_potential_boundary_value} that
\begin{equation}\label{eq:umu12}
\re \int_{-q}^x(\xi_{1,+}(s)-\xi_{2,+}(s))\ud s =2U^{\mu_1}(x)-U^{\mu_2}(x),\qquad x\in (-q,0).
\end{equation}
This, together with Remark \ref{remark:discriminant_zeros} below\footnote{In fact, in \eqref{eq:umu12} the $+$-boundary value can be omitted. Our proof that the left-hand side of \eqref{eq:umu12} does not vanish relies on the {\it equalities} \eqref{eq:variation123}, \eqref{eq:variation12} and \eqref{eq:variation232} in an implicit manner, but obviously not on their corresponding inequalities.}, implies that the inequality \eqref{eq:vector_equilibrium_conditions2} is strict. All the conditions on $\supp\mu_1$ are thus proven.

Next we handle the conditions on $\supp\mu_2$. Observe that for $j=1,3,$ and $x>0$,
\begin{align*}
\frac{\ud}{\ud x}(x(U^{\mu_j})'(x)) & = \int \frac{x}{(s-x)^2}\ud \mu_j(s)+\int \frac{1}{s-x}\ud\mu_j(s) \\
& = \int \left( \frac{x}{(s-x)^2}+\frac{s-x}{(s-x)^2} \right) \ud\mu_j(s) = \int \frac{s}{(s-x)^2}\ud\mu_j(s)<0,
\end{align*}
where we have made use of the fact that $\mu_j$ is a positive measure supported on $\R_-$.
Furthermore, a simple calculation also shows that
$$(x(\sqrt{x})')'>0, \qquad   x>0.$$ Hence, on account of \eqref{eq:external_field_mu2}, \eqref{eq:inequality_alpha_beta} and the above two inequalities,
we conclude that $$(xQ_2'(x))'>0, \qquad x>0.$$
By \cite[Theorem~IV.1.10 - (c)]{SaffTotik}, this implies that
$$\supp\mu_2=[\tilde p,p]$$ for some $p>\tilde p\geq 0$ and also that the inequality \eqref{eq:vector_equilibrium_conditions4} is strict. To see that $\tilde p=0$, we note that the equality \eqref{eq:variation123}, already proven, gives us that
\begin{equation}\label{eq:equalitytildep}
2U^{\mu_2}(\tilde p)+Q_2(\tilde p)=\ell.
\end{equation}
If $\tilde p>0$, it then follows from \eqref{eq:derivative_real_potential} that
the function $Q_2+2U^{\mu_2}$ is strictly increasing on $(0,\tilde p)$. This,
together with \eqref{eq:equalitytildep}, implies that
$$
2U^{\mu_2}(x)+Q_2(x)<\ell,\qquad x\in (0,\tilde p),
$$
contradicting the inequality \eqref{eq:vector_equilibrium_conditions4}. Hence,
we have to have that $\tilde p=0,$, which concludes \eqref{eq:supports_mu1_mu2}.

As for $\mu_3$, it is a consequence of \eqref{eq:balayage_identity_potentials} that the measure $\frac{1}{2}\bal(\mu_2,\R_-)$ is fully supported on $\R_-$ and satisfies the equality  \eqref{eq:variation232} everywhere on its support. Hence, we must have
$$\mu_3=\frac{1}{2}\bal(\mu_2,\R_-),\qquad \ell_3=0,$$
and \eqref{eq:density_mu3} follows immediately from \eqref{eq:balayage_formula}.

This completes the proof of Proposition \ref{prop:varicondition}.
\end{proof}

The arguments above give us the qualitative properties claimed by Theorem~\ref{thm_equilibrium_problem}. The proofs of the quantitative claims of Theorem~\ref{thm_equilibrium_problem}, namely formulas \eqref{eq:formula_q_endpoint}, \eqref{eq:local_behavior_mu1}, \eqref{eq:formula_p_endpoint}, \eqref{eq:local_behavior_mu2} and \eqref{eq:local_behavior_mu3}, will be given in Section \ref{section:conclusion_proof} below.

\subsection{A four-sheeted Riemann surface \texorpdfstring{$\mathcal R$}{R}}
To prove Theorem \ref{thm:spectral_curve}, we need a Riemann surface consisting of
four sheets $\mathcal R_j$, $j=1,2,3,4$, given by
\begin{equation}\label{def:sheets}
\begin{aligned}
\mathcal R_1 &=\C\setminus (-\infty,-q],  && \qquad
\mathcal R_2 = \C\setminus \left((-\infty,-q]\cup [0,p]\right),
\\
\mathcal R_3 &= \C\setminus (-\infty,p],
&& \qquad
\mathcal R_4 = \C\setminus (-\infty,0],
\end{aligned}
\end{equation}
where the constants $p,q$ are given in \eqref{eq:formula_p_endpoint} and \eqref{eq:formula_q_endpoint}, respectively.

The sheet $\mathcal R_1$ is connected to the sheet $\mathcal R_2$ through $(-\infty,-q]$, $\mathcal R_2$ is connected to $\mathcal R_3$ through $[0,p]$ and $\mathcal R_3$ is connected to $\mathcal R_4$ through $(-\infty,0]$. All these gluings are performed in the usual crosswise manner; see Figure \ref{fig: Riemann surface}. We then compactify the resulting surface by adding a common point at $\infty$ to the sheets $\mathcal R_1$ and $\mathcal R_2$, and a common point at $\infty$ to the sheets $\mathcal R_3$ and $\mathcal R_4$. We denote this compact Riemann surface by $\mathcal R$.

\begin{figure}[t]
\begin{center}
\begin{tikzpicture}

%
%
\draw[thick] (-0.5,4.5)--(4,4.5)--(5,6);
\draw[thick] (-0.5,3)--(4,3)--(5,4.5);
\draw[thick] (-0.5,1.5)--(4,1.5)--(5,3);
\draw[thick] (-0.5,0)--(4,0)--(5,1.5);

\node[above] at (5.5,5) {$\mathcal R_1$};
\node[above] at (5.5,3.5) {$\mathcal R_2$};
\node[above] at (5.5,2) {$\mathcal R_3$};
\node[above] at (5.5,.5) {$\mathcal R_4$};

%
%
\draw[line width=1.2pt] (0,5.25)--(2,5.25);
\draw[line width=1.2pt] (0,3.75)--(2,3.75);
\draw[line width=1.2pt] (4,3.75)--(3,3.75);
\draw[line width=1.2pt] (4,2.25)--(3,2.25);
\draw[line width=1.2pt] (3,2.25)--(0,2.25);
\draw[line width=1.2pt] (3,0.75)--(0,0.75);
%
%
\filldraw [black] (2,5.25) circle (2pt) node [above] (q1) {$-q$};
\filldraw [black] (2,3.75) circle (2pt) node [above] (q2) {};
\filldraw [black] (3,3.75) circle (2pt) node [above] (02) {$0$};
\filldraw [black] (4,3.75) circle (2pt) node [above] (p2) {$p$};
\filldraw [black] (3,2.25) circle (2pt) node [above] (03) {};
\filldraw [black] (4,2.25) circle (2pt) node [above] (p3) {};
\filldraw [black] (3,0.75) circle (2pt) node [above] (04) {};

%
%
\draw[dashed] (2,5.25)--(2,3.75);
\draw[dashed] (4,3.75)--(4,2.25);
\draw[dashed] (3,3.75)--(3,0.75);

\end{tikzpicture}
\end{center}
\caption{The Riemann surface $\mathcal R$.}
\label{fig: Riemann surface}
\end{figure}
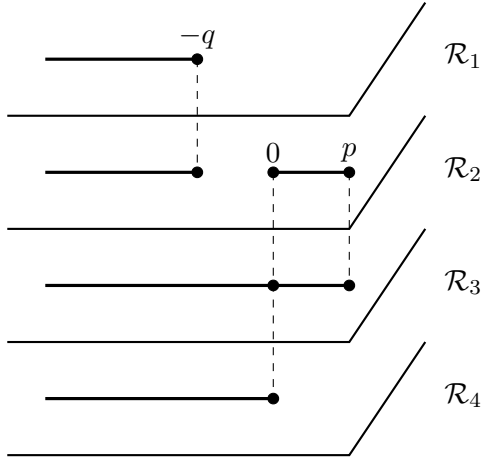

The surface $\mathcal R$ has the following branch points:
\begin{itemize}
\item Common branch points to $\mathcal R_1$ and $\mathcal R_2$ at $\infty$ and $z=-q$.
\item A common branch point to $\mathcal R_2$ and $\mathcal R_3$ at $z=p$.
\item A common branch point to $\mathcal R_3$ and $\mathcal R_4$ at $\infty$.
\item A common branch point to $\mathcal R_2$, $\mathcal R_3$ and $\mathcal R_4$ at $z=0$.
\end{itemize}
The last branch point enlisted above has ramification index $3$, whereas the others have ramification index $2$.  Consequently, it follows from the Riemann Hurwitz formula (cf. \cite{MRick}) that $\mathcal R$ has genus $0$.

\begin{prop}\label{prop:xij}
The function $\xi_j$ defined in \eqref{eq:definition_xi_functions} is analytic on $\mathcal R_j$.
\end{prop}
\begin{proof}
From the description of the supports of $\mu_1$, $\mu_2$ and $\mu_3$ in Proposition~\ref{prop:varicondition}, it follows that the $\xi_j$'s are analytic in their domains of definition as in \eqref{eq:definition_xi_functions}. A comparison of these domains with \eqref{def:sheets} then shows that we only need to show that $\xi_1$ and $\xi_2$ are analytic across $(-q,0)$, and in addition that $\xi_1$ does not have a singularity at $z=0$.

With the constraint measure $\sigma$ given in \eqref{def:constraint_measure}, a simple residue calculation shows that
\begin{equation}\label{eq:Csigma}
C^{\sigma}(z)=-\frac{\alpha}{\sqrt{z}}, \qquad z\in \C\setminus \R_-.
\end{equation}
By \eqref{eq:definition_xi_functions}, it is then readily seen that
\begin{equation}\label{eq:xi1_sigma}
\xi_1(z)=C^{\mu_1}(z)-C^{\sigma}(z),\qquad z\in \C\setminus \R_-.
\end{equation}
On account of the fact that $\mu_1=\sigma$ in $(-q,0)$ and the first equation in \eqref{eq:plemelj_relations}, we obtain that
$$
\xi_{1,+}(x)-\xi_{1,-}(x)=0, \qquad x\in (-q,0),
$$
thus concluding that $\xi_1$ is indeed analytic across $(-q,0)$, and also that $z=0$ is an isolated singularity. However, because $\mu_1$ is equal to $\sigma$ near $z=0$,
it follows from \eqref{eq:general_local_behavior} and \eqref{eq:behavior_cauchy_transform_endpoint} that, as $z\to 0$,
$$
C^{\mu_1}(z)=\Boh(|z|^{-1/2}).
$$
Hence, $\xi_1(z)=\Boh(|z|^{-1/2})$ as well, which implies that $z=0$ is in fact a removable singularity of $\xi_1$.

The proof for $\xi_2$ follows from the fact that
$$
\xi_2(z)=-\xi_1(z)+C^{\mu_2}(z),
$$
and that $\supp \mu_2$ does not intersect $(-q,0)$.

This completes the proof of Proposition \ref{prop:xij}.
\end{proof}

With the functions $\xi_j$, $j=1,2,3,4$, defined in \eqref{eq:definition_xi_functions}, set
\begin{equation}\label{def:xi}
\xi:\bigcup\limits_{j=1}^4 \mathcal R_j \to \C,\qquad \restr{\xi}{\mathcal R_j}=\xi_j.
\end{equation}
From the previous Proposition, $\xi$ is a well-defined meromorphic function on each of the sheets. It turns out that, in fact, it extends meromorphically to the whole surface $\mathcal R$, as claimed by our next result.

\begin{prop}\label{prop:merofunction}
The function $\xi$ defined in \eqref{def:xi} extends to a meromorphic function on the Riemann surface $\mathcal R$, and its unique pole is the branch point at $z=0$.
\end{prop}

\begin{proof}
We need to show that the analytic continuation of $\xi_j$ to $\mathcal R_{j+1}$ is $\xi_{j+1}$, $j=1,2,3$. For the sake of brevity, we will only consider the case when $j=1$, while the other cases can be proved similarly.

To show that the analytic continuation of $\xi_1$ to $\mathcal R_2$ is $\xi_2$, we note from \eqref{eq:derivative_potential_boundary_value}--\eqref{eq:plemelj_relations} that for $x<-q$
\begin{equation*}
2 \frac{\ud}{\ud x} U^{\mu_1}(x)=2\textrm{PV}\int \frac{\ud \mu_1(s)}{s-x}=C^{\mu_1}_+(x)+C^{\mu_1}_-(x)
\end{equation*}
and
\begin{equation*}
\frac{\ud}{\ud x} U^{\mu_2}(x)=C^{\mu_2}(x)=C^{\mu_2}_\pm(x).
\end{equation*}
Thus, by taking derivatives with respect to $x$ on both sides of \eqref{eq:variation12}, it follows that
\begin{align*}
0  = C^{\mu_1}_\pm(x)+C^{\mu_1}_\mp(x)-C^{\mu_2}(x) &= C^{\mu_1}_\pm(x)+\frac{\alpha}{(\sqrt{x})_\pm}+C^{\mu_1}_\mp(x)-C^{\mu_2}(x)+\frac{\alpha}{(\sqrt{x})_\mp} \nonumber \\
  & = \xi_{1,\pm}(x)-\xi_{2,\mp}(x),\qquad x\in (-\infty,-q), 
\end{align*}
as required, where we have made use of the fact that $(\sqrt{x})_+=-(\sqrt{x})_-$ for $x<0$ in the second equality.

Thus, the only possible poles of $\xi$ have to be at the branch points. Proposition~\ref{prop:varicondition} already tells us that the densities of $\mu_1$, $\mu_2$ and $\mu_3$ remain bounded except possibly at the origin. At this stage, we already know that $\xi$ - and hence each $\mu_j$ - is algebraic, so for each of $\mu_1,\mu_2$ and $\mu_3$ the behavior \eqref{eq:general_local_behavior}--\eqref{eq:behavior_cauchy_transform_endpoint} has to take place as $x\to p,-q$, for some $a>0$, giving us that $\xi$ cannot blow up at these points.

This way, we have shown that the only possible poles of $\xi$ are $z=0,\infty$. However, the large $z$ asymptotics of $\xi_j$, $j=1,2,3,4$, (which are immediate from \eqref{eq:definition_xi_functions} but for convenience also given in \eqref{eq:asymptotics_xi} below) show that the function $\xi$ is analytic at $\infty$ and non-constant, so the point $z=0$ common to the last three sheets has indeed to be a pole of $\xi$.

This completes the proof of Proposition \ref{prop:merofunction}.
\end{proof}

\subsection{Proof of Theorem~\ref{thm:spectral_curve}}\label{sec:thm_rational_parametrization}
By Proposition \ref{prop:merofunction}, we have that the functions $\xi_j$, $j=1,2,3,4$, are the
four distinct solutions to the following algebraic equation of order four:
\begin{equation*}
0=\prod_{i=1}^4(\xi-\xi_i)=\xi^4+R_3(z)\xi^3+R_2(z)\xi^2+R_1(z)\xi+R_0(z),
\end{equation*}
where the functions $R_j(z)$, $j=0,1,2,3$, are rational functions whose set of poles coincide with the set of poles for $\xi$, so they can have poles only at $z=0$. In view of \eqref{eq:definition_xi_functions}, it is easily seen that
\begin{equation*}
R_3(z)=-\xi_1-\xi_2-\xi_3-\xi_4=0.
\end{equation*}

To show that $R_1$, $R_2$ and $R_3$ are indeed given by the ones in \eqref{eq:spectral_curve}, we need to know the local behavior of each $\xi_j$, $j=2,3,4$, near the origin.

Because $\mathcal R$ has a branch point of ramification index $3$ at $z=0$, we have that, as $z \to 0$,
\begin{equation}\label{eq:local_behavior_xi}
\xi_j(z)=\tilde c_jz^{\frac{\delta}{3}}(1+\boh(z)),\qquad j=2,3,4,
\end{equation}
for some nonzero integer $\delta$ and some nonzero constants $\tilde c_2,\tilde c_3,\tilde c_4$.
Thus, in virtue of the Sokhotski-Plemelj relations \eqref{eq:plemelj_relations} and \eqref{eq:behavior_cauchy_transform_endpoint},
it follows that the densities of the three measures $\mu_1$, $\mu_2$ and $\mu_3$
behave algebraically near the origin as well, that is, as $z \to 0$,
\begin{equation}\label{eq:mujorigin}
\frac{\ud\mu_j}{\ud x}(x)= c_j z^{q_j}(1+\boh(1)),\qquad j=1,2,3,
\end{equation}
for some nonzero constants $c_1, c_2, c_3$ and some rational numbers $q_1$, $q_2$ and $q_3$ with $q_j>-1$. We note that
the latter condition holds because the measures $\mu_j$'s are finite.
Also, we see from \eqref{eq:variation232} and \eqref{eq:derivative_potential_boundary_value} that
$$
2 \re C^{\mu_3}_+(x)= \re C^{\mu_2}_+(x), \qquad x<0.
$$
This, together with \eqref{eq:plemelj_relations} and \eqref{eq:behavior_cauchy_transform_endpoint}, implies that
either $q_2, q_3 \geq 0$ or $-1<q_2=q_3<0$. Hence, we further obtain from \eqref{eq:local_behavior_xi} and the definition of $\xi_3$
in \eqref{eq:definition_xi_functions} that
$$
\frac{\delta}{3}=\left\{
                   \begin{array}{ll}
                     -\frac12, & \hbox{if $q_2,q_3 \geq 0$,} \\
                     \min\{q_3,-\frac{1}{2}\}, & \hbox{if $-1<q_2=q_3<0$.}
                   \end{array}
                 \right.
$$
Because $\delta$ is an integer, we learn from the above formula that the only possibility left is
\begin{equation}\label{eq:delta2}
\delta=-2,
\end{equation}
or equivalently,
\begin{equation}\label{eq:q2q3}
q_2=q_3=-\frac{2}{3}.
\end{equation}

In view of the Vieta relations, \eqref{eq:local_behavior_xi}, \eqref{eq:delta2} and the fact that $\xi_1$
is analytic near $z=0$, we obtain that, as $z\to 0$,
\begin{align}
R_0(z) & =\xi_1\xi_2\xi_3\xi_4=\Boh (z^{-2}), \label{eq:expansion_R0_origin}\\
R_1(z) & =- \xi_1\xi_2\xi_3-\xi_1\xi_2\xi_4-\xi_1\xi_3\xi_4 - \xi_2\xi_3\xi_4=\Boh(z^{-2}), \\
R_2(z) & = \sum_{\substack{j\neq k \\ 1 \leq j,k \leq 4}}\xi_j\xi_k = \Boh(z^{-4/3}).\label{eq:expansion_R2_origin}
\end{align}
Since $R_j$, $j=0,1,2,$ are rational functions with possible finite poles only at $z=0$, we conclude that
\begin{equation}\label{eq:Rjform}
R_j(z)=\frac{P_j(z)}{z^2}, \quad j=0,1, \quad \textrm{and}
\quad
R_2(z)=\frac{P_2(z)}{z},
\end{equation}
for some polynomials $P_0,P_1$ and $P_2$.

On the other hand, as $z\to\infty$, it follows from \eqref{eq:asymptotics_cauchy_transf_log_pot} and the local coordinates on $\mathcal R$ around the branch points at
$\infty$ that
\begin{equation}\label{eq:asymptotics_xi}
\begin{aligned}
&\xi_1(z) =\frac{\alpha} {\sqrt{z}}-\frac{1}{2z}-\frac{c_1}{z^{3/2}}+\Boh(z^{-2}),  && \qquad
\xi_2(z) =-\frac{\alpha} {\sqrt{z}}-\frac{1}{2z}+\frac{c_1}{z^{3/2}}+\Boh(z^{-2}),  \\
&\xi_3(z) =-\frac{\beta} {\sqrt{z}}+\frac{1}{2z}-\frac{c_3}{z^{3/2}}+\Boh(z^{-2}) , && \qquad
\xi_4(z) =\frac{\beta} {\sqrt{z}}+\frac{1}{2z}+\frac{c_3}{z^{3/2}}+\Boh(z^{-2}),
\end{aligned}
\end{equation}
for some constants $c_1$ and $c_3$.

Looking at the polynomial part of \eqref{eq:asymptotics_xi}, and expanding as in \eqref{eq:expansion_R0_origin}--\eqref{eq:expansion_R2_origin} but near $z=\infty$, we see from \eqref{eq:Rjform} that the coefficients $R_j$, $j=0,1,2$, reduce to the ones given in \eqref{eq:spectral_curve}.

This completes the first part of the proof of  Theorem \ref{thm:spectral_curve}.

To obtain the rational parametrization for \eqref{eq:spectral_curve}, which is known to exist because $\mathcal R$ has genus $0$, we first remark that the point $(\xi,z)=(0,\infty)$ is the only point of high
order branching of the curve, as all the other points are either simple branch points or regular points. As a consequence, the line
\begin{equation}\label{eq:ansatz_rational_parametrization}
z=\frac{t}{\xi^2},\qquad t\in \C,
\end{equation}
should intersect the point $(0,\infty)$ with high multiplicity. Substituting \eqref{eq:ansatz_rational_parametrization} into \eqref{eq:spectral_curve}, we arrive at
$$
t^2-(\alpha^2+\beta^2)t+(\alpha^2-\beta^2)\xi+\alpha^2\beta^2=0,
$$
from which it follows that $\xi=h(t)$ with $h$ given in \eqref{def:h}.
Thus, the map
\begin{equation}\label{eq:rational_parametrization}
(\xi,z)=H(t):=\left(h(t),\frac{t}{h(t)^2}\right), \qquad t \in \overline \C,
\end{equation}
is a rational parametrization of the Riemann surface $\mathcal R$. Counting its degree,
we see that this parametrization is maximal \cite[Theorem~4.21]{rational_algebraic_curves_book}.

This completes the proof of Theorem~\ref{thm:spectral_curve}.
\qed

\subsection{Proof of Theorem~\ref{thm_equilibrium_problem}}
\label{section:conclusion_proof}

As we observed at the end of Section \ref{sec:qualitative}, Proposition~\ref{prop:varicondition} already provides most of the claims in Theorem~\ref{thm_equilibrium_problem}, and it only remains to prove \eqref{eq:formula_q_endpoint}, \eqref{eq:local_behavior_mu1}, \eqref{eq:formula_p_endpoint}, \eqref{eq:local_behavior_mu2} and \eqref{eq:local_behavior_mu3}.

The local behavior of the density functions near the origin for the measures $\mu_2$ and $\mu_3$ claimed in \eqref{eq:local_behavior_mu2} and \eqref{eq:local_behavior_mu3} was already obtained; see \eqref{eq:mujorigin} and \eqref{eq:q2q3}. To verify the other formulas, we need an analysis of the spectral curve \eqref{eq:spectral_curve}.

From the construction of the Riemann surface $\mathcal R$, its only finite branch points are $p$, $-q$ and $0$. The discriminant of \eqref{eq:spectral_curve} with respect to $\xi$, as computed with Mathematica, is
\begin{equation}\label{eq:discriminant_spectral_curve}
\frac{1}{z^8}\left(\alpha ^2-\beta ^2\right)^2D_1(z)
\end{equation}
with
$$
D_1(z)=
-27 \left(\alpha ^2-\beta ^2\right)^2+4z \left(\alpha ^2+\beta ^2\right) (\alpha ^4-34 \alpha ^2 \beta ^2+\beta ^4)+16z^2 \alpha ^2 \beta ^2 \left(\alpha ^2-\beta ^2\right)^2
$$
being a quadratic polynomial. The leading coefficient of $D_1$ is positive and
$$D_1(0)=-27 \left(\alpha ^2-\beta ^2\right)^2<0,$$
so we have that the discriminant of the spectral curve has two simple zeros with distinct signs. Hence, these two real roots have to be the nonzero branch points of $\mathcal R$, namely $p$ and $-q$, and the formulas \eqref{eq:formula_q_endpoint} and \eqref{eq:formula_p_endpoint} are obtained by solving the quadratic equation $D_1(z)=0$.

Finally, from the relation \eqref{eq:xi1_sigma}, the definition of $\xi_2$ in \eqref{eq:definition_xi_functions}, and the fact that $p$ and $-q$ are two simple zeros of \eqref{eq:discriminant_spectral_curve}, we conclude \eqref{eq:local_behavior_mu1} and the local behavior of $\mu_2$ near $z=p$ as stated in \eqref{eq:local_behavior_mu2}.

This completes the proof of Theorem \ref{thm_equilibrium_problem}. \qed

\begin{remark}\label{remark:discriminant_zeros}
We note that the arguments above also imply that
$$\xi_j(x)-\xi_k(x)\neq 0, \qquad~~ j\neq k, $$ for $x\in (-q,0) \cup (p,+\infty)$, because the discriminant \eqref{eq:discriminant_spectral_curve} does not vanish on these two intervals.
\end{remark}

\subsection{The uniformization of the Riemann surface \texorpdfstring{$\mathcal R$}{R} in detail}\label{sec:rational_parametrization}

For later purpose, it is convenient to give a geometric description of the opens sets $\mathcal D_k$ that are uniquely determined by
\begin{equation}\label{def:Dk}
\mathcal{D}_k=H^{-1}(\mathcal R_k), \qquad k=1,2,3,4,
\end{equation}
where $H$ is given in \eqref{eq:rational_parametrization}. To obtain these sets, we first analyze the images of the branch points of $\mathcal R$ on the $t$-sphere.

The finite branch points of $\mathcal R$ where $\xi$ remains bounded, that is, the branch points $z=p$ and $z=-q$,
are determined as the values of $t$ for which the equation
\begin{equation}\label{def:zt}
z=z(t)=\frac{t}{h(t)^2}=\frac{t(\beta^2-\alpha^2)^2}{(t-\alpha^2)^2(t-\beta^2)^2}
\end{equation}
has multiple solutions. Since
$$
z'(t)=-\frac{(\beta^2-\alpha^2)^2}{(t-\alpha^2)^3(t-\beta^2)^3}\widehat h(t),
$$
where
$$
\widehat h(t):=3t^2-t(\alpha^2+\beta^2)-\alpha^2\beta^2,
$$
these points are the roots of $\widehat h (t)$, i.e.,
\begin{equation}\label{def:tpm}
t_\pm =\frac{1}{6}(\beta^2+\alpha^2\pm \sqrt{\alpha^4+14\alpha^2\beta^2+\beta^4})
\end{equation}
with $t_-<0<t_+$. As a consequence,
$$
z(t_+)=\frac{t_+}{h(t_+)^2}>0,\qquad z(t_-)=\frac{t_-}{h(t_-)^2}<0,
$$
so actually
$$
z(t_+)=p,\qquad z(t_-)=-q,
$$
which is also consistent with \eqref{eq:formula_p_endpoint} and \eqref{eq:formula_q_endpoint}.

To find the $t$-points corresponding to $\infty^{(1)}=\infty^{(2)}$ and $\infty^{(3)}=\infty^{(4)}$,
we must find the values of $t$ for which $z(t)$ blows up. These are thus given by the zeros of $h(t)$,
that is,
$$
t=\alpha^2\quad \textrm{or} \quad t=\beta^2.
$$
To identify the images, we see from \eqref{eq:asymptotics_xi} that
$$
\sqrt{z}\xi_{1,2}(z)=\pm\alpha+\boh(1),\quad \sqrt{z}\xi_{3,4}(z)=\mp \beta+\boh(1), \quad z\to \infty,
$$
whereas using the rational parametrization $H$,
$$
|\sqrt{z}\xi(z)|=\left|\sqrt{\frac{t}{h(t)^2}}h(t)\right|=
\begin{cases}
 \alpha + \boh(1),&~~ t\to \alpha^2, \\
\beta + \boh(1), &~~ t\to \beta^2,
\end{cases}
$$
hence,
$$
z(\alpha^2)=\infty^{(1)}=\infty^{(2)},\qquad z(\beta^2)=\infty^{(3)}=\infty^{(4)}.
$$
Moreover, since
$$
\widehat h(\alpha^2)=-2\alpha^2(\beta^2-\alpha^2)<0,\qquad \widehat h(\beta^2)=2\beta^2(\beta^2-\alpha^2)>0,
$$
we have the ordering
$$
t_-<0<\alpha^2<t_+<\beta^2.
$$

The remaining branch point of $\mathcal R$ is the one at $z=0$ connecting $\mathcal R_2$, $\mathcal R_2$ and $\mathcal R_3$. According to Proposition~\ref{prop:merofunction}, this branch point corresponds to the only $t$-point for which $h(t)=\xi$ blows up, so it is $t=\infty$.

In summary, we have the following proposition regarding the mapping properties of the rational parametrization $H$ defined in \eqref{eq:rational_parametrization}.
\begin{prop}\label{prop:zt}
The $z\leftrightarrow t $ correspondence for the branch points of the Riemann surface $\mathcal R$ under the rational parametrization $H$ is listed in Table~\ref{table:branchpoints}. Furthermore, the local coordinate $z(t)$ admits the following behavior near each of its critical points.
%
\begin{equation}\label{eq:asymptotics_rational_parametrization}
\begin{aligned}
z(t) & = \frac{\alpha^2}{(\beta^2-\alpha^2)^2}\frac{1}{(t-\alpha^2)^{2}}(1+\Boh(t-\alpha^2)),&&~~ t\to \alpha^2, \\
z(t) & = \frac{\beta^2}{(\beta^2-\alpha^2)^2}\frac{1}{(t-\beta^2)^{2}}(1+\Boh(t-\beta^2)),&&~~  t\to \beta^2, \\
z(t) & = -q+\Boh((t-t_-)^2),&&~~ t \to t_-, \\
z(t) & = p+\Boh((t-t_+)^2),&&~~ t \to t_+, \\
z(t) & = \frac{(\beta^2-\alpha^2)^2}{t^3}(1+\Boh(t^{-1})),&&~~ t\to \infty.
\end{aligned}
\end{equation}
\end{prop}
\begin{proof}
We have already proved the images of the branch points of $\mathcal R$ in the $t$-sphere, while the local behavior of $z$ near each of its critical points follows directly from \eqref{def:zt}.
\end{proof}

\begin{table}[h]
\begin{center}
\begin{tabular}{c|c}
branch points on $\mathcal R$ & \shortstack{points on $t$-sphere \\ (in increasing order of magnitude)} \\
\hline
$-q$ & $t_-$ \\
\hline
$\infty^{(1)}=\infty^{(2)}$ & $\alpha^2$ \\
\hline
$p$ & $t_+$ \\
\hline
$\infty^{(3)}=\infty^{(4)}$ & $\beta^2$ \\
\hline
$0$ & $\infty$
\end{tabular}
\caption{The $z\leftrightarrow t $ correspondence for the branch points of $\mathcal R$.}\label{table:branchpoints}
\end{center}
\end{table}

The inverse map $H^{-1}$ maps the branch cuts $\overline {\Delta_k}$ of $\mathcal R$ to simple analytic arcs $\gamma_k^\pm\subset \overline \C$, $k=1,2,3$ that
can only intersect at the points of the $t$-sphere enlisted in Table \ref{table:branchpoints}.
Due to the symmetry, $\gamma_k^{-}$ is the complex conjugate of $\gamma_k^+$, and the $+$-sign indicates that $\gamma_k^+$ is on the upper half plane.
The index of each of these arcs is determined by the following rules.
\begin{itemize}
\item $\gamma_1^\pm$ is the arc that connects $t_-$ and $\alpha^2$.
\item $\gamma_2^\pm$ is the arc that connects $t_+$ and $\infty$.
\item $\gamma_3^\pm$ is the arc that connects $\beta^2$ and $\infty$.
\end{itemize}
A basic geometric analysis of the conformal map $H$ then shows the following.
\begin{itemize}
\item The contour $H(\gamma_1^+)$ ($H(\gamma_1^-$)) is the upper (lower) part of the interval $\overline {\Delta_1}$ on $\mathcal R_1$, which is the same as the lower (upper) part of this interval on $\mathcal R_2$.

\item The contour $H(\gamma_2^-)$ ($H(\gamma_2^+)$) is the upper (lower) part of the interval $\overline {\Delta_2}$ on $\mathcal R_2$, which is the same as the lower (upper) part of this interval on $\mathcal R_3$.

\item The contour $H(\gamma_3^+)$ ($H(\gamma_3^-)$) is the upper (lower) part of the interval $\overline {\Delta_3}$ on $\mathcal R_3$, which is the same as the lower (upper) part of this interval on $\mathcal R_4$.
\end{itemize}
This also means that each of the arcs $\gamma_k:=\gamma_k^+\cup\gamma_k^-$ is an analytic closed contour on $\overline \C$, which is the common boundary component of
$\mathcal D_k$ and $ \mathcal D_{k+1}$, $k=1,2,3$, where $\mathcal{D}_k$ is defined in \eqref{def:Dk}. The above correspondence is illustrated in Figure~\ref{fig:uniformization}.

Finally, we observe that $H$ maps the intervals $(t_-,0)$ and $(-\infty,t_-)$ to the interval $(-q,0)$ on the sheets $\mathcal R_1$ and $\mathcal R_2$, respectively. This is an immediate consequence of the description above, combined with real symmetry.

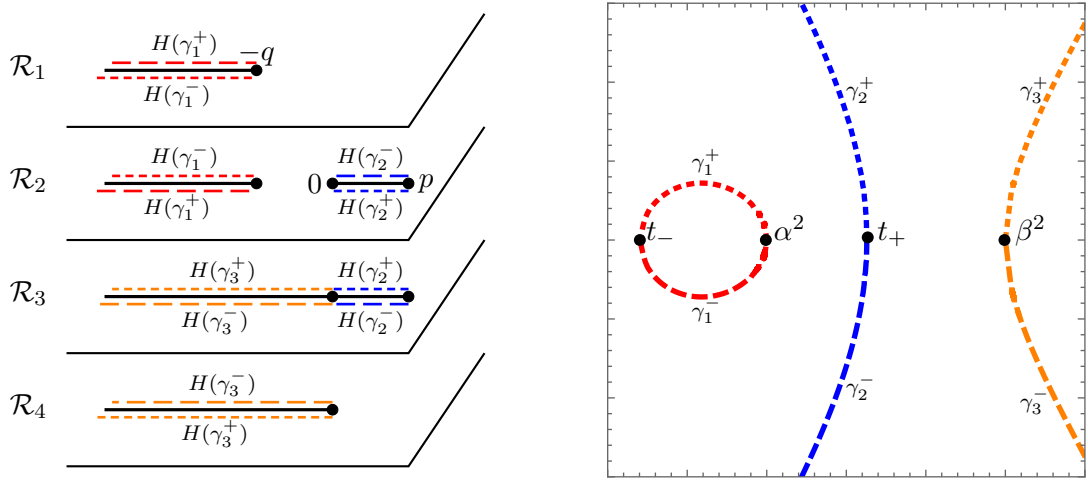
\begin{figure}
\begin{center}
\begin{minipage}{.5\textwidth}
\begin{tikzpicture}

%
%
\draw[thick] (-0.5,4.5)--(4,4.5)--(5,6);
\draw[thick] (-0.5,3)--(4,3)--(5,4.5);
\draw[thick] (-0.5,1.5)--(4,1.5)--(5,3);
\draw[thick] (-0.5,0)--(4,0)--(5,1.5);

\node[above] at (-1,5) {$\mathcal R_1$};
\node[above] at (-1,3.5) {$\mathcal R_2$};
\node[above] at (-1,2) {$\mathcal R_3$};
\node[above] at (-1,.5) {$\mathcal R_4$};

%
%
\draw[line width=1.2pt] (0,5.25)--(2,5.25);
\draw[line width=1.2pt] (0,3.75)--(2,3.75);
\draw[line width=1.2pt] (4,3.75)--(3,3.75);
\draw[line width=1.2pt] (4,2.25)--(3,2.25);
\draw[line width=1.2pt] (3,2.25)--(0,2.25);
\draw[line width=1.2pt] (3,0.75)--(0,0.75);

%
%
\draw[thin,dash pattern={on 7pt off 3pt},red,line width=1pt] (.1,5.35) --
node[above=-2pt] {$\scriptstyle\color{black}H(\gamma_1^+)$} (2,5.35);
\draw[thin,dashed,dash pattern={on 3pt off 2pt},red,line width=1pt] (-.1,5.15) -- node[below=-3pt] {$\scriptstyle\color{black}H(\gamma_1^-)$} (2,5.15);
\draw[thin,dash pattern={on 3pt off 2pt},red,line width=1pt] (.1,3.85) --
node[above=-2pt] {$\scriptstyle\color{black}H(\gamma_1^-)$}(2,3.85);
\draw[thin,dash pattern={on 7pt off 3pt},red,line width=1pt] (-.1,3.65) --
node[below=-3pt] {$\scriptstyle\color{black}H(\gamma_1^+)$} (2,3.65);

%
%
\draw[dash pattern={on 7pt off 3pt},blue,line width=1pt] (4,3.85) --
node[above=-2pt] {$\scriptstyle\color{black}H(\gamma_2^-)$} (3,3.85);
\draw[dash pattern={on 3pt off 2pt},blue,line width=1pt] (4,3.65) --
node[below=-3pt] {$\scriptstyle\color{black}H(\gamma_2^+)$} (3,3.65);
\draw[dash pattern={on 3pt off 2pt},blue,line width=1pt] (4,2.35) --
node[above=-2pt] {$\scriptstyle\color{black}H(\gamma_2^+)$} (3,2.35);
\draw[dash pattern={on 7pt off 3pt},blue,line width=1pt] (4,2.15) --
node[below=-3pt] {$\scriptstyle\color{black}H(\gamma_2^-)$} (3,2.15);

%
%
\draw[dash pattern={on 3pt off 2pt},orange,line width=1pt] (3,2.35) --
node[above=-2pt] {$\scriptstyle\color{black}H(\gamma_3^+)$} (.1,2.35);
\draw[dash pattern={on 7pt off 3pt},orange,line width=1pt] (3,2.15) --
node[below=-3pt] {$\scriptstyle\color{black}H(\gamma_3^-)$} (-.1,2.15);
\draw[dash pattern={on 7pt off 3pt},orange,line width=1pt] (3,0.85) --
node[above=-2pt] {$\scriptstyle\color{black}H(\gamma_3^-)$} (.1,0.85);
\draw[dash pattern={on 3pt off 2pt},orange,line width=1pt] (3,0.65) --
node[below=-3pt] {$\scriptstyle\color{black}H(\gamma_3^+)$} (-.1,0.65);

%
%

\filldraw [black] (2,5.25) circle (2pt) node [above=-2pt] (q1) {$-q$};
\filldraw [black] (2,3.75) circle (2pt) node [above] (q2) {};
\filldraw [black] (3,3.75) circle (2pt) node [left] (02) {$0$};
\filldraw [black] (4,3.75) circle (2pt) node [right] (p2) {$p$};
\filldraw [black] (3,2.25) circle (2pt) node [above] (03) {};
\filldraw [black] (4,2.25) circle (2pt) node [above] (p3) {};
\filldraw [black] (3,0.75) circle (2pt) node [above] (04) {};

\end{tikzpicture}
\end{minipage}%
\begin{minipage}{.5\textwidth}
\begin{overpic}[scale=1]
{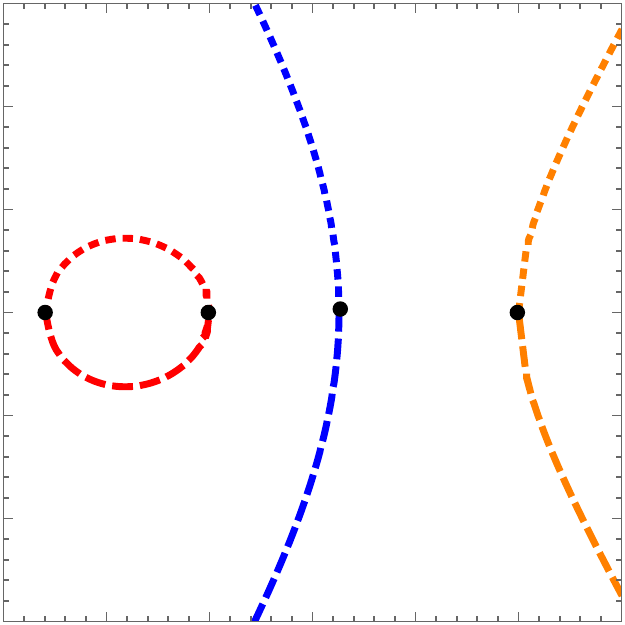}
\put(18,65){$\scriptstyle \gamma_1^+$}
\put(18,34){$\scriptstyle \gamma_1^-$}
\put(50,80){$\scriptstyle \gamma_2^+$}
\put(50,18){$\scriptstyle \gamma_2^-$}
\put(86,80){$\scriptstyle \gamma_3^+$}
\put(86,15){$\scriptstyle \gamma_3^-$}
\put(8.5,50){$t_-$}
\put(35,50){$\alpha^2$}
\put(56.5,50){$t_+$}
\put(85,50){$\beta^2$}
\end{overpic}
\end{minipage}
\caption{The uniformization of the Riemann surface $\mathcal R$. Right figure: the arcs $\gamma_k^+$ (short dashed), $\gamma_k^-$ (long dashed), $k=1,2,3,4$, and the $t$-points on the $t$-sphere that give rise to branch points on $\mathcal R$ (compare with Table \ref{table:branchpoints}). Left figure, the images on $\mathcal R$ of the arcs $\gamma_k^\pm$ under the rational parametrization $H$, with the same short-long dashed codes.}
\label{fig:uniformization}
\end{center}
\end{figure}

\section{Auxiliary functions}\label{sec:auxifunc}
In this section, we  introduce some auxiliary functions for later use.

\subsection{The \texorpdfstring{$\lambda$}{lambda}-functions}
The $\lambda$-functions are defined as the anti-derivative of the $\xi$-functions \eqref{eq:definition_xi_functions}:
\begin{align}
&\lambda_1(z) = \int_{-q}^{z}\xi_1(s)\ud s+ \int_{p}^{-q}\xi_{2,-}(s)\ud s, &&  z\in \C\setminus (-\infty,-q], \label{def:lambda_functions1}\\
&\lambda_2(z) = \int_{p}^{z}\xi_2(s)\ud s, &&  z\in \C\setminus (-\infty,p], \label{def:lambda_functions2} \\
&\lambda_3(z) = \int_{p}^{z}\xi_3(s)\ud s , &&  z\in \C\setminus (-\infty,p], \label{def:lambda_functions3}\\
&\lambda_4(z) = \int_{0}^{z}\xi_4(s)\ud s -\int_0^p \xi_{2,+}(s)\ud s , &&  z\in \C\setminus (-\infty,0].\label{def:lambda_functions4}
\end{align}

We have the following asymptotic behaviors of the $\lambda$-functions for large $z$.
\begin{prop}\label{prop:asylambda}
As $z\to \infty$, we have
\begin{align}
\lambda_1(z) & =2\alpha \sqrt{z}-\frac{1}{2}\log z + \theta_1 + \frac{2c_1}{\sqrt{z}}+\Boh(z^{-1}), \label{eq:lamda1asy}\\
\lambda_2(z) & =-2\alpha \sqrt{z}-\frac{1}{2}\log z + \theta_1-\pi i - \frac{2c_1}{\sqrt{z}}+\Boh(z^{-1}), \label{eq:lambda2asy}\\
\lambda_3(z) & =-2\beta \sqrt{z}+\frac{1}{2}\log z + \theta_3  + \frac{2c_3}{\sqrt{z}}+\Boh(z^{-1}), \label{eq:lamda3asy}\\
\lambda_4(z) & =2\beta \sqrt{z}+\frac{1}{2}\log z + \theta_3-\pi i - \frac{2c_3}{\sqrt{z}}+\Boh(z^{-1}),\label{eq:lamda4asy}
\end{align}
for some constants $\theta_1$ and $\theta_3$, where $c_1$ and $c_3$ are the same as in \eqref{eq:asymptotics_xi}.
\end{prop}
\begin{proof}
In virtue of \eqref{eq:asymptotics_xi}, it is readily seen that, as $z\to \infty$,
\begin{align*}
\lambda_1(z) & =2\alpha \sqrt{z}-\frac{1}{2}\log z + \theta_1 + \frac{2c_1}{\sqrt{z}}+\Boh(z^{-1}), \\
\lambda_2(z) & =-2\alpha \sqrt{z}-\frac{1}{2}\log z + \theta_2 - \frac{2c_1}{\sqrt{z}}+\Boh(z^{-1}), \\
\lambda_3(z) & =-2\beta \sqrt{z}+\frac{1}{2}\log z + \theta_3 + \frac{2c_3}{\sqrt{z}}+\Boh(z^{-1}), \\
\lambda_4(z) & =2\beta \sqrt{z}+\frac{1}{2}\log z + \theta_4 - \frac{2c_3}{\sqrt{z}}+\Boh(z^{-1}),
\end{align*}
for some constants $\theta_1,\theta_2,\theta_3,\theta_4$. To show that $\theta_2=\theta_1-\pi i$, we note from \eqref{def:lambda_functions1} and \eqref{def:lambda_functions2}
that, if $x<-q$,
\begin{align*}
\lambda_{1,+}(x)-\lambda_{2,-}(x)
& =\int_{-q}^{x}\xi_{1,+}(s)\ud s + \int_{p}^{-q}\xi_{2,-}(s)\ud s-\int_{p}^{x}\xi_{2,-} (s) \ud s
\\
& =\int_{-q}^{x}(\xi_{1,+}(s)-\xi_{2,-}(s))\ud s =0,
\end{align*}
since $\xi_{1,+}(s)=\xi_{2,-}(s)$ for $s<-q$. Inserting \eqref{eq:lamda1asy} and \eqref{eq:lambda2asy} into the above equality yields that  $\theta_2=\theta_1-\pi i$.

In a similar manner, it is easily seen that
$$\lambda_{3,-}(x)-\lambda_{4,+}(x) =0, \qquad x<0. $$
This, together with \eqref{eq:lamda3asy} and \eqref{eq:lamda4asy}, implies that $\theta_4=\theta_3-\pi i$, as required.

This completes the proof of Proposition \ref{prop:asylambda}.
\end{proof}

\subsection{The $\phi$-functions}

For the sake of clarity, we also define the following $\phi$-functions:
\begin{align}
 \phi_1(z) & = \int_{-q}^z(\xi_1(s)-\xi_2(s))\ud s, && \quad z\in \C\setminus \left((-\infty,-q]\cup \R_+\right), \label{def:phi_functions1}\\
 \phi_2(z) & = \int_{p}^z(\xi_2(s)-\xi_3(s))\ud s, && \quad z\in \C\setminus (-\infty, p],\label{def:phi_functions2} \\
 \phi_3(z) & = \int_0^z(\xi_3(s)-\xi_4(s))\ud s, && \quad z\in \C\setminus (-\infty,p],\label{def:phi_functions3}
\end{align}
where  the path of integration in $\phi_3$ emerges from $z=0$ in the upper half plane. Note that each of the
$\lambda$-functions and the $\phi$-functions is analytic in its domain of definition.

Some properties of these auxiliary functions are collected in the following proposition.
\begin{prop}\label{prop:lamdaphi}
Let $x\in \R$, with $\Delta_i$, $i=1,2,3$, defined in \eqref{def:deltai}, we have
\begin{align*}
 \lambda_{1,+}(x)-\lambda_{1,-}(x) & = \phi_{1,+}(x), \quad x\in \Delta_1, 
\\
 \lambda_{2,+}(x)-\lambda_{2,-}(x) & =
\begin{cases}
\phi_{2,+}(x), & x\in \Delta_2, \\
-2\pi i , & x\in \Delta_3\setminus \overline{\Delta_1}, \\
-2 \pi i + \phi_{1,-}(x), & x \in \Delta_1
\end{cases}
%
\\
\lambda_{3,+}(z)-\lambda_{3,-}(z) & =
\begin{cases}
\phi_{2,-}(x), & x\in\Delta_2, \\
2 \pi i+ \phi_{3,+}(x), & x\in\Delta_3,
\end{cases}
%
\\
\lambda_{4,+}(x)-\lambda_{4,-}(x) & =2\pi i + \phi_{3,-} (x), \quad x\in \Delta_3,
%
\end{align*}
and
\begin{align*}
\lambda_{1,+}(x)-\lambda_{2,-}(x) & =
\begin{cases}
0, & x\in\Delta_1, \\
\phi_{1}(x), & x\in \Delta_3\setminus \overline{\Delta_1},
\end{cases}
%
\\
\lambda_{1,-}(x)-\lambda_{2,+}(x)&=2\pi i, \quad x\in\Delta_1,
%
\\
\lambda_{2,\pm}(x)-\lambda_{3,\mp}(x) & =
\begin{cases}
0, & x\in \Delta_2 , \\
\phi_{2}(x) , & x>p,
\end{cases}
%
\\
\lambda_{3,+}(x)-\lambda_{4,-}(x)&=2\pi i, \quad x\in \Delta_3,
%
\\
\lambda_{3,-}(x)-\lambda_{4,+}(x)&=0, \qquad x \in \Delta_3.
%
\end{align*}
Furthermore, we have
\begin{align}
\phi_{1,+}(x)+\phi_{1,-}(x) &= 0, \quad x \in \Delta_1,
\label{eq:phi1_pm}
\\
 \phi_{2,+}(x)+\phi_{2,-}(x)&= 0,\quad x\in \Delta_2,
\label{eq:phi2_pm}
\\
\phi_{3,+}(x)+\phi_{3,-}(x) &= -2\pi i, \quad x \in \Delta_3.
\label{eq:phi3_pm}
\end{align}
\end{prop}
\begin{proof}
These formulas follow directly from the definitions of the $\lambda$-functions and the $\phi$-functions given in \eqref{def:lambda_functions1}--\eqref{def:lambda_functions4} and \eqref{def:phi_functions1}--\eqref{def:phi_functions3}, as well as Proposition \ref{prop:merofunction}. We omit the details here.
\end{proof}

Finally, we present some inequalities satisfied by the $\phi$-functions in the neighborhoods of their branch cuts. These inequalities will be essential in our further asymptotic analysis.

\begin{prop}\label{prop:inequalities_phi}
For each $i=1,2,3$, there exists an open neighborhood $\mathcal G_i$ of the interval $\Delta_i$, such that the following inequalities hold:
\begin{align*}
\re \phi_1(z)>0, &\qquad z\in \mathcal G_1\setminus \Delta_1,\\
\re \phi_2(z)<0, &\qquad z\in \mathcal G_2\setminus \Delta_2, \\
\re \phi_3(z)<0, &\qquad z\in \mathcal G_3\setminus \Delta_3.
\end{align*}
Furthermore, we also have that
\begin{equation}\label{eq:inequalities_phi_off_support}
\begin{aligned}
\phi_2(x)>0, & \qquad x>p, \\
\phi_1(x)<0, & \qquad x\in \Delta_3 \setminus \overline{\Delta_1}.
\end{aligned}
\end{equation}
\end{prop}
\begin{proof}
We will only prove the existence of $\mathcal G_1$, since the existence of $\mathcal G_2$ and $\mathcal G_3$ follow in a similar manner.

If $x \in \Delta_1= (-\infty, -q)$, note that $\xi_{1,\pm}(x)=\xi_{2,\mp}(x)$, it is readily
seen from \eqref{eq:definition_xi_functions}, \eqref{eq:plemelj_relations}, \eqref{eq:Csigma} and \eqref{def:phi_functions1} that
\begin{align*}
\phi_{1,\pm}(x) & =\pm\int_{-q}^x(\xi_{1,+}(s)-\xi_{1,-}(s))\ud s
 = \pm\int_{-q}^x\left(C^{\mu_1}_+(s)-C^{\mu_1}_-(s)+\frac{\alpha}{\sqrt{s}_+}-\frac{\alpha}{\sqrt{s}_-}\right) \ud s
\\
&=\pm\int_{-q}^x\left(C_+^{\mu_1-\sigma}(s)-C_-^{\mu_1-\sigma}(s)\right)\ud s=\pm2\pi i (\sigma-\mu_1)((x,-q)).
\end{align*}
Thus $\phi_{1,\pm}(x)$ is purely imaginary along $\Delta_1$, and the functions
$$x\mapsto \im \phi_{1,+}(x),\qquad x\mapsto \im \phi_{1,-}(x),$$ are strictly decreasing and increasing, respectively. By the Cauchy-Riemann equations, we then get immediately that $\re \phi_1 (z)$ is strictly positive above and below the interval $(-\infty,-q)$, assuring the existence of $\mathcal G_1$.

To conclude the first inequality in \eqref{eq:inequalities_phi_off_support}, we start with the identity
\begin{align*}
\re \phi_2(z) & = \re \int_{p}^z(\xi_2(s)-\xi_3(s))\ud s \nonumber
 =\re \int_p^z\left(2C^{\mu_2}(s)-C^{\mu_1}(s)-C^{\mu_3}(s)+\frac{\beta-\alpha}{\sqrt{s}}\right)\ud s\nonumber\\
& =2U^{\mu_2}(z)-U^{\mu_1}(z)-U^{\mu_3}(z)+2(\beta-\alpha)\re\sqrt{z}-c,\qquad z\in \C\setminus (-\infty,p], 
\end{align*}
for some constant $c$. This identity extends to $\C$ by continuity, and in virtue of the equality \eqref{eq:variation123}, we get
$$
0=\re \phi_2(p)=\ell -c,
$$
so $c=\ell$. The inequality then follows directly from \eqref{eq:variation123ineq}.

In a similar fashion, the second inequality in \eqref{eq:inequalities_phi_off_support} follows from \eqref{eq:umu12} and \eqref{eq:variation23}.

This completes the proof of Proposition \ref{prop:inequalities_phi}.
\end{proof}

We are now ready to carry out asymptotic analysis of the RH problem \ref{rhp:Y} for $Y$.

\section{First transformation \texorpdfstring{$Y\rightarrow X$}{Y to X}}\label{sec:firstrans}
The aim of this transformation to simplify the block matrix $W(x)$ appearing in the jump condition
\eqref{defjumpmatrixY} for $Y$. The cost we have to pay is to create a new jump on the negative real axis.
Following \cite{DKRZ12,KMW09}, the main idea is to use the special properties of modified Bessel functions.

We start by setting
\begin{equation}\label{def of y_i}
y_{1,a}(z)=z^{(a+1)/2}I_{a+1}(2\sqrt{z}), \qquad
y_{2,a}(z)=z^{(a+1)/2}K_{a+1}(2\sqrt{z}),
\end{equation}
where $a>-1$ is a real parameter. In general, we have that both $y_{1,a}$ and $y_{2,a}$ are analytic in the complex plane with a cut along the negative real axis. Some properties of $y_{i,a}$ are collected in what follows for later use.

\begin{itemize}
  \item Connection formulas (see \cite[Formulas 10.34.1 and 10.34.2]{DLMF}): if $x<0$,
  \begin{equation}
\begin{aligned}\label{jump for y_i}
\left(y_{1,a}\right)_+(x) &= e^{2 a \pi i}\left(y_{1,a}\right)_{-}(x), \\
\left(y_{2,a}\right)_+(x) &= \left(y_{2,a}\right)_-(x)+ i \pi e^{ a \pi i}\left(y_{1,a}\right)_-(x),
\end{aligned}
\end{equation}
where the orientation of $\R_-$ is taken from the left to the right.
  \item Derivatives (see \cite[Formulas 10.29.2 and 10.29.5]{DLMF}):
\begin{equation}\label{derivatives of y_i}
y_{1,a}'(z)=z^{a/2}I_{a}(2\sqrt{z})=y_{1,a-1}(z), \qquad
y_{2,a}'(z)=-z^{a/2}K_{a}(2\sqrt{z})=-y_{2,a-1}(z).
\end{equation}
  \item The Wronskian relation (see \cite[Formula 10.28.2]{DLMF}):
\begin{equation}\label{eq:Wronskian relation}
y_{1,a}(z)y_{2,a}'(z)-y_{1,a}'(z)y_{2,a}(z)=-z^a/2,\qquad z\in\mathbb{C} \setminus \R_-.
\end{equation}
\end{itemize}

By \eqref{def of y_i} and \eqref{derivatives of y_i}, it is readily seen that the matrix $W$ in \eqref{eq:matrixW} can be rewritten as
\begin{equation}\label{eq:Wdecomp}
W(x)=w_1(x)^T w_2(x),
\end{equation}
where
\begin{align}
w_1(x) & :=
\begin{pmatrix}
\omega_{\kappa,\alpha}(x) & \omega_{\kappa+1,\alpha}(x)
\end{pmatrix}
=
\begin{pmatrix}
\tau_1^{-\kappa} y'_{1,\kappa}(\tau_1^2 x) &
\tau_1^{-\kappa-1} y_{1,\kappa}(\tau_1^2 x)
\end{pmatrix},
\label{def:w1}\\
w_2(x) & :=
\begin{pmatrix}
\rho_{\nu-\kappa,\beta}(x) & \rho_{\nu-\kappa+1,\beta}(x)
\end{pmatrix}=
\begin{pmatrix}
-\tau_2^{\kappa-\nu}y'_{2,\nu-\kappa}(\tau_2^2 x) &
\tau_2^{\kappa-\nu-1}y_{2,\nu-\kappa}(\tau_2^2 x)
\end{pmatrix}, \label{def:w2}
\end{align}
with
\begin{equation*}
\tau_1:=\alpha n,\qquad \tau_2:=\beta n.
\end{equation*}

With the help of functions $y_{i,a}(z)$ given in \eqref{def of y_i}, we further define two $2\times 2$ matrices
\begin{equation}\label{A1}
A_1(z)=\tau_1^{-\kappa}z^{-\frac{\kappa}{2}}
\begin{pmatrix}
-\frac{1}{\pi i}y_{2,\kappa}'(\tau_1^2z) & y_{1,\kappa}'(\tau_1^2z) \\
-\frac{1}{\pi i}\tau_1^{-1}y_{2,\kappa}(\tau_1^2z)  & \tau_1^{-1}y_{1,\kappa}(\tau_1^2z)
\end{pmatrix}
\end{equation}
and
\begin{equation}\label{A2}
A_2(z)=2\tau_2^{\kappa-\nu}z^{\frac{\kappa-\nu}{2}}
\begin{pmatrix}
 y_{1,\nu-\kappa}(\tau_2^2z) &   -\frac{1}{\pi i}y_{2,\nu-\kappa}(\tau_2^2z) \\
 \tau_2y_{1,\nu-\kappa}'(\tau_2^2z) & -\frac{\tau_2}{\pi i}y_{2,\nu-\kappa}'(\tau_2^2z)
\end{pmatrix}.
\end{equation}
In view of \eqref{eq:Wronskian relation}, it is easily seen that
\begin{equation}\label{eq:detA12}
\det A_1(z)=\frac{1}{2\pi \tau_1 i} \qquad \textrm{and} \qquad
\det A_2(z)=\frac{2\tau_2}{\pi i}.
\end{equation}

Our first transformation is then defined by
\begin{equation}\label{eq:YtoX}
X(z)=C_XY(z)\diag (A_1(z),A_2(z))\diag\left(z^{\frac{\kappa}{2}\sigma_3},z^{\frac{\kappa-\nu}{2}\sigma_3}\right),
\end{equation}
where $\sigma_3=\begin{pmatrix}
0 & 1 \\
-1 & 0
\end{pmatrix}$ is the third Pauli matrix and
\begin{equation*}
C_X=
\diag\left(
\sqrt{2\pi \tau_1}
\begin{pmatrix}
i & 0 \\
\frac{4(\kappa+1)^2-1}{16\tau_1} & 1
\end{pmatrix},
\sqrt{\frac{\pi}{2\tau_2}}
\begin{pmatrix}
1 & \frac{4(\nu-\kappa+1)^2-1}{16\tau_2} \\
0 & i
\end{pmatrix}
\right).
\end{equation*}
By \eqref{eq:detA12}, it is easily seen that $\det X=1$. We further have that $X$ satisfies the following RH problem.

\begin{lem}\label{lem:RHpforX}
The function $X$ defined in \eqref{eq:YtoX} has the following properties:
\begin{enumerate}
\item[\rm (1)] $X$ is defined and analytic in $ \mathbb{C} \setminus \R$.
\item[\rm (2)] For $x\in \mathbb{R}$, $X(z)$ satisfies the jump conditions
\begin{equation}\label{jump for X}
  X_+(x)=X_-(x)
  \left\{
  \begin{array}{ll}
  I_4+x^{\kappa}E_{23}, & \hbox{if $x>0$,} \\
  I_4-|x|^\kappa E_{21}-|x|^{\nu-\kappa}E_{34}, & \hbox{if $x<0$},
  \end{array}
  \right.
  \end{equation}
  where the $4\times 4$ matrix $E_{ij}$ is defined in \eqref{def:Eij}.

\item[\rm (3)] As $z\to\infty$, we have
\begin{equation}\label{eq:Xinfty}
X(z)=(I_4+\Boh(z^{-1}))\mathcal B(z)\diag\left( z^\frac{n}{2}e^{-2\tau_1z^{\frac{1}{2}} } ,z^\frac{n}{2}e^{2\tau_1z^{\frac{1}{2}} },z^{-\frac{n}{2}}e^{2\tau_2z^{\frac{1}{2}} } ,z^{-\frac{n}{2}}e^{-2\tau_2z^{\frac{1}{2}} } \right),
\end{equation}
where
\begin{align}\label{def:Bz}
\mathcal B(z)& = \frac{1}{\sqrt{2}}\diag(z^{-\frac{1}{4}},z^{\frac{1}{4}},z^{\frac{1}{4}},z^{-\frac{1}{4}})
\diag\left(
\begin{pmatrix}
1 & i \\
i & 1
\end{pmatrix},
\begin{pmatrix}
1 & i \\
i & 1
\end{pmatrix}
 \right)
 \diag\left(z^{\frac{\kappa}{2}\sigma_3},z^{\frac{\kappa-\nu}{2}\sigma_3}\right) \nonumber
 \\
 & = \frac{1}{\sqrt{2}}\diag(z^{(\frac{\kappa}{2}-\frac{1}{4})\sigma_3},z^{(\frac{\kappa-\nu}{2}+\frac{1}{4})\sigma_3})
\diag\left(
\begin{pmatrix}
1 & iz^{-\kappa} \\
iz^{\kappa} & 1
\end{pmatrix},
\begin{pmatrix}
1 & iz^{\nu-\kappa} \\
i z^{\kappa-\nu} & 1
\end{pmatrix}
 \right).
\end{align}

\item[\rm (4)] $X$ has the following local behaviors near the origin.
\begin{itemize}
\item For $\kappa>0$, $\nu>0$, $\nu\neq \kappa$
$$
X(z)=\Boh(1),\qquad z\to 0.
$$

\item For $\kappa=\nu>0$,
$$
X(z)\diag\left(1,1,1,\frac{1}{\log z}\right)=\Boh(1),\qquad z\to 0.
$$

\item For $\kappa=0$, $\nu>0$,
$$
X(z)\diag\left(\frac{1}{\log z},1,\frac{1}{\log z},\frac{1}{\log z}\right)=\Boh(1),\qquad z\to 0.
$$

\item For $\kappa=\nu=0$,
$$
X(z)\diag\left(\frac{1}{\log z},1,\frac{1}{\log z},\frac{1}{(\log z)^2}\right)=\Boh(1),\qquad z\to 0.
$$
\end{itemize}
\end{enumerate}
\end{lem}
\begin{proof}
To show the jump condition as stated in item (2), we see from \eqref{eq:YtoX} and \eqref{defjumpmatrixY} that
\begin{multline}\label{eq:jumpproof0}
X_-^{-1}(x)X_+(x)
\\
=\left\{
                    \begin{array}{ll}
                      \begin{pmatrix}
I_2 & x^{-\frac \kappa 2 \sigma_3}A_1^{-1}(x)W(x)A_2(x)x^{\frac{\kappa-\nu}{2}\sigma_3}
\\
0 & I_2
\end{pmatrix}, & \hbox{if $x>0$,}
\\
 \begin{pmatrix}
x_-^{-\frac \kappa 2 \sigma_3}A_{1,-}^{-1}(x)A_{1,+}(x)x_+^{\frac{\kappa}{2}\sigma_3} & 0
\\
0 & x_-^{\frac{\nu-\kappa}{2} \sigma_3}A_{2,-}^{-1}(x)A_{2,+}(x)x_+^{\frac{\kappa-\nu}{2}\sigma_3}
\end{pmatrix}, & \hbox{if $x<0$.}
                    \end{array}
                  \right.
\end{multline}
By \eqref{eq:Wdecomp}--\eqref{A2}, it follows from \eqref{eq:Wronskian relation} and a straightforward calculation that
\begin{equation}\label{eq:relations_A_W}
A_1^{-1}(x)w_1(x)^T=
\begin{pmatrix}
0 \\ x^{\frac{\kappa}{2}}
\end{pmatrix},
\qquad
w_2(x)A_2(x)=
\begin{pmatrix}
x^{\frac{\nu-\kappa}{2}} & 0
\end{pmatrix},
\end{equation}
so
\begin{align}\label{eq:jumpproof1}
x^{-\frac \kappa 2 \sigma_3}A_1^{-1}(x)W(x)A_2(x)x^{\frac{\kappa-\nu}{2}\sigma_3}
& =  x^{-\frac \kappa 2 \sigma_3}A_1^{-1}(x)w_1(x)^T w_2(x)A_2(x)x^{\frac{\kappa-\nu}{2}\sigma_3}
\nonumber \\
& = x^{-\frac \kappa 2 \sigma_3}
\begin{pmatrix}
0 & 0
\\
x^{\frac \nu 2} & 0
\end{pmatrix}
x^{\frac{\kappa-\nu}{2}\sigma_3}=\begin{pmatrix}
0 & 0
\\
x^{\kappa} & 0
\end{pmatrix}.
\end{align}
Similarly, by making use of \eqref{jump for y_i} and \eqref{eq:Wronskian relation}, one can check that if $x<0$,
\begin{equation*}
x_-^{-\frac \kappa 2 \sigma_3}A_{1,-}^{-1}(x)A_{1,+}(x)x_+^{\frac{\kappa}{2}\sigma_3}
=x_-^{-\frac \kappa 2 \sigma_3}
\begin{pmatrix}
e^{-\kappa \pi i} & 0
\\
-1 & e^{\kappa \pi i}
\end{pmatrix}x_+^{\frac{\kappa}{2}\sigma_3}=\begin{pmatrix}
1 & 0
\\
-|x|^\kappa & 1
\end{pmatrix}
\end{equation*}
and
\begin{align}\label{eq:jumpproof3}
&x_-^{\frac{\nu-\kappa}{2} \sigma_3}A_{2,-}^{-1}(x)A_{2,+}(x)x_+^{\frac{\kappa-\nu}{2}\sigma_3}
\nonumber
\\
&=x_-^{\frac{\nu-\kappa}{2} \sigma_3}
\begin{pmatrix}
e^{(\nu-\kappa) \pi i} & -1
\\
0 & e^{(\kappa-\nu) \pi i}
\end{pmatrix}
x_+^{\frac{\kappa-\nu}{2}\sigma_3}=\begin{pmatrix}
1 & -|x|^{\nu-\kappa}
\\
0 & 1
\end{pmatrix}.
\end{align}
Inserting \eqref{eq:jumpproof1}--\eqref{eq:jumpproof3} into \eqref{eq:jumpproof0} gives us \eqref{jump for X}.

To establish the large $z$ behavior of $X$, it suffices to derive the asymptotics of $A_1$ and $A_2$. We follow closely \cite{KMW09} and start with
known asymptotic formulas for the Bessel functions \cite[Formulas (10.40.1) and (10.40.2)]{DLMF} to obtain
\begin{equation*}
\begin{aligned}
y_{1,a}(\tau^2z) & =\frac{1}{2\sqrt{\pi}}\tau^{a+\frac{1}{2}} z^{\frac{a}{2}+\frac{1}{4}}e^{2\tau z^{\frac{1}{2}}}
\\  & \qquad  \times
\left(1-\frac{4(a+1)^2-1}{16\tau z^{\frac{1}{2}}} + \frac{ ( 4(a+1)^2-1 )( 4(a+1)^2-9 ) }{ 512 \tau^2 z } + \Boh(z^{-\frac{3}{2}})\right),
\\
y_{2,a}(\tau^2z)  &  =\frac{\sqrt{\pi}}{2}\tau^{a+\frac{1}{2}} z^{\frac{a}{2}+\frac{1}{4}}e^{-2\tau z^{\frac{1}{2}}}
\\ & \qquad  \times
\left(1+\frac{4(a+1)^2-1}{16\tau z^{\frac{1}{2}}} + \frac{ ( 4(a+1)^2-1 )( 4(a+1)^2-9 ) }{ 512 \tau^2 z } + \Boh(z^{-\frac{3}{2}})\right),
\end{aligned}
\end{equation*}
for $z\to\infty$ with $|\arg z|<\pi$ and $\tau>0$. This, together with \eqref{A1} and \eqref{derivatives of y_i}, implies that, as $z\to \infty$,
\begin{multline*}
A_1(z)=-\frac{i}{2\sqrt{\pi\tau_1}} z^{-\frac{\sigma_3}{4}} \\
\times
\left[
\begin{pmatrix}
\phantom{-}1 & i \\
-1 & i
\end{pmatrix}
+\frac{D_1}{z^{\frac{1}{2}}}
\begin{pmatrix}
\phantom{-}1 & -i \\
-1 & -i
\end{pmatrix}
+ \frac{D_2}{z}
\begin{pmatrix}
\phantom{-}1 & i \\
-1 & i
\end{pmatrix}
+\Boh(z^{-\frac{3}{2}})
 \right] e^{-2\tau_1z^{\frac{1}{2}}\sigma_3},
\end{multline*}
where
\begin{align*}
D_1& = \frac{1}{16\tau_1}
\begin{pmatrix}
4\kappa^2-1 & 0 \\
0 & 4(\kappa+1)^2-1
\end{pmatrix}, \\
D_2 & = \frac{1}{512\tau_1^2}
\begin{pmatrix}
(4\kappa^2-1)(4\kappa^2-9) & 0 \\
0 & (4(\kappa+1)^2-1)(4(\kappa+1)^2-9)
\end{pmatrix}.
\end{align*}
Using the identity
$$
\begin{pmatrix}
\phantom{-}1 & -i \\
-1 & -i
\end{pmatrix}
=
\begin{pmatrix}
\phantom{-}0 & -1 \\
-1 & \phantom{-}0
\end{pmatrix}
\begin{pmatrix}
\phantom{-}1 & i \\
-1 & i
\end{pmatrix},
$$
we further simplify the previous formula to
\begin{equation*}
A_1(z)=-\frac{i}{2\sqrt{\pi\tau_1}}z^{-\frac{\sigma_3}{4}}
\left[
I_2
+\frac{D_1}{z^{\frac{1}{2}}}
\begin{pmatrix}
\phantom{-}0 & -1 \\
-1 & \phantom{-}0
\end{pmatrix}
+ \frac{D_2}{z}
+\Boh(z^{-\frac{3}{2}})
 \right]
 \begin{pmatrix}
\phantom{-}1 & i \\
-1 & i
\end{pmatrix}
e^{-2\tau_1z^{\frac{1}{2}}\sigma_3}.
\end{equation*}
On account of the fact that
$$
\frac{z^{-\frac{\sigma_3}{4}}}{z^\frac{1}{2}}
\begin{pmatrix}
\phantom{-}0 & -1 \\
-1 & \phantom{-}0
\end{pmatrix}
=
\begin{pmatrix}
\phantom{-}0 & -z^{-1} \\
-1 & 0
\end{pmatrix}
z^{-\frac{\sigma_3}{4}},
$$
we finally arrive at
\begin{align}
A_1(z) &
=-\frac{i}{2\sqrt{\pi\tau_1}}
\left[
I_2
+D_1
\begin{pmatrix}
\phantom{-}0 & -z^{-1} \\
-1 & \phantom{-}0
\end{pmatrix}
+\Boh(z^{-1})
 \right]
 z^{-\frac{\sigma_3}{4}}
 \begin{pmatrix}
\phantom{-}1 & i \\
-1 & i
\end{pmatrix}
e^{-2\tau_1z^{\frac{1}{2}}\sigma_3}
\nonumber \\
&
=
-\frac{i}{2\sqrt{\pi\tau_1}}
\left[
\begin{pmatrix}
1 & 0 \\
-\frac{4(\kappa+1)^2-1}{16\tau_1} & 1
\end{pmatrix}
+\Boh(z^{-1})
 \right]
 z^{-\frac{\sigma_3}{4}}
 \begin{pmatrix}
\phantom{-}1 & i \\
-1 & i
\end{pmatrix}
e^{-2\tau_1z^{\frac{1}{2}}\sigma_3}
\nonumber \\
& =\frac{1}{\sqrt{2\pi\tau_1}}\left[
\begin{pmatrix}
-i & 0 \\
i\frac{4(\kappa+1)^2-1}{16} & 1
\end{pmatrix}
+\Boh(z^{-1})
\right]
z^{-\frac{\sigma_3}{4}}\frac{1}{\sqrt{2}}
\begin{pmatrix}
1 & i \\
i & 1
\end{pmatrix}
e^{-2\tau_1z^{\frac{1}{2}}\sigma_3}, \label{eq:asyA1}
\end{align}
which is valid for $z\to \infty$ along $ \mathbb C\setminus \R_-$.

In a similar way, we also obtain that if $z\to \infty$ along $ \mathbb C\setminus \R_-$,
\begin{align}
A_2(z) &  =
-\sqrt{\frac{\tau_2}{\pi}}i
\left[
\begin{pmatrix}
1 & -\frac{4(\nu-\kappa+1)^2-1}{16\tau_2} \\
0 & 1
\end{pmatrix}
+\Boh(z^{-1})
 \right]
 z^{\frac{\sigma_3}{4}}
 \begin{pmatrix}
i & -1 \\
i & \phantom{-} 1
\end{pmatrix}
e^{2\tau_2z^{\frac{1}{2}}\sigma_3}
\nonumber \\
& =
\sqrt{\frac{2\tau_2}{\pi}}
\left[
\begin{pmatrix}
1 & \frac{4(\nu-\kappa+1)^2-1}{16}i \\
0 & -i
\end{pmatrix}
+\Boh(z^{-1})
 \right]
 \tau_2^{\frac{\sigma_3}{2}}
 z^{\frac{\sigma_3}{4}}
 \frac{1}{\sqrt{2}}
 \begin{pmatrix}
1 & i \\
i &  1
\end{pmatrix}
e^{2\tau_2z^{\frac{1}{2}}\sigma_3}. \label{eq:asyA2}
\end{align}
A combination of \eqref{eq:YtoX}, \eqref{eq:asyA1} and \eqref{eq:asyA2} then gives us \eqref{eq:Xinfty}.

Finally, it follows from the known behavior of the modified Bessel functions near the origin (cf. \cite[Formulas 10.30.1--10.30.3]{DLMF})
that, as $z\to 0$,
\begin{align*}
y_1(z)&\sim \frac{1}{\Gamma(a+2)}z^{a+1}, &&\qquad
y_1'(z)\sim \frac{1}{\Gamma(a+1)}z^{a}, \\
y_2(z)&\sim\frac{1}{2}\Gamma(a+1),
&&\qquad y_2'(z)\sim\left\{
\begin{array}{ll}
-\frac{1}{2}\Gamma(a), &  a >0, \\
\frac{1}{2}\log(z), & a=0, \\
-\frac{1}{2}\Gamma(-a)z^{a}, &  a<0.
\end{array}
\right.
\end{align*}
The behavior of $X$ near the origin in item (4) then follows from a straightforward calculation.

This completes the proof of Lemma \ref{lem:RHpforX}.
\end{proof}

\section{Second transformation \texorpdfstring{$X \rightarrow T$}{X to T}}
With the $\lambda$-functions given in \eqref{def:lambda_functions1}--\eqref{def:lambda_functions4}, we define the second transformation $ X \rightarrow T$ by
\begin{equation}\label{eq:XtoT}
T(z)=C_T X(z) \diag\left(e^{n\lambda_1(z)},e^{n\lambda_2(z)},e^{n\lambda_3(z)},e^{n\lambda_4(z)}\right),
\end{equation}
where
$$
C_T=\left(I_4-2nc_1i E_{21}+2nc_3i E_{34}\right)\diag\left(e^{-n\theta_1},e^{-n\theta_1},e^{-n\theta_3},e^{-n\theta_3}\right)
$$
with the constants $c_1,c_3,\theta_1,\theta_3$ as in Proposition \ref{prop:asylambda}. Then, $T$ satisfies the following RH problem.

\begin{lem}\label{lem:rhpforT}
The function $T$ defined in \eqref{eq:XtoT} has the following properties:
\begin{enumerate}
\item[\rm (1)] $T$ is defined and analytic in $ \mathbb{C} \setminus \R$.
\item[\rm (2)] For $x\in \mathbb{R}$, $T$ satisfies the jump condition
\begin{equation}\label{eq:jumpT}
T_+(x)=T_-(x)J_{T}(x),
\end{equation}
where
\begin{equation}\label{eq:jumps_T_2}
J_T(x)=
\begin{dcases}
I_4+x^{\kappa}e^{-n\phi_2(x)}E_{23},& x>p, \\
\diag\left(1,e^{n\phi_{2,+}(x)},e^{n\phi_{2,-}(x)},1\right)+x^{\kappa}E_{23},& x\in \Delta_2, \\
\diag\left(1,1,e^{n\phi_{3,+}(x)},e^{n\phi_{3,-}(x)}\right)
 -|x|^\kappa e^{n\phi_1(x)}E_{21}-|x|^{\nu-\kappa}E_{34},
& x\in \Delta_3 \setminus \overline{\Delta_1}, \\
 \begin{multlined}
  \hspace{-1.5mm}
\diag\left(e^{-n\phi_{1,-}(x)},e^{-n\phi_{1,+}(x)},e^{n\phi_{3,+}(x)},e^{n\phi_{3,-}(x)}\right)\\
 -|x|^{\kappa}E_{21}-|x|^{\nu-\kappa}E_{34},
 \end{multlined}
& x \in \Delta_1,
\end{dcases}
\end{equation}
and where the $\phi$-functions are defined in \eqref{def:phi_functions1}--\eqref{def:phi_functions3}.

\item[\rm (3)] As $z\to\infty$, we have
\begin{equation}\label{eq:asyT}
T(z)=(I_4+\Boh(z^{-1}))\mathcal B(z),
\end{equation}
where the function $\mathcal B$ is given in \eqref{def:Bz}.

\item[\rm (4)]The matrix $T$ has the same behavior as $X$ as $z\to 0$.
\end{enumerate}
\end{lem}
\begin{proof}
To show the jump condition \eqref{eq:jumpT}, it is readily seen from \eqref{eq:XtoT} and \eqref{jump for X} that
\begin{equation*}
 J_T(x)=
 \begin{cases}
  \diag\left(e^{n(\lambda_{j,+}(x)-\lambda_{j,-}(x))}\right)_{j\leq 4}+x^{\kappa}e^{n(\lambda_{3,+}(x)-\lambda_{2,-}(x))}E_{23}, & \hbox{if $x>0$,} \\
 &
 \\
  \begin{multlined}
  \hspace{-1.5mm}
  \diag\left(e^{n(\lambda_{j,+}(x)-\lambda_{j,-}(x))}\right)_{j\leq 4} \\ -|x|^\kappa E_{21}e^{n(\lambda_{1,+}(x)-\lambda_{2,-}(x))}-|x|^{\nu-\kappa}e^{n(\lambda_{4,+}(x)-\lambda_{3,-}(x))}E_{34},
\end{multlined}
  & \hbox{if $x<0$}.
\end{cases}
  \end{equation*}
This formula simplifies further to \eqref{eq:jumps_T_2} with the aid of Proposition \ref{prop:lamdaphi}.

For the asymptotic behavior of $T$ near infinity, we observe from \eqref{eq:Xinfty} and Proposition \ref{prop:asylambda} that, as $z\to \infty$,
\begin{align*}
 &X(z) \diag\left(e^{n\lambda_1(z)},e^{n\lambda_2(z)},e^{n\lambda_3(z)},e^{n\lambda_4(z)}\right)
 \\
 & =\left(
 \diag\left(e^{n\theta_1(z)},e^{n\theta_1(z)},e^{n\theta_3(z)},e^{n\theta_3(z)}\right)+\Boh(z^{-1})\right)\mathcal{B}(z)
 \\
  &\quad \times
 \left(\diag\left( 1+\frac{2nc_1}{\sqrt{z}}, 1-\frac{2nc_1}{\sqrt{z}}, 1+\frac{2nc_3}{\sqrt{z}}, 1-\frac{2nc_3}{\sqrt{z}}\right) + \Boh(z^{-1})I_4\right).
\end{align*}
By moving the last diagonal matrix in the above formula to the left, it follows that
\begin{equation*}
 X(z) \diag\left(e^{n\lambda_1(z)},e^{n\lambda_2(z)},e^{n\lambda_3(z)},e^{n\lambda_4(z)}\right)
 =C_T^{-1}(
 I_4+\Boh(z^{-1}))\mathcal{B}(z).
\end{equation*}
This, together with \eqref{eq:XtoT}, implies \eqref{eq:asyT}.

Finally, since each of the $\lambda$-functions is bounded near the origin, it is clear that the matrix $T$ has the same behavior as $X$ as $z\to 0$.

This completes the proof of Lemma \ref{lem:rhpforT}.
\end{proof}


\section{Third transformation \texorpdfstring{$T\to S$}{T to S}}

The third transformation involves the so-called lens opening. The goal of this step is to convert the highly oscillatory jumps
into a more convenient form on the original contours while creating extra jumps tending to the identity matrices exponentially fast
on the new contours. This transformation is based on the following classical factorizations:
$$
\begin{pmatrix}
e^{-u} & v \\
0 & e^{u}
\end{pmatrix}
=
\begin{pmatrix}
1 & 0 \\
v^{-1}e^{u} & 1
\end{pmatrix}
\begin{pmatrix}
0 & v \\
-v^{-1} & 0
\end{pmatrix}
\begin{pmatrix}
1 & 0 \\
v^{-1}e^{-u} & 1
\end{pmatrix}
$$
and
$$
\begin{pmatrix}
e^{-u} & 0 \\
v & e^{u}
\end{pmatrix}
=
\begin{pmatrix}
1 & v^{-1}e^{-u} \\
0 & 1
\end{pmatrix}
\begin{pmatrix}
0 & -v^{-1} \\
v & 0
\end{pmatrix}
\begin{pmatrix}
1 & v^{-1}e^{u} \\
0 & 1
\end{pmatrix}.
$$

Note that the jump matrices in \eqref{eq:jumps_T_2} can be viewed as $2\times 2$ block matrices, the factorizations above can be easily applied.
For instance, if $x\in \Delta_2$, it follows from \eqref{eq:phi2_pm} that
\begin{align}\label{eq:JTdecomp1}
J_T(x) & = \diag\left(1,
\begin{pmatrix}
e^{n\phi_{2,+}(x)} & x^{\kappa} \\
0 & e^{n\phi_{2,-}(x)}
\end{pmatrix},
1\right)
\nonumber \\
& =
\diag\left( 1,
\begin{pmatrix}
1 & 0 \\
x^{-\kappa}e^{n\phi_{2,-}(x)} & 1
\end{pmatrix}
\begin{pmatrix}
0 & x^{\kappa} \\
-x^{-\kappa} & 0
\end{pmatrix}
\begin{pmatrix}
1 & 0 \\
x^{-\kappa}e^{n\phi_{2,+}(x)} & 1
\end{pmatrix},
1\right), \nonumber  \\
& = (I_4+x^{-\kappa}e^{n\phi_{2,-}(x)}E_{32})
\diag\left( 1,
\begin{pmatrix}
0 & x^{\kappa} \\
-x^{-\kappa} & 0
\end{pmatrix},
1\right)
(I_4+x^{-\kappa}e^{n\phi_{2,+}(x)}E_{32}).
\end{align}
In a similar spirit, we use \eqref{eq:phi3_pm} to see that for $x\in \Delta_3 \setminus \overline{\Delta_1}$,
\begin{align}\label{eq:JTdecomp2}
J_T(x) & =   \diag\left(I_2,
\begin{pmatrix}
e^{n\phi_{3,+}(x)} & -|x|^{\nu-\kappa} \\
0 & e^{n\phi_{3,-}(x)}
\end{pmatrix}
\right) -|x|^{\kappa} e^{n\phi_1(x)}E_{21}
\nonumber \\
& =
(I_4- \mathfrak c_{\kappa-\nu}x^{\kappa-\nu}_- e^{n\phi_{3,-}(x)}E_{43})
\left[
\diag\left(I_2,
\begin{pmatrix}
0 & -|x|^{\nu-\kappa} \\
|x|^{\kappa-\nu} & 0
\end{pmatrix}
\right)
-|x|^\kappa e^{n\phi_1(x)}E_{21} \right] \nonumber  \\
& \quad \times (I_4-\mathfrak c_{\nu-\kappa}x^{\kappa-\nu}_+e^{n\phi_{3,+}(x)}E_{43}),
\end{align}
where $\mathfrak c_{\alpha}$ is defined in \eqref{def:calpha},
and, finally, using \eqref{eq:phi1_pm}, we obtain that for $x\in \Delta_1$,
\begin{align}\label{eq:JTdecomp3}
J_T(x)
 & =
\diag\left(
\begin{pmatrix}
e^{-n\phi_{1,-}(x)} & 0 \\
-|x|^{\kappa} & e^{-n\phi_{1,+}(x)}
\end{pmatrix},
\begin{pmatrix}
e^{n\phi_{3,+}(x)} & -|x|^{\nu-\kappa} \\
0 & e^{n\phi_{3,-}(x)}
\end{pmatrix}
\right)
\nonumber  \\
%
& =
(I_4-\mathfrak c_{\kappa-\nu}x^{\kappa-\nu}_-e^{n\phi_{3,-}(x)}E_{43})
(I_4-\mathfrak c_{-\kappa}x^{-\kappa}_-e^{-n\phi_{1,-}(x)}E_{12})
\nonumber \\
& \quad \times
\diag\left(
\begin{pmatrix}
0 &  |x|^{-\kappa} \\
 -|x|^{\kappa} & 0
\end{pmatrix},
\begin{pmatrix}
0 & -|x|^{\nu-\kappa} \\
|x|^{\kappa-\nu} & 0
\end{pmatrix}
\right) \nonumber  \\
& \quad  \times
(I_4-\mathfrak c_{\kappa}x^{-\kappa}_+e^{-n\phi_{1,+}(x)}E_{12})
(I_4-\mathfrak c_{\nu-\kappa}x^{\kappa-\nu}_+e^{n\phi_{3,+}(x)}E_{43}).
\end{align}

For each $k=1,2,3$, we set simply connected domains $\mathcal L_k^\pm$ (the lenses) on the $\pm$-side of $\Delta_k$, with oriented boundaries $\partial \mathcal L_k^\pm\cup\Delta_k$
as shown in Figure~\ref{fig:lenses}. Moreover, it is required that
\begin{equation}\label{eq:choice_lenses}
\partial\mathcal L_1^\pm\subset \mathcal L_3^\pm \quad \textrm{and} \quad
\overline{\mathcal L_k^\pm}\setminus \Delta_k \subset \mathcal G_k,
\end{equation}
where the open neighborhood $\mathcal G_k$ of the interval $\Delta_k$ is given in Proposition \ref{prop:inequalities_phi}.

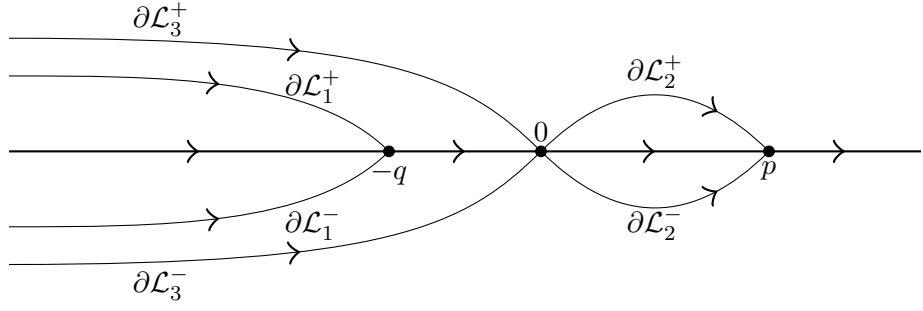
\begin{figure}[t]
\begin{center}
\begin{tikzpicture}
%
\filldraw [black] (3,3) circle (2pt)
                   (5,3) circle (2pt)
                   (8,3) circle (2pt);
\node[below] (q) at (3,3) {$-q$};
\node[above] (zero) at (5,3) {$0$};
\node[below] (p) at (8,3) {$p$};
\draw[thick,->-=.5] (5,3)--(8,3);
\draw[thick,->-=.5] (-2,3)--(3,3);
\draw[thick,->-=.5] (3,3)--(5,3);
\draw[thick,->-=.5] (8,3)--(10,3);
\draw[->-=.75] (5,3) .. controls (6,4) and (7,4) .. (8,3);
\draw[->-=.75] (5,3) .. controls (6,2) and (7,2) .. (8,3);
\draw[-<-=.5] (3,3) .. controls (2,4) and (0,4) .. (-2,4);
\draw[-<-=.5] (3,3) .. controls (2,2) and (0,2) .. (-2,2);
\draw[-<-=.5] (5,3) .. controls (4,4) and (3,4.5) .. (-2,4.5);
\draw[-<-=.5] (5,3) .. controls (4,2) and (3,1.5) .. (-2,1.5);
\node[above] at (2,3.5) {$\partial\mathcal L_1^+$};
\node[below] at (2,2.4) {$\partial\mathcal L_1^-$};
\node[above] at (6.5,3.7) {$\partial\mathcal L_2^+$};
\node[below] at (6.5,2.4) {$\partial\mathcal L_2^-$};
\node[above] at (0,4.4) {$\partial\mathcal L_3^+$};
\node[below] at (0,1.6) {$\partial\mathcal L_3^-$};
\end{tikzpicture}
\end{center}
\caption{The lenses used for the transformation $T\mapsto S$.}
\label{fig:lenses}
\end{figure}

Based on the decompositions of $J_T$ given in \eqref{eq:JTdecomp1}--\eqref{eq:JTdecomp3} and also on the lenses just defined, the third transformation reads
\begin{equation}\label{eq:TtoS}
S(z) = T(z)
\begin{cases}
(I_4\mp z^{-\kappa}e^{n\phi_2(z)}E_{32}), & z\in \mathcal L_2^{\pm},\\
(I_4\pm \mathfrak c_{\nu-\kappa} z^{\kappa-\nu}e^{n\phi_3(z)}E_{43}), & z\in \mathcal L_3^{\pm}\setminus \mathcal L_1^{\pm}, \\
(I_4\pm \mathfrak c_{\nu-\kappa} z^{\kappa-\nu}e^{n\phi_3(z)}E_{43})
(I_4\pm \mathfrak c_\kappa z^{-\kappa}e^{-n\phi_1(z)}E_{12}),
& z \in \mathcal L_1^{\pm},\\
I_4, & z \mbox{ outside the lenses}.
\end{cases}
\end{equation}
Since both $\kappa$ and $\nu$ are integers, it is easily seen that
$$\mathfrak c_{\nu-\kappa}=\mathfrak c_{\kappa-\nu},\qquad \mathfrak c_{\kappa}=\mathfrak c_{-\kappa}.$$
Also note that the factors of the form $(I_4+(\ast) E_{12})$ and $(I_4+ (\ast) E_{43})$ appearing above commute, it is then straightforward to check that
the matrix $S$ satisfies the following RH problem.

\begin{rhp}
The function $S$ defined in \eqref{eq:TtoS} has the following properties:
\begin{enumerate}
\item[\rm (1)] $S$ is defined and analytic in $ \mathbb{C} \setminus\Gamma_S$, where
\begin{equation}\label{def:gammaS}
\Gamma_S:=\R\cup \left( \bigcup\limits_{j=1}^3\partial \mathcal L_j^\pm \right).
\end{equation}
\item[\rm (2)] For $z\in \Gamma_S$, $S$ satisfies the jump condition
$$S_+(z)=S_-(z)J_S(z),$$
where
\begin{equation}\label{eq:jumps_S}
J_S(z)=
\begin{cases}
J_T(z)=I_4+x^{\kappa}e^{-n\phi_2(z)}E_{23},& z\in (p,+\infty), \\
I_4+ z^{-\kappa}e^{n\phi_2(z)}E_{32}, & z\in \partial\mathcal L_2^{\pm},\\
I_4- \mathfrak c_{\nu-\kappa} z^{\kappa-\nu}e^{n\phi_3(z)}E_{43}, & z\in \partial\mathcal L_3^{\pm}, \\
I_4- \mathfrak c_{\kappa} z^{-\kappa}e^{-n\phi_1(z)}E_{12}, & z \in \partial\mathcal L_1^{\pm},\\
\diag\left( 1,
\begin{pmatrix}
0 & z^{\kappa} \\
-z^{-\kappa} & 0
\end{pmatrix},
1\right), & z\in \Delta_2, \\
\diag\left(I_2,
\begin{pmatrix}
0 & -|z|^{\nu-\kappa} \\
|z|^{\kappa-\nu} & 0
\end{pmatrix}
\right)
-|z|^{\kappa}e^{n\phi_1(z)}E_{21} , & z\in \Delta_3 \setminus \overline{\Delta_1}, \\
\diag\left(
\begin{pmatrix}
0 &  |z|^{-\kappa} \\
 -|z|^{\kappa} & 0
\end{pmatrix},
\begin{pmatrix}
0 & -|z|^{\nu-\kappa} \\
|z|^{\kappa-\nu} & 0
\end{pmatrix}
\right), & z\in \Delta_1.
\end{cases}
\end{equation}

\item[\rm (3)] As $z\to\infty$, we have
\begin{equation*}
S(z)=(I_4+\Boh(z^{-1}))\mathcal B(z),
\end{equation*}
where the function $\mathcal B$ is given in \eqref{def:Bz}.

\item[\rm (4)]
As $z\to 0$, $S$ has the same behavior as $T$ provided $z\to 0$ outside the lenses that end in $0$.

\end{enumerate}
\end{rhp}

\section{Global parametrix}

By \eqref{eq:choice_lenses}, \eqref{eq:jumps_S} and Proposition \ref{prop:inequalities_phi}, it is easily seen that, as $n\to \infty$,
$$
J_S(z)=I_4+\boh(1), \qquad \; z\in \cup_{j=1}^3\partial \mathcal L_j^{\pm}\cup (p,+\infty),
$$
uniformly valid for $z$ bounded away from the endpoints of the sets $\Delta_k$, $k=1,2,3$. This, together with the second inequality in \eqref{eq:inequalities_phi_off_support}, leads us to the following model RH problem, also called {\it global parametrix} RH problem.

\begin{rhp}\label{rhp:global_parametrix}
We look for a $4\times 4$ matrix-valued function
$G $
satisfying the following properties:
\begin{enumerate}
\item[\rm (1)] $G$ is defined and analytic in $\C\setminus (-\infty,p]$.
\item[\rm (2)] $G$ satisfies the jump condition
\begin{equation*}
G_{+}(x) = G_{-}(x)J_G(x),
\end{equation*}
where
\begin{equation}\label{eq:jumps_global_parametrix}
J_G(x)=
\begin{cases}
\diag\left(
\begin{pmatrix}
0 &  |x|^{-\kappa} \\
 -|x|^{\kappa} & 0
\end{pmatrix},
\begin{pmatrix}
0 & -|x|^{\nu-\kappa} \\
|x|^{\kappa-\nu} & 0
\end{pmatrix}
\right),
& x\in \Delta_1, \\
\diag\left( 1,
\begin{pmatrix}
0 & x^{\kappa} \\
-x^{-\kappa} & 0
\end{pmatrix},
1\right),
& x\in \Delta_2, \\
\diag\left(I_2,
\begin{pmatrix}
0 & -|x|^{\nu-\kappa} \\
|x|^{\kappa-\nu} & 0
\end{pmatrix}
\right),
& x \in \Delta_3\setminus \overline{\Delta_1}.
\end{cases}
\end{equation}

\item[\rm (3)] As $z\to\infty$ away from $\R_-$, we have
\begin{equation}\label{eq:asympt_behavior_global_param}
    G(z) = (I_4+\Boh(z^{-1}))\mathcal B(z),
\end{equation}
where $\mathcal B$ is as in \eqref{def:Bz}.

\end{enumerate}
\end{rhp}

Note that we are not imposing any endpoint behaviors for $G$, so the solution to RH problem \ref{rhp:global_parametrix} might not be unique. Nevertheless, we will construct some $G$ explicitly that will be enough to finish the further asymptotic analysis. The construction relies on the uniformization map of the Riemann surface described in Section \ref{sec:rational_parametrization}.

\subsection{Construction of the global parametrix for  \texorpdfstring{$\kappa=\nu=0$}{k=v=0}}

In this section, we will solve the model RH problem \ref{rhp:global_parametrix} with $\kappa=\nu=0$, whose solution will be denoted by $G_0$. The basic idea is to lift the original RH problem to the Riemann surface $\mathcal{R}$, and then transform the matrix-valued RH problem into several scalar RH problems on the $t$-plane with the aid of the rational parametrization \eqref{eq:rational_parametrization}.

To proceed, let $t_k=t_k(z)$, $k=1,2,3,4$, be the inverse of the map $z=z(t)$ in \eqref{def:zt} restricted to $\mathcal R_k$, i.e.,
\begin{equation*}
t_k: \mathcal R_k \rightarrow \mathbb{C}.
\end{equation*}
We then have the following proposition.
\begin{prop}\label{prop:G0}
A solution of the model RH problem \ref{rhp:global_parametrix} with $\kappa=\nu=0$ is given by
\begin{equation}\label{eq:G0}
G_0(z)=\begin{pmatrix}
G_{k}(t_j(z))
\end{pmatrix}_{k,j=1}^4,
\end{equation}
where
\begin{align}
G_1(t)&=\mathfrak e_1
\left(\frac{t-\alpha^2}{t-t_-}\right)^{\frac{1}{2}}\frac{(t-\beta^2)^{\frac{3}{2}}}{(t-t_+)^{\frac{1}{2}}}, \label{def:G1}\\
G_2(t)& = \mathfrak e_2  \; \frac{ t\; (t-\beta^2)^\frac{3}{2}}{(t-t_+)^{\frac{1}{2}}}\left((t-\alpha^2)(t-t_-)\right)^{-\frac{1}{2}} , \\
G_3(t)& = \mathfrak e_3 \;  t \; (t-\alpha^2)\left( \frac{t-\alpha^2}{t-t_-}\right)^{\frac{1}{2}}(t-t_+)^{-\frac{1}{2}}(t-\beta^2)^{-\frac{1}{2}}, \\
G_4(t) & =\mathfrak e_4 \; (t-\alpha^2)\left( \frac{t-\alpha^2}{t-t_-}\right)^{\frac{1}{2}}(t-t_+)^{-\frac{1}{2}}(t-\beta^2)^{\frac{1}{2}}. \label{def:G4}
\end{align}
Here, $t_{\pm}$ is given in \eqref{def:tpm}, the branch cut for the root of $(t-\alpha^2)^\bullet(t-t_-)^\bullet$ is taken along $\gamma_1^+$, the branch cuts of $(t-t_+)^{\frac12}$ and $(t-\beta^2)^{\frac12}$ are taken along $\gamma_2^-$ and $\gamma_3^-$, respectively, and $\mathfrak{e}_1,\mathfrak{e}_2,\mathfrak{e}_3,\mathfrak{e}_4$ are explicitly computable non-zero constants.
\end{prop}
\begin{proof}
Suppose that
\begin{equation*}
G_0(z)=\begin{pmatrix}
g_{k,j}(z)
\end{pmatrix}_{k,j=1}^4
\end{equation*}
solves the model RH problem \ref{rhp:global_parametrix} with $\kappa=\nu=0$. We lift the RH problem to the Riemann surface $\mathcal{R}$ by treating each entry $g_{k,j}(z)$ of the $k$-th row of $G_0$ as defined on the sheet $\mathcal{R}_j$ of $\mathcal{R}$ and define
\begin{equation}\label{eq:definition_G0_2}
g_{k}: \mathcal{R} \to \mathbb{C}, \qquad  \restr{g_k}{\mathcal R_j}=g_{k,j}, \qquad j,k=1,2,3,4.
\end{equation}
It is then easily seen that the RH problem for $G_0$ is equivalent to the following RH problem on $\mathcal{R}$.
\begin{rhp}
For $k=1,2,3,4$, the function $g_k$ defined in \eqref{eq:definition_G0_2} has the following properties:
\begin{enumerate}[\rm (1)]
\item $g_k$ is analytic in $\mathcal{R}\setminus \Gamma_g$, where
\begin{equation*}
\Gamma_g: = H(\gamma_1^+)\cup H(\gamma_2^-)\cup H(\gamma_3^-)
\end{equation*}
with $H$ being the rational parametrization \eqref{eq:rational_parametrization}. Here, each of the contours $H(\gamma_1^+)$, $H(\gamma_2^-)$ and $H(\gamma_3^-)$ is a real interval on $\mathcal R$ with the orientation taken from the left to the right.

\item $g_k$ satisfies the jump condition
$$
g_{k,+}(z)=-g_{k,-}(z),\qquad z\in \Gamma_g.
$$

\item $g_k$ has the following large $z$ asymptotic behaviors.

 \begin{itemize}
   \item As $z\to \infty$ along $\mathcal R_1$,
$$
g_1(z)=\frac{ z^{-\frac{1}{4} } }{ \sqrt{2} }(1+\Boh(z^{-\frac{1}{2} } )),  \quad g_2(z)=\frac{i z^{\frac{1}{4}} }{ \sqrt{2} }(1+\Boh(z^{-1})), \quad g_3(z)=g_4(z)=\Boh(z^{-\frac{3}{4}}).
$$

   \item As $z\to \infty$ along $\mathcal R_2$,
$$
g_1(z)=\frac{ iz^{-\frac{1}{4} } }{ \sqrt{2} }(1+\Boh(z^{-\frac{1}{2} } )),  \quad g_2(z)=\frac{z^{\frac{1}{4}} }{ \sqrt{2} }(1+\Boh(z^{-1})), \quad g_3(z)=g_4(z)=\Boh(z^{-\frac{3}{4}}).
$$
\item As $z\to \infty$ along $\mathcal R_3$,
$$
g_1(z)=g_2(z)=\Boh(z^{-\frac{3}{4}}),\quad
g_3(z)=\frac{ z^{\frac{1}{4} } }{ \sqrt{2} }(1+\Boh(z^{-1 } )),  \quad g_4(z)=\frac{iz^{-\frac{1}{4}} }{ \sqrt{2} }(1+\Boh(z^{-\frac{1}{2}})).
$$
\item As $z\to \infty$ along $\mathcal R_4$,
$$
g_1(z)=g_2(z)=\Boh(z^{-\frac{3}{4}}),\quad
g_3(z)=\frac{ i z^{\frac{1}{4} } }{ \sqrt{2} }(1+\Boh(z^{-1 } )),  \quad g_4(z)=\frac{z^{-\frac{1}{4}} }{ \sqrt{2} }(1+\Boh(z^{-\frac{1}{2}})).
$$
 \end{itemize}

\end{enumerate}
\end{rhp}

Using the rational parametrization \eqref{eq:rational_parametrization}, we further transfer the above RH problem for $g_k$ to a scalar RH problem on the $t$-complex plane
by setting
\begin{equation}\label{def:Gk}
G_k(t)=g_k(z(t)).
\end{equation}
The RH problem for $g_k$ is then equivalent to the following RH problem for $G_k$.

\begin{rhp}\label{rhp:Gk}
For $k=1,2,3,4$, the function $G_k$ defined in \eqref{def:Gk} has the following properties:
\begin{enumerate}[\rm (1)]
\item $G_k$ is analytic in $\mathbb{C} \setminus \Gamma_G$, where
\begin{equation*}
\Gamma_G=H^{-1}(\Gamma_g)=\gamma_1^+\cup\gamma_2^-\cup\gamma_3^-.
\end{equation*}
\item $G_k$ satisfies the jump condition
$$G_{k,+}(t)=-G_{k,-}(t),\qquad t\in \Gamma_G.$$
\item As $t\to \alpha^2$, we have
\begin{align*}
G_1(t) & =\frac{1}{\sqrt{2}}\left(\frac{\alpha}{\beta^2-\alpha^2}\right)^{-\frac{1}{2}}(\alpha^2-t)^{\frac{1}{2}}(1+\Boh(t-\alpha^2)), \\
G_2(t) & =\frac{i}{\sqrt{2}}\left(\frac{\alpha}{\beta^2-\alpha^2}\right)^{\frac{1}{2}}(\alpha^2-t)^{-\frac{1}{2}}(1+\Boh(t-\alpha^2)), \\
G_k(t) & =\Boh((t-\alpha^2)^{\frac{3}{2}}),\qquad k=3,4.
\end{align*}
\item As $t\to \beta^2$, we have
\begin{align*}
G_k(t) & =(\Boh(t-\beta^2)^{\frac{3}{2}}),\qquad k=1,2, \\
G_3(t) & =\frac{1}{\sqrt{2}}\left(\frac{\alpha}{\beta^2-\alpha^2}\right)^{\frac{1}{2}}(\beta^2-t)^{-\frac{1}{2}}(1+\Boh(t-\beta^2)), \\
G_4(t) & =\frac{i}{\sqrt{2}}\left(\frac{\alpha}{\beta^2-\alpha^2}\right)^{-\frac{1}{2}}(\beta^2-t)^{\frac{1}{2}}(1+\Boh(t-\beta^2)).
\end{align*}
\end{enumerate}
\end{rhp}

It is straightforward to check that the function $G_k$ defined in \eqref{def:G1}--\eqref{def:G4} with specified branch cuts satisfies the RH problem \ref{rhp:Gk}. In particular, the constants $\mathfrak{e}_\bullet$ are determined by the explicit leading coefficients given in items (3) and (4) of the above RH problem.

This completes the proof of Proposition \ref{prop:G0}.
\end{proof}

From \eqref{def:G1}--\eqref{def:G4}, it follows that
\begin{equation*}
G_{k}(t)=\left\{
           \begin{array}{ll}
             \mathfrak e_k t (1+\Boh(t^{-1})), & \hbox{ $t \to \infty$,} \\
             \Boh((t-t_+)^{-\frac12}), & \hbox{ $t \to t_+$,} \\
            \Boh((t-t_-)^{-\frac12}), & \hbox{ $t \to t_-$.}
           \end{array}
         \right.
\end{equation*}
This, together with \eqref{eq:G0} and Proposition \ref{prop:zt}, implies that the following rough estimate of $G_0$ near the endpoints of the jump contours:
\begin{equation}\label{eq:asymptotics_global_parametrix_zero_zero_param}
G_0(z)=\left\{
         \begin{array}{ll}
           \Boh(z^{-\frac{1}{3}}), & \hbox{$z \to 0$,} \\
           \Boh((z-p)^{-\frac14}), & \hbox{$z \to p$,} \\
           \Boh((z+q)^{-\frac14}), & \hbox{$z \to -q$.}
         \end{array}
       \right.
\end{equation}

\subsection{Construction of the global parametrix for general \texorpdfstring{$\kappa$ and $\nu$}{k and v}}
With the aid of $G_0$ in \eqref{eq:G0}, we could construct the global parametrix for general parameters $\kappa$ and $\nu$. To state the result,
let us define
\begin{equation}\label{def:twologs}
\log(t-\alpha^2): \mathbb{C} \setminus ( \gamma_1^+\cup(-\infty,t_-] ) \to \mathbb{C},
\quad
\log(t-\beta^2): \mathbb{C} \setminus  \gamma_3^- \to \mathbb{C},
\end{equation}
where both the branches are chosen to be purely real for large positive values of $t$, and further set
\begin{equation}\label{def:Fk}
F_k(z)=e^{-\kappa \log(t_k(z)-\alpha^2)-(\kappa-\nu)\log(t_k(z)-\beta^2)},\qquad z\in \mathcal R_k,\qquad k=1,2,3,4.
\end{equation}

\begin{prop}\label{prop:GlobCons}
A solution of the model RH problem \ref{rhp:global_parametrix} is given by
\begin{equation}\label{eq:G}
G(z)=\diag(\mathfrak f_1 ,\mathfrak f_1,\mathfrak f_3,\mathfrak f_3) G_0(z) \diag\left(\mathfrak c_{\kappa}F_1(z),F_2(z)e^{-\kappa \log z},F_3(z),\mathfrak c_{\kappa-\nu}F_4(z)e^{(\nu-\kappa)\log z}\right),
\end{equation}
where $G_0$ given in \eqref{eq:G0} solves the RH problem \ref{rhp:global_parametrix} with $\kappa=\nu=0$, the function $F_k$, $k=1,2,3,4$, is defined in \eqref{def:Fk}, the branch cut of $\log z$ is taken along the negative real axis, and $\mathfrak{f}_1,\mathfrak{f}_3$ are explicitly computable non-zero constants.
\end{prop}

\begin{proof}
By the definition \eqref{def:twologs}, it is easily seen that the maps
$$
\mathcal R_k \ni z \mapsto \log (t_k(z)-\alpha^2), \; \log(t_k(z)-\beta^2),
$$
satisfy the following boundary relations:
\begin{itemize}
  \item if $x \in \Delta_1$,
\begin{align*}
(\log (t_1(z)-\alpha^2))_+ -(\log(t_2(z)-\alpha^2))_- & =-2\pi i,\\
(\log (t_1(z)-\alpha^2))_- -(\log(t_2(z)-\alpha^2))_+ & =0,
\end{align*}
  \item if $x\in \Delta_3$,
\begin{align*}
(\log (t_3(z)-\beta^2))_+ -(\log(t_4(z)-\beta^2))_-&=0,\\
(\log (t_3(z)-\beta^2))_- -(\log(t_4(z)-\beta^2))_+&=2\pi i,
\end{align*}
  \item if $x\in \Delta_3 \setminus \overline {\Delta_1}$,
\begin{align*}
 (\log(t_2(z)-\alpha^2))_+-(\log(t_2(z)-\alpha^2))_-=-2\pi i,
\end{align*}
\end{itemize}
and are otherwise analytic in their domains of definition. As a consequence, the function
\begin{equation*}
F: \mathcal{R} \to \mathbb{C}, \qquad  \restr{F}{\mathcal R_k}=F_{k}, \qquad k=1,2,3,4,
\end{equation*}
with $F_k$ given in \eqref{def:Fk} extends to a meromorphic function on $\mathcal R$, and it is easy to check that the function $G$ defined in \eqref{eq:G} satisfies the jump condition \eqref{eq:jumps_global_parametrix}.

Finally, in virtue of the expansions in \eqref{eq:asymptotics_rational_parametrization}, we have that, as $z\to \infty$,
\begin{align*}
F_{1}(z) & = \frac{\mathfrak c_{-\kappa}}{\mathfrak f_1} z^{\frac{\kappa}{2}}(1+\Boh(z^{-1/2})), \qquad &&F_{2}(z)  = \frac{1}{\mathfrak f_1} z^{\frac{\kappa}{2}}(1+\Boh(z^{-1/2})) \\
F_{3}(z) & =\frac{1}{\mathfrak f_3} z^{\frac{\kappa-\nu}{2}}(1+\Boh(z^{-1/2})), \qquad && F_{4}(z)  =\frac{\mathfrak c_{\nu-\kappa}}{\mathfrak f_3} z^{\frac{\kappa-\nu}{2}}(1+\Boh(z^{-1/2})),
\end{align*}
for some non-zero constants $\mathfrak f_1$, $\mathfrak f_3$, which implies the large $z$ asymptotics stated in \eqref{eq:asympt_behavior_global_param}.

This completes the proof of Proposition \ref{prop:GlobCons}.
\end{proof}

By Proposition \ref{prop:zt}, it is also readily seen that
\begin{equation*}
G(z)=\left\{
         \begin{array}{ll}
           \Boh((z-p)^{-\frac14}), & \hbox{$z \to p$,} \\
           \Boh((z+q)^{-\frac14}), & \hbox{$z \to -q$.}
         \end{array}
       \right.
\end{equation*}
The local behavior of $G$ near the origin, however, is crucial in our further analysis. By setting
\begin{equation}\label{def:matrix_Upm}
\mathcal U^+ =
\widehat{\mathcal U}
\diag(\omega^{\frac{\kappa+\nu}{2}\sigma_3},1),
\qquad
\mathcal U^-=
\mathcal U^+ \diag\left(
\begin{pmatrix}
0 & -1 \\
1 & 0
\end{pmatrix},
1
\right),
\end{equation}
where
\begin{equation}\label{def:calU}
\widehat{\mathcal U }= \begin{pmatrix}
\omega^- & \omega^+ & 1 \\
-1& -1 & -1 \\
\omega^+ & \omega^- & 1
\end{pmatrix}
\end{equation}
with $\omega=e^{2\pi i /3}$, $\omega^\pm = \omega^{\pm 1}$, we have the following proposition regarding the asymptotics of $G$ near the origin.

\begin{prop}\label{lem:prefactor_global_param}
The matrix
\begin{equation}\label{def:matrix_G_hat}
\widehat G(z):= G(z)\diag\left(1, z^{\frac{A}{3}}(\mathcal U^\pm)^{-1}z^{\frac{B}{3}}\right),\qquad \pm \im z>0,
\end{equation}
is analytic in a neighborhood of $z=0$, and has an analytic inverse as well, where
\begin{align}\label{def:matrix_A}
A&=A(\nu,\kappa)=\diag (\nu+\kappa,\nu-2\kappa,\kappa-2\nu),
\\
B&=\diag(1,0,-1),\label{def:matrix_B}
\end{align}
and the matrices $\mathcal U^\pm$ are defined through \eqref{def:matrix_Upm}--\eqref{def:calU}.
\end{prop}
\begin{proof}

It is clear that $\widehat G$ defined in \eqref{def:matrix_G_hat} is analytic in the upper and lower half planes.
We now compute its jumps across the real axis in a neighborhood of the origin.

For $0<x<p$, it follows from \eqref{eq:jumps_global_parametrix} and \eqref{def:matrix_Upm} that
\begin{align*}
(\widehat G_-(x))^{-1}{\widehat G_+(x)} & =\diag \left(1,x^{-\frac{B}{3}}  \mathcal U^{-}x^{-\frac{A}{3}}\right) J_G(x)\diag\left(1,x^{\frac{A}{3}}(\mathcal U^+)^{-1}x^{\frac{B}{3}}\right) \\
& = \diag\left( 1, x^{-\frac{B}{3}} \mathcal U^-
\diag\left(
\begin{pmatrix}
0 & 1 \\
-1 & 0
\end{pmatrix},
1
 \right)
(\mathcal U^{+})^{-1}x^{\frac{B}{3}}\right) \\
& = \diag\left(1, x^{-\frac{B}{3}}\mathcal U^+ (\mathcal U^{+})^{-1}x^{\frac{B}{3}}\right)=I_4.
\end{align*}
Similarly, if $-q<x<0$, we use again \eqref{eq:jumps_global_parametrix} and compute
\begin{align*}
(\widehat G_-(x))^{-1}{\widehat G_+(x)} & =\diag\left(1,x_-^{-\frac{B}{3}} \mathcal U^{-}x^{-\frac{A}{3}}_-\right)J_G(x)\diag\left(1, x^{\frac{A}{3}}_+(\mathcal U^+)^{-1}x^{\frac{B}{3}}_+\right)
\\
& = \diag \left(1, x^{-\frac{B}{3}}_- \mathcal U^{-}
\diag\left(\mathfrak c_{\frac{2(\kappa+\nu)}{3}},
\begin{pmatrix}
0 & -\mathfrak c_{-\frac{\kappa+\nu}{3}} \\
\mathfrak c_{-\frac{\kappa+\nu}{3}} & 0
\end{pmatrix}
 \right)
( \mathcal U^{+})^{-1}x^{\frac{B}{3}}_+ \right)
\\
& = \diag \left(1, x^{-\frac{B}{3}}_- \widehat{\mathcal U }
\begin{pmatrix}
0 & 0 & 1 \\
1 & 0 & 0 \\
0 & 1 & 0
\end{pmatrix}
\widehat{\mathcal U }^{-1}
x^{\frac{B}{3}}_+ \right).
\end{align*}
Note that
$$
\widehat{\mathcal U}^{-1}=
\frac{1}{3}
\begin{pmatrix}
\omega^+ & -1 & \omega^- \\
\omega^- & -1 & \omega^+ \\
1 & -1 & 1
\end{pmatrix},
$$
a straightforward calculation shows that
\begin{align*}
(\widehat G_-(x))^{-1}{\widehat G_+(x)}
& = \diag \left(1, x^{-\frac{B}{3}}_-
\diag\left(\omega^-, 1, \omega^+ \right)
x^{\frac{B}{3}}_+ \right)=I_4, \qquad -q<x<0.
\end{align*}
Hence, we can conclude that $G$ is analytic in a neighborhood of the origin with $z=0$ being an isolated singularity.

We next show that $z=0$ is a removable singularity. Note that, as $z \to 0$,
\begin{equation*}
F_1(z) =\Boh(1), \qquad
F_{2,3,4}(z) = \Boh(z^{\frac{2\kappa -\nu}{3}}).
\end{equation*}
Thus,
$$
\diag\left(\mathfrak c_{\kappa}F_1(z),F_2(z)e^{-\kappa \log z},F_3(z),\mathfrak c_{\kappa-\nu}F_4(z)e^{(\nu-\kappa)\log z}\right) = \widehat F(z)\diag(1,z^{-\frac{A}{3}}),
$$
where $\widehat F$ is a diagonal matrix satisfying
$$\widehat F(z)=F_0 +\Boh(z^{\frac{1}{3}}), \qquad z\to 0, $$ for some non-singular constant matrix $F_0$. This, together with \eqref{eq:G} and \eqref{eq:asymptotics_global_parametrix_zero_zero_param}, implies that
\begin{equation*}
G(z)\diag\left(1, z^{\frac{A}{3}}\right)=\diag(\mathfrak f_1 ,\mathfrak f_1,\mathfrak f_3,\mathfrak f_3) G_0(z)\widehat F(z)=\Boh(z^{-\frac{1}{3}}),\qquad z\to 0.
\end{equation*}
By \eqref{def:matrix_G_hat}, we further get that
$$\widehat G(z)=\Boh(z^{-\frac{2}{3}}), \qquad z\to 0,$$ so $z=0$ must be a removable singularity, as claimed.

Finally, the existence of the analytic inverse follows immediately because the determinants of $G$, $\mathcal U^{\pm}$, $z^{\frac{A}{3}}$ and $z^{\frac{B}{3}}$ are all constant and non-zero, so the same is true for $\det \widehat G$.

This completes the proof of Proposition \ref{lem:prefactor_global_param}.
\end{proof}

\section{Local parametrices near \texorpdfstring{$p$}{p} and \texorpdfstring{$-q$}{-q}}\label{section:airy_parametrices}
From our definition of $\phi$-functions given in \eqref{def:phi_functions1}--\eqref{def:phi_functions3}, it is readily seen that
\begin{align*}
\phi_2(z)&=C_2(z-p)^{\frac32}(1+\Boh(z-p)),  \quad~~~ \; \; \;z\to p,
\\
\phi_1(z)&=-C_1(z+q)^{\frac32}(1+\Boh(z+p)),  \qquad  z\to -q,
\end{align*}
for some positive constants $C_1$ and $C_2$. Hence, by setting $D_p(\delta)$ and $D_{-q}(\delta)$ with $\delta>0$ sufficiently small to be two small disks around $p$ and $-q$, we could construct local parametrices $L_{p}$ and $L_{-q}$ in each of the disk with the aid of the standard $2\times 2$ Airy parametrix \cite{DKMVZ99}. Since this construction is very well-known, we omit the details but mention that as one of the outcomes we get the matching
\begin{equation}\label{eq:matchingp}
L_{j}(z)=(I_4+\Boh(n^{-1}))G(z), \qquad n\to\infty,
\end{equation}
uniformly for $z\in \partial D_{j}(\delta)$, $j=p,-q$.

\section{Local parametrix near the origin}

In this section, we will construct the local parametrix near the origin, which is somewhat involved and performed in several steps. The main difficulty lies in the fact one cannot expect a nice matching like \eqref{eq:matchingp} immediately in this case, and this phenomenon is quite common in the asymptotic analysis of higher order RH problem; cf. \cite{BB15,BK07,KM19}. Here, we follow a new and novel technique recently developed by Kuijlaars and Molag \cite{KM19}, which requires to construct a matching condition on two circles.

Let $D(\delta)$ and $D(r)$ be disks centered at the origin with radii $\delta>r>0$. We will take $D(\delta)$ to be small but fixed and $D(r)=D(r_n)$ to be shrinking with $n$. A more precise requirement on $r$ will be given later.

On account of the second inequality in \eqref{eq:inequalities_phi_off_support} and the fact that $\kappa \geq 0$, we could simply ignore the $(2,1)$-entry of $J_S$ on $(-\delta, 0)$ for large $n$ and this leads us to consider the following RH problem.
\begin{rhp}\label{rhp:local_parametrix}
We look for a $4\times 4$ matrix-valued function $L_0$ with the following properties:
\begin{enumerate}[\rm (1)]
\item $L_0$ is defined and analytic in $D(\delta) \setminus \left((\Gamma_S\cap D(r))\cup (-\delta,\delta)\right)$, where the contour $\Gamma_S$ is defined in \eqref{def:gammaS}.
\item $L_0$ satisfies the jump condition
\begin{equation*}
L_{0,+}(z)=L_{0,-}(z)\begin{cases}
                         J_G(z), & z\in (-\delta,\delta)\setminus [-r,r],
                         \\
                         J_{L_0}(z), & z\in D(r) \cap \Gamma_S,
                       \end{cases}
\end{equation*}
where $J_G$ is defined in \eqref{eq:jumps_global_parametrix}, and
\begin{equation*}
J_{L_0}(z)=
\begin{cases}
J_S(z), & z\in (D(r)\cap \Gamma_S)\setminus (-r,0),
\\
\diag\left(I_2,\begin{pmatrix}
0 & -|z|^{\nu-\kappa} \\
|z|^{\kappa-\nu} & 0 \\
\end{pmatrix}\right),
&
z\in (-r,0),
\end{cases}
\end{equation*}
with $J_S$ given in \eqref{eq:jumps_S}.

\item As $z\to 0$, $L_0$ has at worse a power log singularity.

\item As $n\to\infty$, we have the matching conditions
\begin{equation}\label{eq:local_param_decay_boundary}
L_0(z)=(I_4+\Boh(n^{-1}))G(z), \qquad z\in \partial D(\delta) \setminus (-\delta,\delta),
\end{equation}
where $G$ is the global parametrix \eqref{eq:G}, and
\begin{equation}\label{eq:local_param_decay_shrinking_boundary}
L_{0,+}(z)=(I_4+\Boh(n^{-1}))L_{0,-}(z),\qquad z\in \partial D(r)\setminus \Gamma_S,
\end{equation}
where the orientation of the circle is taken in a counter-clockwise manner and the error terms in \eqref{eq:local_param_decay_boundary} and \eqref{eq:local_param_decay_shrinking_boundary} are uniform in $z$.
\end{enumerate}
\end{rhp}

In previous works in the literature, only the matching \eqref{eq:local_param_decay_boundary} is present, with possibly a shrinking radius $\delta=\delta_n$. In these scenarios, one often has to make several post-corrections to the matching, as the initial error term is not of the appropriate order. As mentioned earlier, in \cite{KM19} Kuijlaars and Molag explored the introduction of this new matching condition \eqref{eq:local_param_decay_shrinking_boundary}, which allows to keep $\delta$ fixed but make $r=r_n$ shrinking. It turns out that this double-matching \eqref{eq:local_param_decay_boundary}--\eqref{eq:local_param_decay_shrinking_boundary} makes the coming calculations more systematic, and this will be the approach we follow. We next present some preliminary work before doing that.

As the first step to solve the RH Problem \ref{rhp:local_parametrix}, we remove all the $\phi$-functions from the jumps of $L_0$ by defining

\begin{equation}\label{eq:from_L0_to_P}
P(z)=L_0(z)\diag\left(1,e^{-n(\lambda_2(z)+c)},e^{-n(\lambda_3(z)+c)},e^{-n(\lambda_4(z)+c)}\right), \quad z\in D(r)\setminus \Gamma_S,
\end{equation}
where
\begin{equation}\label{def:constantc}
c=\int_0^p \xi_{2,+}(s)\ud s=2\pi i + \int_0^p \xi_{2,-}(s)\ud s=\int_0^p \xi_{3,-}(s)\ud s=2\pi i + \int_0^p \xi_{3,+}(s)\ud s,
\end{equation}
and the $\lambda$-functions  are defined in \eqref{def:lambda_functions2}--\eqref{def:lambda_functions4}.

An appropriate, but straightforward, combination of \eqref{def:lambda_functions2}--\eqref{def:lambda_functions4}, \eqref{def:phi_functions2}--\eqref{def:phi_functions3} and Proposition \ref{prop:lamdaphi} leads us to consider the following RH Problem that $P$ must satisfy.

\begin{rhp}\label{rhp:parametrix_P}
The function $P$ defined in \eqref{eq:from_L0_to_P} has the following properties:
\begin{enumerate}[\rm (1)]
\item $P$ is defined and analytic on $D(r)\setminus \Gamma_S$.
\item $P$ satisfies the jump condition
\begin{equation*}
P_{+}(z)=P_{-}(z)J_P(z), \qquad z\in \Gamma_S \cap D(r),
\end{equation*}
where
$$
J_P(z)=
\begin{cases}
I_4+ z^{-\kappa}E_{32},
&
z\in D(r)\cap \partial \mathcal L_2^\pm, \\
I_4- \mathfrak c_{\nu-\kappa} z^{\kappa-\nu}E_{43},
&
z\in D(r)\cap \partial \mathcal L_3^\pm, \\
\diag\left( 1,
\begin{pmatrix}
0 & z^{\kappa} \\
-z^{-\kappa} & 0
\end{pmatrix},
1\right),
&
z\in (0,r),\\
\diag\left(I_2,\begin{pmatrix}
0 & -|z|^{\nu-\kappa} \\
|z|^{\kappa-\nu} & 0 \\
\end{pmatrix}\right),
&
z\in (-r,0).
\end{cases}
$$

\item As $z\to 0$, $P$ has at worse a power log singularity.

\end{enumerate}
\end{rhp}
Note that we do not pose any asymptotic behavior of $P$ on $\partial D(r)$. We will give an explicit solution to the above RH problem, and, after some further manipulations, modify $P$ in such a way that, at the end of the day, the corresponding matrix $L_0$ solves the RH problem \ref{rhp:local_parametrix}.

For later use, we introduce the functions $\widehat\lambda_k^\pm(z)$ defined by
\begin{equation}\label{def:lambda_hat}
\widehat\lambda_k^\pm(z)=\int_0^z\xi_{k}(s) \ud s, \qquad  \pm \im z>0, \qquad k=2,3,4,
\end{equation}
where the path for $\widehat\lambda_k^+(z)$ ($\widehat\lambda_k^-(z)$) is contained in the upper (lower) half plane. It is easily seen from \eqref{def:lambda_functions2}--\eqref{def:lambda_functions4}, \eqref{def:constantc} and \eqref{def:lambda_hat} that
\begin{equation}\label{eq:lambdaandhatlambda}
e^{-n(\lambda_k(z)+c)}=e^{-n\widehat\lambda_k^{\pm}(z)}, \qquad \pm \im z>0, \qquad k=2,3,4.
\end{equation}
The asymptotic behaviors of $\widehat\lambda_k^\pm$ near the origin are collected in the following proposition.
\begin{prop}\label{prop:hatlambda}
There exist analytic functions $f_4$, $g_4$ and $h_4$ in a neighborhood of $z=0$ so that for $\pm \im z>0$
\begin{align}
\widehat\lambda_2^\pm(z)&=\omega^{\pm}z^{1/3}f_4(z)+\omega^{\mp}z^{2/3}g_4(z)+zh_4(z),
\label{eq:hatlambda1}
\\
\widehat\lambda_3^\pm(z)&=\omega^{\mp}z^{1/3}f_4(z)+\omega^{\pm}z^{2/3}g_4(z)+zh_4(z),
\label{eq:hatlambda2}
\\
\widehat\lambda_4^\pm(z)&=z^{1/3}f_4(z)+z^{2/3}g_4(z)+zh_4(z).
\label{eq:hatlambda3}
\end{align}
Furthermore, we have
\begin{equation}\label{eq:conformal_map_norming}
f_4(0)=3(\beta^2-\alpha^2)^{1/3}>0.
\end{equation}
\end{prop}
\begin{proof}
From the local behavior of the $\xi$-functions near $z=0$ (which can be derived from \eqref{eq:local_behavior_xi}, \eqref{eq:delta2} and the spectral curve \eqref{eq:spectral_curve}), it is readily seen that, as $z \to 0$,
\begin{equation}\label{eq:eqhatlambdaexpan}
\widehat \lambda^{\pm}_k(z)=z^{1/3}f_{k}^{\pm}(z)+z^{2/3}g_{k}^\pm(z)+zh_k^\pm (z), \qquad \pm \im z>0, \qquad k=2,3,4,
\end{equation}
where $f_k^\pm$, $g_k^\pm$ and $h_k^\pm$ are analytic in a neighborhood of $z=0$ and all the roots are taken
the principal branches with cuts along the negative axis. The jump relations for the $\xi$-functions across the positive axis imply in particular that
\begin{align*}
& f_2^\pm(z)=f_3^\mp(z), &&  g_2^\pm(z)=g_3^\mp(z), &&  h_2^\pm(z)=h_3^\mp(z),
\\
& f_4^+(z)=f_4^-:=f_4(z),&& g_4^+=g_4^-:=g_4(z), && h_4^+(z)=h_4^-(z):=h_4(z),
\end{align*}
while the jump conditions across the negative axis give that
\begin{align*}
& f_2^-(z)=\omega f_2^+(z), &&  g_2^-(z)=\omega^- g_2^+(z), &&  h_2^-(z)=h_2^+(z), \\
& f_3^\pm(z)=\omega^\mp f_4(z), && g_3^\pm (z)= \omega^\pm g_4(z), && h_3^\pm(z)=h_4(z).
\end{align*}
As a consequence, we obtain the relations
$$
f_2^\pm(z)=f_3^\mp(z)=\omega^\pm f_4(z), \quad g_2^\pm(z)=g_3^\mp(z)=\omega^\mp g_4(z), \quad h_2^\pm(z)=h_3^\pm(z) = h_4(z).
$$
Inserting the above formulas into \eqref{eq:eqhatlambdaexpan} gives us \eqref{eq:hatlambda1}--\eqref{eq:hatlambda3}.

Finally, the fact that $\xi_4>0$ on the positive axis (see \eqref{eq:definition_xi_functions}) allows us to conclude from \eqref{eq:spectral_curve} that
$$\xi_4(z) \sim 3(\beta^2-\alpha^2)^{1/3}z^{-2/3}, \qquad  z\to 0,$$
which in turn implies \eqref{eq:conformal_map_norming}.

This completes the proof of Proposition \ref{prop:hatlambda}.
\end{proof}

The explicit construction of $P$ is based on the bare Meijer-G parametrix which is described in the next section.

\subsection{The Meijer-G parametrix of Bertola-Bothner}

The model RH problem we need to solve for $P$ was introduced by Bertola and Bothner in the context of a model of several coupled positive-definite matrices \cite{BB15}, which is called `bare Meijer-G parametrix for $p$-chain', $p=2,3,\ldots$, therein. The one that is relevant to the present work corresponds to the case $p=2$ and reads as follows.\footnote{
For convenience, the correspondence between our notations and those used in \cite{BB15} is listed below:
$$
a_1=\kappa,\quad a_2=\nu-\kappa, \quad A_1=-\nu-\kappa, \quad A_2=2\kappa-\nu,\quad A_3=2\nu -\kappa,
$$
and
$$
(a_{j,k})_{j,k=1,2}=
\begin{pmatrix}
\kappa & \nu \\
0 & \nu-\kappa
\end{pmatrix},
\qquad
\Omega_{\pm} = \diag(\omega^\pm,\omega^\mp,1).
$$
Moreover, in \cite{BB15}, $\omega=\omega_{BB}=e^{\pi i/3}$, so $\omega_{BB}^2=\omega$ and the contours are
$$
\mathfrak r_0=\Gamma_0, \quad \mathfrak r_1=\Gamma_1, \quad \mathfrak r_2=\Gamma_5, \quad \mathfrak r_3=\Gamma_2, \quad \mathfrak r_4=\Gamma_4, \quad \mathfrak r_5=\Gamma_3,
$$
with all the $\mathfrak r_k$'s oriented from the origin towards $\infty$.}

\begin{rhp}\label{rhp:MeijerG}
The function $\Psi$ is a $3\times 3$ matrix-valued function satisfying the following properties:
\begin{enumerate}
\item[\rm (1)] $\Psi$ is defined and analytic in $ \mathbb{C} \setminus\Gamma_\Psi$,
where
$$
\Gamma_\Psi:=\bigcup\limits_{k=0}^5 \Gamma_k, \qquad \Gamma_k=e^{k\frac{\pi i}{3}} [0,+\infty), \quad k=0,\hdots,5,
$$
with the orientations as shown in Figure \ref{fig: model rhp}.
\item[\rm (2)]$\Psi$ satisfies the jump condition
$$\Psi_{+}(z)=\Psi_{-}(z)J_\Psi(z),\qquad z\in\Gamma_{\Psi},$$
where
\begin{equation}\label{def:JPsi}
J_\Psi(z)=
\begin{cases}
\begin{pmatrix}
0 & z^{\kappa} & 0 \\
-z^{-\kappa} & 0 & 0\\
0 & 0 & 1
\end{pmatrix},
& z\in \Gamma_0, \\
\begin{pmatrix}
1 & 0 & 0 \\
z^{-\kappa} & 1 & 0 \\
0 & 0 & 1
\end{pmatrix}, & z\in \Gamma_1\cup \Gamma_5,\\
\begin{pmatrix}
1 & 0 & 0 \\
0 & 1 & 0 \\
0 & -\mathfrak c_{\nu-\kappa} z^{\kappa-\nu} & 1
\end{pmatrix}, & z\in \Gamma_{2}\cup\Gamma_4, \\
\begin{pmatrix}
1 & 0 & 0 \\
0 & 0 & -|z|^{\nu-\kappa} \\
0 & |z|^{\kappa-\nu} & 0
\end{pmatrix}, & z\in \Gamma_3.
\end{cases}
\end{equation}

\item[\rm (3)] As $z\to 0$, $\Psi$ has at worse a power-log singularity. In particular, we have, as $z\to 0$,
\begin{equation}\label{eq:zeropsi}
\Psi(z) \begin{pmatrix}
1 & 0 & 0
\end{pmatrix}^{T}=\Boh(1).
\end{equation}

\item[\rm (4)] As $z\to\infty$ with $\pm \im z>0$, we have
\begin{equation}\label{eq:asymptotics_meijer_parametrix_infinity}
\Psi(z)=z^{-\frac{B}{3}} \; \mathcal U^\pm  K(z) z^{-\frac{A}{3}} \diag(e^{-3z^{1/3}\omega^\pm},e^{-3z^{1/3}\omega^\mp},e^{-3z^{1/3}}),
\end{equation}
where the diagonal matrices $A$ and $B$ are as in \eqref{def:matrix_A} and \eqref{def:matrix_B},
$\mathcal U^\pm$ are given in \eqref{def:matrix_Upm} and $K$ admits an asymptotic expansion of the form
\begin{equation}\label{eq:asympt_expansion_error_Psi}
K(z)\sim I_3+\sum_{j=1}^\infty \frac{K_j}{z^{\frac{j}{3}}},\qquad z\to \infty,
\end{equation}
where the coefficient $K_j$, $j=1,2,3,\ldots$, possibly depends on the sector $\Theta_j$, $j=0,\hdots, 5$, along which $z\to\infty$. Here, $\Theta_k$, $k=0,1,\ldots,5$, denotes the region between the contours $\Gamma_k$ and $\Gamma_{k+1}$; see Figure \ref{fig: model rhp} for an illustration.
\end{enumerate}
\end{rhp}

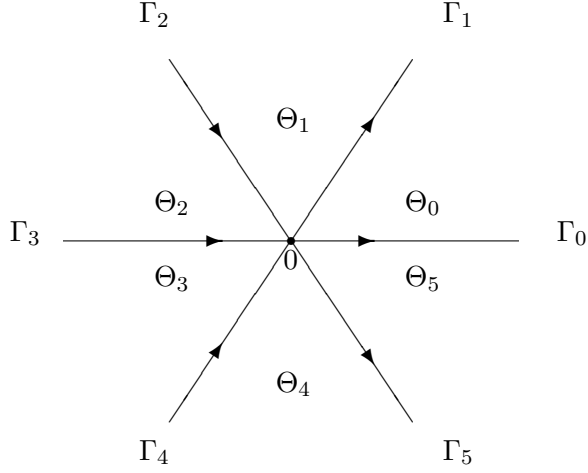
\begin{figure}[t]
\begin{center}
   \setlength{\unitlength}{1truemm}
   \begin{picture}(100,70)(-5,2)
       \put(40,40){\line(-2,-3){16}}
       \put(40,40){\line(-2,3){16}}
       \put(40,40){\line(-1,0){30}}
       \put(40,40){\line(1,0){30}}
       \put(40,40){\line(2,3){16}}
       \put(40,40){\line(2,-3){16}}

       \put(30,55){\thicklines\vector(2,-3){1}}
       \put(30,40){\thicklines\vector(1,0){1}}
       \put(30,25){\thicklines\vector(2,3){1}}
       \put(50,55){\thicklines\vector(2,3){1}}
       \put(50,40){\thicklines\vector(1,0){1}}
       \put(50,25){\thicklines\vector(2,-3){1}}

       \put(39,36.3){$0$}
       \put(20,11){$\Gamma_4$}

       \put(20,69){$\Gamma_2$}
       \put(3,40){$\Gamma_3$}
       \put(60,11){$\Gamma_5$}
       \put(60,69){$\Gamma_1$}
       \put(75,40){$\Gamma_0$}

       \put(22,44){$\Theta_2$}
       \put(22,34){$\Theta_3$}
       \put(55,44){$\Theta_0$}
       \put(55,34){$\Theta_5$}
       \put(38,55){$\Theta_1$}
       \put(38,20){$\Theta_4$}

       \put(40,40){\thicklines\circle*{1}}

   \end{picture}
   \caption{The jump contour $\Gamma_{\Psi}$ for the model RH problem $ \Psi$ and the regions $\Theta_k$, $k=0,\ldots,5$.}
   \label{fig: model rhp}
\end{center}
\end{figure}

The precise asymptotic behavior of $\Psi$ as $z\to 0$ depends on whether the values $\kappa$ and $\nu$ are zero or equal to one another. This behavior is indicated in \cite{BB15} as the behavior of certain iterated Cauchy transforms (see the RH Problem 4.22 and also equations (2.6) and (2.7) therein), but for our purposes the behavior as in item (3) above will suffice.

In the construction carried out by Bertola and Bothner \cite[Theorem 4.23]{BB15}, they actually only show that $K(z)=I+\Boh(z^{-1/3})$. Nevertheless, the existence of the full asymptotic expansion as in \eqref{eq:asympt_expansion_error_Psi} follows from the existence of the asymptotic expansion for the entries of $\Psi$ which, as we now explain, are given by Meijer G-functions - hence this parametrix bears the name {\it the Meijer-G parametrix}.

To describe it, recall that the Meijer G-function is given by the following contour integral in the complex plane:
\begin{equation}\label{def:Meijer}
G^{m,n}_{p,q}\left({a_1,\ldots,a_p \atop b_1,\ldots,b_q} \; \Big{|} \; \zeta \right)
 = \frac{1}{2\pi i}\int_L \frac{\prod_{j=1}^m\Gamma(b_j+u)\prod_{j=1}^n\Gamma(1-a_j-u)}
{\prod_{j=m+1}^q\Gamma(1-b_j-u)\prod_{j=n+1}^p\Gamma(a_j+u)}\zeta^{-u}\ud u.
\end{equation}
Here, it is assumed that
\begin{itemize}
  \item $0\leq m\leq q$ and $0\leq n \leq p$, where $m,n,p$ and $q$
  are integer numbers;
  \item The real or complex parameters $a_1,\ldots,a_p$ and
  $b_1,\ldots,b_q$ satisfy the conditions
  \begin{equation*}
  a_k-b_j \neq 1,2,3, \ldots, \quad \textrm{for $k=1,2,\ldots,n$ and $j=1,2,\ldots,m$,}
  \end{equation*}
  i.e., none of the poles of $\Gamma(b_j+u)$, $j=1,2,\ldots,m$ coincides
  with any poles of $\Gamma(1-a_k-u)$, $k=1,2,\ldots,n$.
\end{itemize}
The contour $L$ is chosen in such a way that all the poles of $\Gamma(b_j+u)$, $j=1,\ldots,m$, are on the left of the path, while all the poles of $\Gamma(1-a_k-u)$, $k=1,\ldots,n$, are on the right, which is usually taken to go from $-i\infty$ to $i\infty$. We now set
\begin{alignat*}{2}
g_1(z) & =\frac{c_1}{2\pi i}\int_{L} \frac{\Gamma(s)}{\Gamma(1+\kappa-s)\Gamma(1+\nu-s)}z^{-s}\ud s
&& =c_1G^{1,0}_{0,3}\left({- \atop 0,-\kappa,-\nu}\;\Big{|}\; z\right), \\
g_2^{(\pm)} (z) & =\frac{c_2}{2\pi i}\int_L \frac{\Gamma(s+\kappa)\Gamma(s)}{\Gamma(1+\nu-\kappa-s)}e^{\pm \pi i s}z^{-s}\ud s
&& = c_2G^{2,0}_{0,3}\left({- \atop 0,\kappa,\kappa-\nu}\;\Big{|}\; e^{\mp \pi i}z\right), \\
g_3(z) & = \frac{c_3}{2\pi i}\int_L \Gamma(s+\nu)\Gamma(s+\nu-\kappa)\Gamma(s)z^{-s}\ud s
&&=c_3G^{3,0}_{0,3}\left({- \atop 0,\nu-\kappa,\nu}\;\Big{|}\; z\right),
\end{alignat*}
where
$$
c_k=(2\pi i )^{3-k}\frac{\sqrt{3}}{2\pi}, \qquad k=1,2,3,
$$
and define the auxiliary function
$$
\widehat\Psi(z)=
\begin{pmatrix}
g_1(z) & g_2^{(\pm)}(z) & g_3(z) \\
z\frac{\ud g_1}{\ud z}(z) & \left(z\frac{\ud}{\ud z}-\kappa\right)g_2^{(\pm)}(z) & \left(z\frac{\ud}{\ud z}-\nu \right)g_3(z) \\
\left(z\frac{\ud}{\ud z}\right)^2 g_1(z) & \left(z\frac{\ud}{\ud z}-\kappa\right)^2g_2^{(\pm)}(z) & \left(z\frac{\ud}{\ud z}-\nu \right)^2g_3(z) \\
\end{pmatrix},
\quad \pm \im z>0.
$$
Then, the solution to the model RH problem \ref{rhp:MeijerG} for $\Psi$ is given by
\begin{equation}\label{def:meijerpara}
\Psi(z)=
\begin{cases}
\widehat \Psi(z),& z\in \Theta_1\cup \Theta_4, \\
\widehat \Psi(z) \diag
\left( 1,
\begin{pmatrix}
1 & 0 \\
\mathfrak c_{\nu-\kappa} \;  z^{\kappa-\nu}  & 1
\end{pmatrix}
\right),& z\in \Theta_2, \\
\widehat \Psi(z) \diag
\left(
\begin{pmatrix}
1 & 0\\
-  z^{-\kappa}  & 1
\end{pmatrix} , 1
\right),& z\in \Theta_0, \\
\widehat \Psi(z) \diag
\left( 1,
\begin{pmatrix}
1 & 0 \\
-\mathfrak c_{\nu-\kappa} \;  z^{\kappa-\nu}  & 1
\end{pmatrix}
\right),& z\in \Theta_3, \\
\widehat \Psi(z) \diag
\left(
\begin{pmatrix}
1 & 0\\
z^{-\kappa}  & 1
\end{pmatrix} , 1
\right),& z\in \Theta_5,
\end{cases}
\end{equation}
recall that the region $\Theta_k$, $k=0,1,\ldots,5$, is shown in Figure \ref{fig: model rhp}.

We conclude this section with some auxiliary results for later purposes.

\begin{lem}\label{lem:entire_factor_expansion_infinity}
The matrix-valued function
\begin{equation}\label{def:matrixQ1}
Q_1(z)=z^{-\frac{B}{3}}\mathcal U^{\pm}\diag(e^{-3z^{1/3}\omega^\pm},e^{-3z^{1/3}\omega^\mp},e^{-3z^{1/3}}) (\mathcal U^\pm)^{-1}z^{\frac{B}{3}},\qquad \pm \im z>0,
\end{equation}
is entire, where the matrices $B$ and $\mathcal U^\pm$ are defined in \eqref{def:matrix_B} and \eqref{def:matrix_Upm}. Similarly,
for any function $\vartheta$ analytic near the origin, the matrix-valued function
\begin{equation*}
Q_2(z)=z^{-\frac{B}{3}}\mathcal U^{\pm}\diag(e^{z^{2/3}\omega^\mp\vartheta(z)},e^{z^{2/3}\omega^\pm\vartheta(z)},e^{z^{2/3}\vartheta(z)}) (\mathcal U^\pm)^{-1}z^{\frac{B}{3}},\qquad \pm \im z>0,
\end{equation*}
is analytic near the origin as well.
\end{lem}
\begin{proof}
Both $Q_1$ and $Q_2$ take the form
$$
Q(z)=z^{-\frac{B}{3}}\mathcal U^{\pm}\diag(e^{\vartheta_1(z)},e^{\vartheta_2(z)},e^{\vartheta_3(z)}) (\mathcal U^\pm)^{-1}z^{\frac{B}{3}},
$$
where $\vartheta_1$, $\vartheta_2$ and $\vartheta_3$ are analytic functions on $\mathcal V\setminus \R$ ($\mathcal V=\C$ for $Q_1$, and $\mathcal V$ is a neighborhood of the origin for $Q_2$), with jumps across $\R$ related through
\begin{align*}
& \vartheta_{1,+}(x)-\vartheta_{2,-}(x)=\vartheta_{2,+}(x)-\vartheta_{1,-}(x)=\vartheta_{3,+}(x)-\vartheta_{3,-}(x)=0,\qquad x>0,\\
& \vartheta_{1,+}(x)-\vartheta_{1,-}(x)=\vartheta_{2,+}(x)-\vartheta_{3,-}(x)=\vartheta_{3,+}(x)-\vartheta_{2,-}(x)=0,\qquad x<0.
\end{align*}
After a cumbersome but straightforward calculation, these relations combined assure us that $Q$ has no jumps across the real axis, so $z=0$ is an isolated singularity of $Q$.
Furthermore, it is clear that
$$
Q(z)=\Boh(z^{-2/3}),
$$
so $z=0$ is actually a removable singularity of $Q$, as required.

This completes the proof of Lemma \ref{lem:entire_factor_expansion_infinity}.
\end{proof}

\subsection{Construction of the local parametrix \texorpdfstring{$P$}{P}}

We now construct the parametrix $P$ that solves the RH problem \ref{rhp:parametrix_P}.  To do so, recall the function $f_4(z)$ given in Proposition \ref{prop:hatlambda} and set
\begin{equation*}
\varphi(z)=\frac{1}{27}z(f_4(z))^{3}.
\end{equation*}
From Proposition~\ref{prop:hatlambda}, it follows that the function $\varphi$ is conformal in a neighborhood of $z=0$. Even more so, we actually have
\begin{equation}\label{eq:varphizero}
\varphi(z)= z(\beta^2-\alpha^2)(1+\Boh(z)), \qquad z\to 0.
\end{equation}
Furthermore, by deforming the lenses if needed, we can assume that $\Gamma_S\cap D(r)$ is mapped by $z\mapsto n^3\varphi(z)$ to the union of the contours $\Gamma_0\cup\cdots\cup\Gamma_5$ and, in virtue of \eqref{eq:varphizero}, also that $\varphi((0,r))\subset \Gamma_0$. We then define
\begin{equation}\label{def:matrix_P}
 P(z)=\diag\left(1, \left( \frac{n}{3}f_4(z) \right)^{B}\Psi(n^3\varphi(z))\left(\frac{n}{3}f_4(z)\right)^A\right),\qquad z\in D(r)\setminus \Gamma_S.
 \end{equation}
where $\Psi$ is the Meijer-G parametrix \eqref{def:meijerpara}, and the matrices $A,B$ are given in \eqref{def:matrix_A} and \eqref{def:matrix_B}, respectively.

\begin{prop}\label{prop:Pz}
The matrix-valued function $P(z)$ defined in \eqref{def:matrix_P} solves the RH problem \ref{rhp:parametrix_P}.
\end{prop}
\begin{proof}
It is easily seen that $P$ is analytic on $D(r) \setminus \Gamma_S$. To show $P$ satisfies other items of the RH problem \ref{rhp:parametrix_P}, we start with checking the jump condition.
If $z\in (0,r)$, with $J_{\Psi}$ defined in  \eqref{def:JPsi}, we have
\begin{align*}
J_P(z) & =P_{-}(z)^{-1}P_+(z)
\\
&= \diag\left(1,\left(\frac{n}{3}f_4(z)\right)_-^{-A}J_\Psi(n^3\varphi(z))\left(\frac{n}{3}f_4(z)\right)_+^{A} \right)
 \\
& =
I_4+
 \left(\frac{n}{3}f_4(z)\right)^{-\nu-\kappa}\left(n^3\varphi(z)\right)^{\kappa}\left(\frac{n}{3}f_4(z)\right)^{\nu-2\kappa}E_{23}
\\
& \quad - \left(\frac{n}{3}f_4(z)\right)^{-\nu+2\kappa}\left(n^3\varphi(z)\right)^{-\kappa} \left(\frac{n}{3}f_4(z)\right)^{\nu+\kappa}E_{32}
\\
& = I_4+\left(\frac{n}{3}f_4(z)\right)^{-3\kappa}\left(\frac{n}{3}z^{\frac{1}{3}}f_4(z)\right)^{3\kappa}E_{23}-\left(\frac{n}{3}f_4(z)\right)^{3\kappa}
\left(\frac{n}{3}z^{\frac{1}{3}}f_4(z)\right)^{-3\kappa}E_{32}
\\
& = \diag\left( 1,
\begin{pmatrix}
0 & z^{\kappa} \\
-z^{-\kappa} & 0
\end{pmatrix},
1\right),
\end{align*}
as expected. The jump matrix of $P$ on other parts of $D(r) \cap \Gamma_S$ can be computed similarly, we omit the details here. Finally, the behavior of $P$ near the origin follows from the behavior of $\Psi$ given in item (3) of the RH problem \ref{rhp:MeijerG} and the fact that all the other terms in \eqref{def:matrix_P} remain bounded as $z\to 0$.

This completes the proof of Proposition \ref{prop:Pz}.
\end{proof}

We further set
\begin{equation}\label{def:hatP}
\widehat P(z):=P(z)\diag\left(1, \diag(e^{n\widehat \lambda_{k+1}^{\pm}(z)})_{k=1}^3  \right),
\end{equation}
where the functions $\widehat\lambda_k^\pm(z)$, $k=2,3,4$, are defined in \eqref{def:lambda_hat}. On account of \eqref{eq:from_L0_to_P} and \eqref{eq:lambdaandhatlambda}, it is easily seen that $\widehat P(z)$ satisfies the same jump condition as $L_0$ for $z\in \Gamma_S \cap D(r)$ and item (3) of the RH problem \ref{rhp:local_parametrix} for $L_0$. As shown later, we will solve the RH problem \ref{rhp:local_parametrix} with the aid of $\widehat P$. For that purpose, we next explore the asymptotics of $\widehat P$ on the boundary of the disk.

From now on, we assume, as mentioned before, that $\delta>0$ is sufficiently small and fixed but make $r=r_n$ shrink with $n$, namely,
\begin{equation}\label{eq:scaling_boundary}
r=r_n=n^{-\frac{3}{2}}.
\end{equation}
Since
$$n^3\varphi(z)\to\infty, \qquad z\in \partial D(r_n),\qquad n\to \infty,$$  under the scaling \eqref{eq:scaling_boundary}, we can use \eqref{def:hatP}, \eqref{def:matrix_P} and \eqref{eq:asymptotics_meijer_parametrix_infinity} to compute
\begin{align}\label{eq:asymptotics_parametrix_P}
\widehat P(z)&= \diag\left(1, \left( \frac{n}{3}f_4(z) \right)^{B}\Psi(n^3\varphi(z))\left(\frac{n}{3}f_4(z)\right)^A\diag(e^{n\widehat \lambda_{k+1}^{\pm}(z)})_{k=1}^3\right),
\nonumber
\\
&=\diag\left(1, z^{-\frac{B}{3}}\mathcal K(z)\mathcal U^{\pm} z^{-\frac{A}{3}}D_n(z) \right),\qquad z\in \partial D(r_n), \qquad n\to\infty,
\end{align}
where $\mathcal K=\mathcal K_n$ is an error matrix explicitly given by
\begin{equation}\label{def:curly_K}
\mathcal K(z):=\mathcal U^\pm K(n^3\varphi(z)) \mathcal (\mathcal U^{\pm})^{-1},\qquad \pm \im z>0,
\end{equation}
with $K$ being given in the asymptotic formula \eqref{eq:asymptotics_meijer_parametrix_infinity}, and
\begin{align}
D_n(z)&= \diag\left( e^{n(\widehat \lambda_2^\pm(z)-3\omega^\pm\varphi(z)^{1/3})}, e^{n(\widehat \lambda_3^\pm(z)-3\omega^\mp\varphi(z)^{1/3})}, e^{n(\widehat \lambda_4^\pm(z)-3\varphi(z)^{1/3})} \right) \nonumber \\
& =  e^{nzh_4(z)}\diag\left(e^{n \omega^\mp z^{2/3}g_4(z)},e^{n \omega^\pm z^{2/3}g_4(z)},e^{n z^{2/3}g_4(z)}  \right),
 \qquad \pm \im z>0.\label{def:matrix_Dn}
\end{align}
In the second equality of \eqref{def:matrix_Dn}, we have made use of Proposition \ref{prop:hatlambda}, which also implies that $D_n(z)$ remains bounded for $z\in \overline{D(r_n)}$ under the scaling \eqref{eq:scaling_boundary}. By defining
\begin{equation}\label{def:hat_Dn}
\widehat D_n(z):=\diag\left(1,(n^{3/2}z)^{-\frac{B}{3}}\mathcal U^\pm D_n(z) (\mathcal U^\pm)^{-1}(n^{3/2}z)^{\frac{B}{3}}\right),\qquad \pm \im z>0,
\end{equation}
it then follows from $\eqref{def:matrix_G_hat}$ that
\begin{align*}
\diag\left(1,z^{-\frac{B}{3}}\mathcal U^\pm z^{-\frac{A}{3}}D_n(z)\right) & = n^{\frac{\widehat B}{2}}\widehat D_n(z)n^{-\frac{\widehat B}{2}}\diag\left(1,z^{-\frac{B}{3}}\mathcal U^\pm z^{-\frac{A}{3}}\right)
\\
&=n^{\frac{\widehat B}{2}} \widehat D_n(z)n^{-\frac{\widehat B}{2}}\widehat G(z)^{-1}G(z),
\end{align*}
where
\begin{equation}\label{def:hatB}
\widehat B:=\diag(0,B)=\diag(0,1,0,-1).
\end{equation}
Thus, we could rewrite the asymptotics in \eqref{eq:asymptotics_parametrix_P} as
\begin{equation}\label{eq:error_P}
\widehat P(z)=\diag\left(1,z^{-\frac{B}{3}}\mathcal K(z)z^{\frac{B}{3}}\right)n^{\frac{\widehat B}{2}} \widehat D_n(z)n^{-\frac{\widehat B}{2}}\widehat G(z)^{-1}G(z),
\end{equation}
for $z\in \partial D(r_n)$ and $n\to\infty$.
We will need some auxiliary results on the matrices $\mathcal K$ and $\widehat D_n$ in the above formula that we discuss next.

\begin{lem}\label{lem:asycurlyK}
With the function $\mathcal K(z)$ defined in \eqref{def:curly_K}, we have that
for $z \in \partial D(r_n)$ and large $n$, $\mathcal K(z)$
admits a formal asymptotic expansion of the form
\begin{equation}\label{eq:puiseaux_exp_error}
\mathcal  K(z)\sim I_3+\sum_{j=1}^{\infty}\frac{\mathcal K_j}{n^jz^\frac{j}{3}},
\end{equation}
where the matrix coefficients $\mathcal K_j$ are independent of $z$ and $n$, and take the following structures:
\begin{equation}\label{eq:puiseaux_exp_error_structure}
\mathcal K_j =
\begin{cases}
\begin{pmatrix}
\ast & 0 & 0 \\
0 & \ast & 0 \\
0 & 0 & \ast
\end{pmatrix},
& j\equiv 0 \mod 3, \\
\begin{pmatrix}
0 & 0 & \ast \\
\ast & 0 & 0 \\
0 & \ast & 0
\end{pmatrix},
& j\equiv 1 \mod 3, \\
\begin{pmatrix}
0 & \ast & 0 \\
0 & 0 & \ast \\
\ast & 0 & 0
\end{pmatrix},
& j\equiv 2 \mod 3.
\end{cases}
\end{equation}
\end{lem}

\begin{proof}
If $z \in \partial D(r_n)$, we have that $z=\Boh(n^{-\frac{3}{2}})$ and $n^3\varphi(z)=\Boh(n^3z)=\Boh(n^{\frac{3}{2}})$, so the existence of the expansion \eqref{eq:puiseaux_exp_error} with coefficients independent of $n$, but possibly depending on the sector along which $z\to \infty$, follows from the asymptotic expansion of $K$ given in \eqref{eq:asympt_expansion_error_Psi}.

Let $u\in \C$ be a dummy variable. It follows from a calculation similar to that carried out in the proof of Proposition \ref{lem:prefactor_global_param} that $\Psi(u)u^{\frac{A}{3}}(\mathcal U^\pm)^{-1}u^{\frac{B}{3}}$ has no jump on $\R\setminus \{0 \}$. Furthermore, from \eqref{eq:asymptotics_meijer_parametrix_infinity}, we find that
$$
\Psi(u)u^{\frac{A}{3}}(\mathcal U^\pm)^{-1}u^{\frac{B}{3}}=u^{-\frac{B}{3}}\mathcal U^{\pm}K(u)(\mathcal U^\pm)^{-1}u^{\frac{B}{3}}Q_1(u),\qquad u\to \infty, \qquad \pm \im u>0,
$$
where $Q_1$ defined in \eqref{def:matrixQ1} is an entire function. Setting $u=n^3\varphi(z)$, this yields
\begin{equation}\label{eq:uinphi}
z^{-\frac{B}{3}}\mathcal K(z) z^{\frac{B}{3}}=\left( \frac{n}{3}f_4(z) \right)^B \Psi(u)u^{\frac{A}{3}}(\mathcal U^\pm)^{-1}u^{\frac{B}{3}} Q_1(u)^{-1}\left( \frac{n}{3}f_4(z) \right)^{-B},
\end{equation}
which should be understood in the scaling \eqref{eq:scaling_boundary} and $n$ sufficiently large but fixed. Now, the functions in $u$ appearing on the right-hand side of \eqref{eq:uinphi} do not have jumps on the real axis, and neither do the functions in $z$ because they are entire. Thus, the right-hand side admits an asymptotic expansion in integer powers of $z$ (recall that $u=u(z)$ is conformal), with $n$-dependent coefficients.

On the other hand, the left-hand side admits an asymptotic expansion in inverse powers of $z^{\frac{1}{3}}$, but possibly with different coefficients in different sectors of the plane.
A comparison of the asymptotic expansions on both sides then yields that the expansion
$$
z^{-\frac{B}{3}}\mathcal K(z)z^{\frac{B}{3}}\sim I_3+\sum_{j=1}^\infty z^{-\frac{B}{3}}\mathcal K_j z^{\frac{B}{3}}\frac{1}{n^j z^{\frac{j}{3}}}
$$
must involve only inverse integer powers, and furthermore the coefficients should not depend on the sector along which $z\to\infty$. Further noticing the identity
$$
z^{-\frac{B}{3}}\mathcal K_j z^{\frac{B}{3}}=\Boh
\begin{pmatrix}
1 & z^{-\frac{1}{3}} & z^{-\frac{2}{3}} \\
z^{\frac{1}{3}} & 1 & z^{-\frac{1}{3}} \\
z^{\frac{2}{3}} & z^{\frac{1}{3}} & 1
\end{pmatrix},
$$
we then conclude the structure \eqref{eq:puiseaux_exp_error_structure}.

This completes the proof of Lemma \ref{lem:asycurlyK}.
\end{proof}

For any $a,b,c\in \C$, it is straightforward to verify the following commutation relations:
\begin{equation}\label{eq:commutation_relations}
\begin{aligned}
z^{-\frac{B}{3}}
\begin{pmatrix}
0 & 0 & a \\
b & 0 & 0 \\
0 & c & 0
\end{pmatrix}
z^{\frac{B}{3}}
& =
\begin{pmatrix}
0 & 0 & 0 \\
b & 0 & 0 \\
0 & c & 0
\end{pmatrix} z^{\frac{1}{3}} +
\begin{pmatrix}
0 & 0 & a \\
0 & 0 & 0 \\
0 & 0 & 0
\end{pmatrix}z^{-\frac{2}{3}}, \\
z^{-\frac{B}{3}}
\begin{pmatrix}
0 & a & 0 \\
0 & 0 & b \\
c & 0 & 0
\end{pmatrix}
z^{\frac{B}{3}}
& =
\begin{pmatrix}
0 & 0 & 0 \\
0 & 0 & 0 \\
c & 0 & 0
\end{pmatrix} z^{\frac{2}{3}} +
\begin{pmatrix}
0 & a & 0 \\
0 & 0 & b \\
0 & 0 & 0
\end{pmatrix}z^{-\frac{1}{3}}.
\end{aligned}
\end{equation}
%
Thus, we obtain from \eqref{eq:commutation_relations}, Lemma \ref{lem:asycurlyK} and a rearrangement of terms that for $z\in \partial D(r_n)$ and large $n$,
\begin{equation}\label{def:first_error}
\diag \left(1,z^{-\frac{B}{3}}\mathcal K(z)z^{\frac{B}{3}} \right)=n^{\widehat B}\left(I_4+ T_0 +E_n^{(1)}(z)\right) n^{-\widehat B},
\end{equation}
where the matrix $T_0$ is a strictly lower triangular constant matrix with first column zero, and the error term $E_n^{(1)}$ admits an asymptotic expansion of the form
\begin{equation}\label{def:En1}
E_n^{(1)}(z)\sim \sum_{j=1}^\infty \frac{(I_4+T_0)A_j^{(1)}}{n^{3j}z^j},\qquad z\in \partial D(r_n), \qquad n\to\infty,
\end{equation}
with the coefficients $A_j^{(1)}$ being independent of $n$. In \eqref{def:En1}, the factor $I_4+T_0$ is added just for later convenience, to avoid some coming cumbersome notations.

To explore the properties of $\widehat D_n$, we need the following basic fact.
\begin{lem}\label{lem:key_difference_estimate}
Suppose that $\{M_n(z)\}$ is a sequence of matrix-valued functions, analytic and uniformly bounded in a neighborhood $D(2\varepsilon)$ of the origin with $\varepsilon=\varepsilon_n\to 0$ as $n\to \infty$. Then, we have
\begin{equation}\label{eq:Mnest}
M_n(z)-M_n(w)=\Boh(\varepsilon^{-1} (z-w)),
\end{equation}
and
\begin{equation}\label{eq:w=0}
M_n(z)=M_n(0)+\Boh(\varepsilon^{-1} z),
\end{equation}
uniformly for $z,w\in \overline{D\left(\varepsilon\right)}$ as $n\to\infty$.
\end{lem}

\begin{proof}
The estimate \eqref{eq:w=0} follows immediately from \eqref{eq:Mnest}. To show \eqref{eq:Mnest}, we fix $z,w\in \overline{D(\varepsilon)}$ and use Cauchy's Theorem to write
$$
M_n(z)-M_n(w)=\frac{1}{2\pi i}\left( \oint_{|t|=2\varepsilon}\frac{M_n(t)}{t-z}\ud t-\oint_{|t|=2\varepsilon}\frac{M_n(t)}{t-w}\ud t  \right)=\frac{z-w}{2\pi i} \oint_{|t|=2\varepsilon} \frac{M_n(t)}{t-z}\frac{\ud t}{t-w}.
$$
Because $\{M_n(z)\}$ is uniformly bounded, the identity above immediately implies \eqref{eq:Mnest}.

This completes the proof of Lemma \ref{lem:key_difference_estimate}.

\end{proof}

We finally state the consequence for $\widehat D_n$ explicitly, as it will be used repeatedly in the next section.

\begin{prop}\label{prop:hatD_estimates}
The matrix-valued function $\widehat D_n(z)$ defined in \eqref{def:hat_Dn} is invertible. Furthermore, the matrices $\widehat D_n(z)^{\pm 1}$ are analytic near the origin and uniformly bounded for $z\in \overline{D(r_n)}$ as $n\to\infty$ with the estimates
%
\begin{equation}\label{eq:derivative_rule_Dhat}
\widehat D_n(z)^{\pm 1}-\widehat D_n(w)^{\pm 1}=\Boh(n^{3/2}(z-w)), \\
\end{equation}
and
\begin{equation}\label{eq:derivative_rule_Dhat2}
\widehat D_n(z)^{\pm 1}=\widehat D_n(0)^{\pm 1}+\Boh(n^{3/2}z), \\
\end{equation}
all valid uniformly for $z,w\in \overline{D\left(r_n\right)}$ as $n\to\infty$.
\end{prop}

\begin{proof}
The invertibility of $\widehat D_n$ follows immediately from its definition \eqref{def:hat_Dn}. For ease of notation, we will focus on $\widehat D_n$ in what follows, since the arguments for $\widehat D_n^{-1}$ are essentially the same.

Recalling the definition of $D_n$ given in \eqref{def:matrix_Dn}, the fact that $\widehat D_n$ is analytic near the origin follows from a direct application of the second part of Lemma~\ref{lem:entire_factor_expansion_infinity}. Moreover, under the scaling \eqref{eq:scaling_boundary}, the function $D_n(z)$ as well as $(n^{3/2}z)^{\pm B/3}$ remain uniformly bounded for $z\in \partial D(r_n)$ as $n\to\infty$. This implies that $\widehat D_n$ is uniformly bounded for $z\in \partial D(r_n)$ and, as a consequence of the maximum principle, also on the whole set $\overline{D(r_n)}$. The estimates \eqref{eq:derivative_rule_Dhat} and \eqref{eq:derivative_rule_Dhat2} are then immediate from Lemma~\ref{lem:key_difference_estimate}.
This completes the proof of Proposition \ref{prop:hatD_estimates}.
\end{proof}

\subsection{Construction of the local parametrix $L_0$}\label{sec:construction_matching}

With the above preparations, we are finally ready to build a solution to the RH problem \ref{rhp:local_parametrix} for $L_0$, construction which is carried out in five steps as explained next.

\subsubsection*{Initial step}

As the initial step, we define
\begin{equation}\label{eq:definition_L01}
L_0^{(1)}(z)=
\begin{cases}
\widehat G(z)n^{\frac{\widehat B}{2}}\widehat{D}_n(z)^{-1}n^{-\frac{\widehat B}{2}}\widehat P(z), & z\in D(r_n)\setminus \Gamma_S, \\
G(z), & z\in D(\delta)\setminus\left( \overline{D(r_n)}\cup (-\delta,\delta)\right),
\end{cases}
\end{equation}
where the matrices $\widehat G$ and $\widehat P$ are given in \eqref{def:matrix_G_hat} and \eqref{def:hatP}, respectively. We then have the following proposition.
\begin{prop}\label{prop:L01}
The matrix-valued function $L_0^{(1)}(z)$ defined in \eqref{eq:definition_L01} satisfies items (1)--(3) and the matching condition $\eqref{eq:local_param_decay_boundary}$ of the RH problem \ref{rhp:local_parametrix} for $L_0$. Moreover, we have, as $n\to \infty$,
\begin{equation}\label{eq:L01bnd}
L_{0,+}^{(1)}(z)=(I_4+\Boh(n^{1/2}))L_{0,-}^{(1)}(z), \qquad z\in \partial D(r_n)\setminus \Gamma_S.
\end{equation}
\end{prop}
\begin{proof}
Note that $\widehat P$ is analytic in $D(r_n)\setminus \Gamma_S$ and the global parametrix $G$ is analytic in $\mathbb{C} \setminus (-\infty, p]$. Thus, the analyticity properties of $L_0$ claimed in item (1) follows from the fact that both $\widehat G(z)$ and $\widehat D_n(z)^{-1}$ are analytic everywhere near $z=0$. The jumps claimed in item (2) follow from the jumps of $\widehat P$ and $G$ and again by the analyticity of $\widehat G(z)$ and $\widehat D_n(z)^{-1}$.  The local behavior of $L_0^{(1)}$ near $z=0$ can be seen from the behavior of $\widehat P$ near $z=0$ and it is also easily seen that the matching condition \eqref{eq:local_param_decay_boundary} is actually exact.

To show \eqref{eq:L01bnd}, we obtain from \eqref{eq:definition_L01}, \eqref{eq:error_P} and \eqref{def:first_error} that, for $z\in \partial D(r_n)\setminus \Gamma_S$ and $n \to \infty$,
\begin{equation}\label{eq:expansion_L01}
L_0^{(1)}(z) = \widehat G(z)n^{\frac{\widehat B}{2}}\widehat{D}_n(z)^{-1}n^{\frac{\widehat B}{2}}\left( I_4+T_0+E_n^{(1)}(z) \right)n^{-\frac{\widehat B}{2}}  \widehat{D}_n(z)n^{-\frac{\widehat B}{2}}\widehat G(z)^{-1}G(z).
\end{equation}
Since $T_0$ is a strictly lower triangular constant matrix with first column zero, we have
$$    n^{\frac{\widehat B}{2}}T_0n^{-\frac{\widehat B}{2}} = \Boh(n^{-1/2}),     $$
and
by \eqref{def:En1},
$$    n^{\frac{\widehat B}{2}}E_n^{(1)}(z)n^{-\frac{\widehat B}{2}} = \Boh(n^{-1/2}),  \qquad z\in  \partial D(r_n)\setminus \Gamma_S. $$
Inserting the above two estimates into \eqref{eq:expansion_L01}, we arrive at \eqref{eq:L01bnd} on account of the analiticity and boundedness in $n$ of both $\widehat G(z)$ and $\widehat D_n^{-1}(z)$ near the origin.

This completes the proof of Proposition \ref{prop:L01}.
\end{proof}

In view of \eqref{eq:L01bnd}, it follows that the matching condition \eqref{eq:local_param_decay_shrinking_boundary} is not satisfied. The next few steps are then devoted to refine the error term in \eqref{eq:L01bnd}.

\subsubsection*{Second step towards the matching}
In the second step, we eliminate the term $T_0$ in \eqref{eq:expansion_L01} by defining
\begin{align}\label{eq:definition_L02}
& L_0^{(2)}(z)
\nonumber
\\
& =
\begin{cases}
\begin{multlined}[b]
\widehat G(z)n^{\frac{\widehat B}{2}}\widehat{D}_n(z)^{-1}n^{\frac{\widehat B}{2}}
\vspace{-15pt}
\\
 \times \left(I_4+T_0 \right)^{-1}n^{-\frac{\widehat B}{2}}\widehat D_n(z)n^{-\frac{\widehat B}{2}}\widehat G(z)^{-1}L_0^{(1)}(z),
\end{multlined}
 & z\in D(r_n)\setminus \Gamma_S, \\
 & \\
                L_0^{(1)}(z)=G(z), & z\in D(\delta)\setminus\left( \overline{D(r_n)}\cup (-\delta,\delta)\right).p
\end{cases}
\end{align}
On account of the triangularity structure of $T_0$, we have that $T_0^3=0$, which implies
\begin{equation}\label{eq:inverse_T0}
(I_4+T_0)^{-1}=I_4-T_0+T_0^2.
\end{equation}
Thus, $L_0^{(2)}$ is well defined, and it has the follow properties.

\begin{prop}\label{prop:L02}
The matrix-valued function $L_0^{(2)}(z)$ defined in \eqref{eq:definition_L02} satisfies items (1)--(3) and the matching condition $\eqref{eq:local_param_decay_boundary}$ of the RH problem \ref{rhp:local_parametrix} for $L_0$. Moreover, we have, as $n\to \infty$,
\begin{equation}\label{eq:L02bnd}
L_{0,+}^{(2)}(z)=(I_4+\Boh(n^{1/2}))L_{0,-}^{(2)}(z), \qquad z\in \partial D(r_n)\setminus \Gamma_S.
\end{equation}
\end{prop}

To prove the above proposition and for later convenience, we need the following lemma, which is a version of the key observation  \cite[Proposition~5.15]{KM19} adapted to our setting.

\begin{lem}\label{lem:Ak2}
With $A_k^{(1)}$, $k=1,2,\ldots$, being the constant matrix in \eqref{def:En1}, define
\begin{equation}\label{def:Ak2}
A_k^{(2)}(z):=\widehat D_n(z)^{-1}n^{\frac{\widehat B}{2}} A^{(1)}_k n^{-\frac{\widehat B}{2}}\widehat D_n(z).
\end{equation}
Then, $A_k^{(2)}(z)$ is analytic near the origin, and we have, as $n \to \infty$, for any indices $k_1,k_2,\ldots,k_m$,
\begin{equation}\label{eq:estimate_Ak2}
A_{k_1}^{(2)}(z)A_{k_2}^{(2)}(z)\cdots A_{k_m}^{(2)}(z)=\Boh(n),
\end{equation}
uniformly for $z\in \overline{ D(r_n)}$, and
\begin{equation}\label{eq:estimate_Ak22}
A_{k_1}^{(2)}(z_1)A_{k_2}^{(2)}(z_2)\cdots A_{k_m}^{(2)}(z_m)=\Boh(n^{m-l}),
\end{equation}
uniformly for $z_1,\ldots, z_m \in \overline{D(r_n)}$, where
$$l:= \#\{j| 1 \leq j\leq m, z_j=z_{j+1} \}.$$
\end{lem}

\begin{proof}
The analyticity of $A_k^{(2)}(z)$ near the origin follows directly from its definition and the analyticity of $\widehat D_n(z)^{\pm 1}$.

By \eqref{def:Ak2}, it is readily seen that
$$
A_{k_1}^{(2)}(z)\cdots A_{k_m}^{(2)}(z) =\widehat D_n(z)^{-1}n^{\frac{\widehat B}{2}} A_{k_1}^{(1)}\cdots A_{k_m}^{(1)} n^{-\frac{\widehat B}{2}}\widehat D_n(z).
$$
This, together with the fact that $\widehat D_n(z)$ remains bounded for $z\in \overline{D(r_n)}$ as $n\to\infty$ (see Proposition~\ref{prop:hatD_estimates}) gives us \eqref{eq:estimate_Ak2}. The proof of \eqref{eq:estimate_Ak22} is similar to that of \eqref{eq:estimate_Ak2}, we omit the details here.

This completes the proof of Lemma \ref{lem:Ak2}.
\end{proof}

\paragraph{Proof of Proposition \ref{prop:L02}}
In \eqref{eq:definition_L02}, the factor multiplying $L_0^{(1)}$ to the left is analytic on $D(\delta)$, which then gives that $L_0^{(2)}$ still satisfies items (1)--(3) of the RH problem \ref{rhp:local_parametrix}, and the matching condition \eqref{eq:local_param_decay_boundary} is obvious as well.

To show \eqref{eq:L02bnd}, we see from  \eqref{eq:definition_L02}, \eqref{eq:expansion_L01} and \eqref{def:En1} that
\begin{align}
&L_{0,+}^{(2)}(z)L_{0,-}^{(2)}(z)^{-1}
\nonumber
\\
& =\widehat G(z)n^{\frac{\widehat B}{2}}\widehat{D}_n(z)^{-1}n^{\frac{\widehat B}{2}}\left(I_4+T_0 \right)^{-1}n^{-\frac{\widehat B}{2}}\widehat D_n(z)n^{-\frac{\widehat B}{2}}\widehat G(z)^{-1}L_0^{(1)}(z)G(z)^{-1}
\nonumber
\\
&=\widehat G(z)n^{\frac{\widehat B}{2}}\widehat{D}_n(z)^{-1}n^{\frac{\widehat B}{2}}\left(I_4+T_0\right)^{-1}(I_4+T_0+E_n^{(1)}(z))n^{-\frac{\widehat B}{2}}\widehat D_n(z)n^{-\frac{\widehat B}{2}}\widehat G(z)^{-1}\nonumber \\
&=\widehat G(z)\left(I_4+n^{\frac{\widehat B}{2}}\widehat{D}_n(z)^{-1}n^{\frac{\widehat B}{2}}\left(I_4+T_0\right)^{-1}E_n^{(1)}(z)n^{-\frac{\widehat B}{2}}\widehat D_n(z)n^{-\frac{\widehat B}{2}}\right)\widehat G(z)^{-1}
\nonumber
\\
& = \widehat G(z)\left(I_4+ \frac{n^{\frac{\widehat B}{2}}A_1^{(2)}(z)n^{-\frac{\widehat B}{2}}}{n^3z}+E_n^{(2)}(z) \right)\widehat G(z)^{-1},
\label{eq:matching_boundary_L02}
\end{align}
uniformly for $z\in \partial D(r_n)$ as $n\to \infty$, where
\begin{equation}\label{eq:expansion_E2}
E_n^{(2)}(z)\sim \sum_{k= 2}^\infty \frac{n^{\frac{\widehat B}{2}}A^{(2)}_k(z)n^{-\frac{\widehat B}{2}}}{n^{3k}z^k}, \qquad z\in \partial D(r_n), \qquad  n\to\infty.
\end{equation}
and $A_k^{(2)}(z)$ is defined in \eqref{def:Ak2}.

By \eqref{eq:estimate_Ak2} and \eqref{eq:scaling_boundary}, it follows that
\begin{equation*}
\frac{n^{\frac{\widehat B}{2}}A_1^{(2)}(z)n^{-\frac{\widehat B}{2}}}{n^3z}=\Boh(n^{1/2}), \qquad z\in \partial D(r_n), \qquad n\to\infty.
\end{equation*}
Similarly, we obtain from \eqref{eq:expansion_E2}, \eqref{eq:estimate_Ak2} and \eqref{eq:scaling_boundary} that
\begin{equation*}
E_n^{(2)}(z)=\Boh(n^{-1}), \qquad z\in \partial D(r_n), \qquad n\to\infty.
\end{equation*}
A combination of the above two estimates, \eqref{eq:matching_boundary_L02} and Proposition \ref{lem:prefactor_global_param} then gives us \eqref{eq:L02bnd}.

This completes the proof of Proposition \ref{prop:L02}. \qed

\subsubsection*{Third step towards the matching}

In the third step, we eliminate the growing term in \eqref{eq:matching_boundary_L02} by defining
\begin{align}\label{eq:definition_L03}
& L_0^{(3)}(z)=
\nonumber \\
& \begin{cases}
\widehat G(z)n^{\frac{\widehat B}{2}}\left(I_4-\dfrac{A^{(2)}_1(z)-A^{(2)}_1(0)}{n^3z}\right)n^{-\frac{\widehat B}{2}}\widehat G(z)^{-1}L_0^{(2)}(z), & z\in D(r_n)\setminus \Gamma_S, \\
\widehat G(z) n^{\frac{\widehat B}{2}}\left(I_4-\dfrac{A^{(2)}_1(0)}{n^3z}\right)^{-1} n^{-\frac{\widehat B}{2}}\widehat G(z)^{-1}L_0^{(2)}(z), & z\in D(\delta)\setminus\left( \overline{D(r_n)}\cup (-\delta,\delta)\right).
\end{cases}
\end{align}
Here, we observe from \eqref{eq:estimate_Ak2} and \eqref{eq:scaling_boundary} that, as $n\to \infty$,
\begin{equation}\label{eq:estA210}
\dfrac{A^{(2)}_1(0)}{n^3z}=
\begin{cases}
\Boh(n^{-1/2}), & \quad z\in \partial D(r_n),
\\
\Boh(n^{-2}), & \quad z\in \partial D(\delta),
\end{cases}
\end{equation}
which implies that the inverse of $I_4-\dfrac{A^{(2)}_1(0)}{n^3z}$ is well-defined in the definition of $L_0^{(3)}$. We then have the following proposition.

\begin{prop}\label{prop:L03}
The matrix-valued function $L_0^{(3)}(z)$ defined in \eqref{eq:definition_L03} satisfies items (1)--(3) and the matching condition $\eqref{eq:local_param_decay_boundary}$ of the RH problem \ref{rhp:local_parametrix} for $L_0$. Moreover, we have, as $n\to \infty$,
\begin{equation}\label{eq:L03bnd}
L_{0,+}^{(3)}(z)=(I_4+\Boh(1))L_{0,-}^{(3)}(z), \qquad z\in \partial D(r_n)\setminus \Gamma_S.
\end{equation}
\end{prop}
\begin{proof}
As before, the fact that the prefactors multiplying $L_0^{(2)}$ to the left are analytic makes sure that $L_0^{(3)}$, too, satisfies items  (1)--(3) of the RH problem \ref{rhp:local_parametrix}. Since $L_0^{(2)}$ already satisfies \eqref{eq:local_param_decay_boundary}, it is then easily seen from \eqref{eq:estA210} that $L_0^{(3)}$ satisfies \eqref{eq:local_param_decay_boundary} as well.

To show \eqref{eq:L03bnd}, we begin with some elementary estimates. From Lemma \ref{lem:Ak2} and \eqref{eq:expansion_E2}, it is easily seen that for $z_1,z_2 \in \partial D(r_n)$ and $n\to\infty$,
\begin{equation}\label{eq:estinL03}
n^{-\frac{\widehat B}{2}}E_n^{(2)}(z_1)n^{\frac{\widehat B}{2}}=\Boh(n^{-2}),
\qquad
\frac{A_1^{(2)}(z_1)A_1^{(2)}(z_2) }{n^6z_2^2}
=
\begin{cases}
\Boh(n^{-1}), & \quad z_1 \neq z_2,
\\
\Boh(n^{-2}), & \quad z_1=z_2.
\end{cases}
\end{equation}
Thus, it follows from \eqref{eq:estA210} and the above formula that
\begin{multline*}
\left(I_4-\dfrac{A^{(2)}_1(z)-A^{(2)}_1(0)}{n^3z}\right)\left(I_4+ \frac{A_1^{(2)}(z)}{n^3z}+n^{-\frac{\widehat B}{2}}E_n^{(2)}(z)n^{\frac{\widehat B}{2}} \right)
\\
=I_4 +\frac{A_1^{(2)}(0)}{n^3z} + \frac{A_1^{(2)}(0)A_1^{(2)}(z) }{n^6z^2}+\Boh(n^{-2}),
\end{multline*}
uniformly on $\partial D(r_n)$ as $n\to \infty$. Combining this with \eqref{eq:matching_boundary_L02}, \eqref{eq:definition_L03}, \eqref{eq:estA210} and \eqref{eq:estinL03}, we get that, as $n\to \infty$,
\begin{align}\label{eq:L03onDrn}
&L_{0,+}^{(3)}(z)L_{0,-}^{(3)}(z)^{-1}
\nonumber
\\
&=\widehat G(z)n^{\frac{\widehat B}{2}}\left(I_4-\dfrac{A^{(2)}_1(z)-A^{(2)}_1(0)}{n^3z}\right)n^{-\frac{\widehat B}{2}}\widehat G(z)^{-1}L_{0,+}^{(2)}(z)L_{0,-}^{(2)}(z)^{-1}
\nonumber
\\
& \quad \times \widehat G(z) n^{\frac{\widehat B}{2}}\left(I_4-\dfrac{A^{(2)}_1(0)}{n^3z}\right)n^{-\frac{\widehat B}{2}}\widehat G(z)^{-1}
 \nonumber
\\
& =\widehat G(z)n^{\frac{\widehat B}{2}} \left(I_4-\dfrac{A^{(2)}_1(z)-A^{(2)}_1(0)}{n^3z}\right)\left(I_4+ \frac{A_1^{(2)}(z)}{n^3z}+n^{-\frac{\widehat B}{2}}E_n^{(2)}(z)n^{\frac{\widehat B}{2}} \right)
\nonumber
\\
& \quad \times \left(I_4-\dfrac{A^{(2)}_1(0)}{n^3z}\right)n^{-\frac{\widehat B}{2}}\widehat G(z)^{-1}
\nonumber
\\
& =\widehat G(z)n^{\frac{\widehat B}{2}} \left(I_4 +\frac{A_1^{(2)}(0)}{n^3z} + \frac{A_1^{(2)}(0)A_1^{(2)}(z) }{n^6z^2}+\Boh(n^{-2})\right)\left(I_4-\dfrac{A^{(2)}_1(0)}{n^3z}\right)n^{-\frac{\widehat B}{2}}\widehat G(z)^{-1}
\nonumber
\\
& =
\widehat G(z) n^{\frac{\widehat B}{2}}\left( I_4 + \frac{A_1^{(2)}(0)A_1^{(2)}(z) }{n^6z^2} - \frac{A_1^{(2)}(0)A_1^{(2)}(z)A_1^{(2)}(0)}{n^{9}z^3} +\Boh(n^{-2}) \right)n^{-\frac{\widehat B}{2}}  \widehat G(z)^{-1},
\end{align}
uniformly for $z\in \partial D(r_n)$. Again, by Lemma \ref{lem:Ak2}, we see that
\begin{equation*}
n^{\frac{\widehat B}{2}} \frac{A_1^{(2)}(0)A_1^{(2)}(z) }{n^6z^2} n^{-\frac{\widehat B}{2}} =\Boh(1), \qquad n^{\frac{\widehat B}{2}} \frac{A_1^{(2)}(0)A_1^{(2)}(z)A_1^{(2)}(0)}{n^{9}z^3} n^{-\frac{\widehat B}{2}}=\Boh(n^{-1/2}),
\end{equation*}
uniformly on $\partial D(r_n)$ as $n\to \infty$. Inserting the above estimates into \eqref{eq:L03onDrn} gives us \eqref{eq:L03bnd}.

This completes the proof of Proposition \ref{prop:L03}.
\end{proof}

As a preparation for the next step, we now introduce some new functions to rewrite \eqref{eq:L03onDrn} in a convenient form. With $A_1^{(2)}(z)$ defined in \eqref{def:Ak2}, write
\begin{equation}\label{eq:identity_A12}
A_1^{(2)}(0)A_1^{(2)}(z)=z A_1^{(3)}(z)+A_1^{(2)}(0)^2,
\end{equation}
where
\begin{equation}\label{def:A13}
A_1^{(3)}(z):=\frac{1}{z}A_1^{(2)}(0)(A_1^{(2)}(z)-A_1^{(2)}(0))
\end{equation}
is analytic on $D(r_n)$. In view of Lemma \ref{lem:Ak2}, we have
\begin{equation}\label{eq:A13bnd}
A_1^{(3)}(z)=\Boh(n^{7/2}),
\end{equation}
and by applying Lemma~\ref{lem:key_difference_estimate} to the bounded analytic function $A_1^{(3)}(z)/n^{7/2}$, it follows that
\begin{equation}\label{eq:estimate_A31}
A_1^{(3)}(z)=A_1^{(3)}(0)+\Boh(n^{5}z),
\end{equation}
both of which being valid uniformly for $z\in \overline{D(r_n)}$ as $n\to\infty$. Combining \eqref{eq:identity_A12} with \eqref{eq:estinL03} then gives us
\begin{equation}\label{eq:decompA12}
\frac{A_1^{(2)}(0) A_1^{(2)}(z)}{n^6z^2}=\frac{A_1^{(3)}(z)}{n^6z}+\Boh(n^{-2}).
\end{equation}

In a way similar to \eqref{eq:identity_A12}, we also rewrite the other fraction in \eqref{eq:L03onDrn} into the form
\begin{equation}\label{decomposition:A32}
-\frac{A_1^{(2)}(0)A_1^{(2)}(z)A_1^{(2)}(0)}{n^{9}z^3} =
\frac{ A_2^{(3)}(z) }{n^{9}z}
+\frac{A_1^{(2)}(z)^3-A_1^{(2)}(z)^2A_1^{(2)}(0)-A_1^{(2)}(0)A_1^{(2)}(z)^2}{n^{9}z^3},
\end{equation}
where
\begin{equation}\label{def:A32}
A^{(3)}_2(z)=\frac{1}{z^2}(A_1^{(2)}(z)-A_1^{(2)}(0))A_1^{(2)}(z)(A_1^{(2)}(0)-A_1^{(2)}(z))
\end{equation}
is an analytic function near the origin. Applying the same arguments as in \eqref{eq:A13bnd} and \eqref{eq:estimate_A31}, it is readily seen that
\begin{equation}\label{eq:estimate_A32}
A_2^{(3)}(z)=\Boh(n^{6}),\qquad  A^{(3)}_2(z)=A_2^{(3)}(0)+\Boh(n^{15/2}z),
\end{equation}
which is valid, as always, uniformly for $z\in \overline{D(r_n)}$ as $n\to\infty$. Due to the decomposition \eqref{decomposition:A32}, we again obtain from Lemma \ref{lem:Ak2} that
\begin{equation}\label{eq:decompA32}
-\frac{A_1^{(2)}(0)A_1^{(2)}(z)A_1^{(2)}(0)}{n^{9}z^3} =
\frac{ A_2^{(3)}(z) }{n^{9}z}
+\Boh(n^{-5/2}).
\end{equation}
Inserting the estimates \eqref{eq:decompA12} and \eqref{eq:decompA32} into \eqref{eq:L03onDrn} we obtain
\begin{equation}\label{eq:L03onDrn2}
L_{0,+}^{(3)}(z)L_{0,-}^{(3)}(z)^{-1} =
\widehat G(z) n^{\frac{\widehat B}{2}}\left( I_4 + \frac{A_1^{(3)}(z) }{n^6z} +\frac{ A^{(3)}_2(z) }{n^{9}z} +\Boh(n^{-2}) \right) n^{-\frac{\widehat B}{2}}\widehat G(z)^{-1},
\end{equation}
uniformly valid for $z\in \partial D(r_n)$ as $n\to\infty$.

\subsubsection*{Fourth step towards the matching}

In a format already familiar to the reader, we define in the fourth step the following transformation:
\begin{align}\label{eq:definition_L04}
& L_0^{(4)}(z)=
\nonumber
\\
& \begin{cases}
\widehat G(z) n^{\frac{\widehat B}{2}} \left(I_4-\dfrac{A^{(3)}_1(z)-A^{(3)}_1(0)}{n^6z}\right) n^{-\frac{\widehat B}{2}} \widehat G(z)^{-1}L_0^{(3)}(z), & z\in D(r_n)\setminus \Gamma_S, \\
\widehat G(z)  n^{\frac{\widehat B}{2}} \left(I_4-\dfrac{A^{(3)}_1(0)}{n^6z}\right)^{-1} n^{-\frac{\widehat B}{2}} \widehat G(z)^{-1} L_0^{(3)}(z), & z\in D(\delta)\setminus\left( \overline{D(r_n)}\cup (-\delta,\delta)\right),
\end{cases}
\end{align}
where $A^{(3)}_1(z)$ is given in \eqref{def:A13}. In view of \eqref{eq:A13bnd}, it follows that, as $n\to \infty$,
\begin{equation*}
\dfrac{A^{(3)}_1(0)}{n^6z}=
\begin{cases}
\Boh(n^{-1}), & \quad z\in \partial D(r_n),
\\
\Boh(n^{-5/2}), & \quad z\in \partial D(\delta),
\end{cases}
\end{equation*}
which implies that the inverse of $I_4-\dfrac{A^{(3)}_1(0)}{n^6z}$ is well defined, and thus is $L_0^{(4)}$. Furthermore, we have the following proposition.
\begin{prop}\label{prop:L04}
The matrix-valued function $L_0^{(4)}(z)$ defined in \eqref{eq:definition_L04} satisfies items (1)--(3) and the matching condition $\eqref{eq:local_param_decay_boundary}$ of the RH problem \ref{rhp:local_parametrix} for $L_0$. Moreover, we have, as $n\to \infty$,
\begin{equation}\label{eq:L04bnd}
L_{0,+}^{(4)}(z)=(I_4+\Boh(n^{-1/2}))L_{0,-}^{(4)}(z), \qquad z\in \partial D(r_n)\setminus \Gamma_S.
\end{equation}
\end{prop}
\begin{proof}
It suffices to show \eqref{eq:L04bnd}, while the other claims can be verified directly. For $z\in \partial D(r_n)\setminus \Gamma_S$ and $n\to \infty$, it is readily seen from \eqref{eq:definition_L04} and \eqref{eq:L03onDrn2} that
\begin{align}\label{eq:L04onDrn}
&L_{0,+}^{(4)}(z)L_{0,-}^{(4)}(z)^{-1}
\nonumber
\\
&=\widehat G(z)n^{\frac{\widehat B}{2}}\left(I_4-\dfrac{A^{(3)}_1(z)-A^{(3)}_1(0)}{n^6z}\right)\left( I_4 + \frac{A_1^{(3)}(z) }{n^6z} +\frac{ A^{(3)}_2(z) }{n^{9}z} +\Boh(n^{-2}) \right)
\nonumber
\\
& \quad \times \left(I_4-\dfrac{A^{(3)}_1(0)}{n^6z}\right)n^{-\frac{\widehat B}{2}}\widehat G(z)^{-1}
 \nonumber
\\
& = \widehat G(z)n^{\frac{\widehat B}{2}} \left(I_4 +\frac{A_1^{(3)}(0)}{n^6z}+\frac{A^{(3)}_2(z)}{n^9z}+\Boh(n^{-2})\right)\left(I_4-\dfrac{A^{(3)}_1(0)}{n^6z}\right)n^{-\frac{\widehat B}{2}}\widehat G(z)^{-1}
\nonumber
\\
& = \widehat G(z)n^{\frac{\widehat B}{2}} \left(I_4 +\frac{A^{(3)}_2(z)}{n^9z}+\Boh(n^{-2})\right)n^{-\frac{\widehat B}{2}}\widehat G(z)^{-1},
\end{align}
where for the second and third equality we have made use of the estimates \eqref{eq:A13bnd}, \eqref{eq:estimate_A31} and \eqref{eq:estimate_A32} to suppress the error terms. By \eqref{eq:estimate_A32}, we further have
$$ \frac{A^{(3)}_2(z)}{n^9z}=\Boh(n^{-3/2}), \qquad z\in \partial D(r_n), \qquad n\to\infty, $$
which, together with \eqref{eq:L04onDrn}, yields \eqref{eq:L04bnd}.

This completes the proof of Proposition \ref{prop:L04}.
\end{proof}

\subsubsection*{Last step towards the matching}

As the fifth and last step, we modify $L^{(4)}_0$ to
\begin{align}\label{eq:definition_L05}
& L_0(z)=L_0^{(5)}(z) =
\nonumber
\\
&\begin{cases}
\widehat G(z) n^{\frac{\widehat B}{2}} \left(I_4-\dfrac{A^{(3)}_2(z)-A^{(3)}_2(0)}{n^9 z}\right) n^{-\frac{\widehat B}{2}} \widehat G(z)^{-1}L_0^{(4)}(z), & z\in D(r_n)\setminus \Gamma_S, \\
\widehat G(z)  n^{\frac{\widehat B}{2}} \left(I_4-\dfrac{A^{(3)}_2(0)}{n^9z}\right)^{-1} n^{-\frac{\widehat B}{2}} \widehat G(z)^{-1} L_0^{(4)}(z), & z\in D(\delta)\setminus\left( \overline{D(r_n)}\cup (-\delta,\delta)\right),
\end{cases}
\end{align}
where $A^{(3)}_2(z)$ is given in \eqref{def:A32}. In view of \eqref{eq:estimate_A32}, it follows that, as $n\to \infty$,
\begin{equation*}
\dfrac{A^{(3)}_2(0)}{n^9 z}=
\begin{cases}
\Boh(n^{-3/2}), & \quad z\in \partial D(r_n),
\\
\Boh(n^{-3}), & \quad z\in \partial D(\delta),
\end{cases}
\end{equation*}
which implies that the inverse of $I_4-\dfrac{A^{(3)}_2(0)}{n^9z}$ is well defined, thus so is $L_0^{(5)}$. Following the same arguments as in the proof of Proposition \ref{prop:L04}, it is straightforward to conclude that
\begin{prop}\label{prop:L05}
The matrix-valued function $L_0^{(5)}(z)$ defined in \eqref{eq:definition_L05} solves the RH problem \ref{rhp:local_parametrix} for $L_0$.
\end{prop}

This completes the construction of the local parametrix near the origin.

\section{Final transformation $S\mapsto R$} \label{sec:StoR}

With the global parametrix $G$ given in Proposition~\ref{prop:GlobCons}, the local parametrices $L_p$ and $L_{-q}$ near $p$ and $-q$ briefly discussed in Section~\ref{section:airy_parametrices} and the local parametrix $L_0$ near the origin constructed in \eqref{eq:definition_L05},  the final transformation is defined by
\begin{equation}\label{def:R}
R(z)=
\begin{cases}
S(z)L_0(z)^{-1},& z \in D(\delta), \\
S(z)L_p(z)^{-1}, & z\in D_p(\delta),\\
S(z)L_{-q}(z)^{-1}, & z\in D_{-q}(\delta), \\
S(z)G(z)^{-1}, & \mathbb{C}\setminus D_R,
\end{cases}
\end{equation}
where
$$D_R:=D(\delta) \cup D_p(\delta) \cup D_{-q}(\delta).$$
It is then straightforward to check that $R$ satisfies the following RH problem.
\begin{rhp}\label{rhp:R}
The function $R$ defined in \eqref{def:R} has the following properties:
\begin{enumerate}[(1)]
\item $R$ is defined and analytic in $\C\setminus \Gamma_R$,
where
\begin{multline*}
\Gamma_R:=\partial D_R \cup \partial D(r_n) \cup (-q+\delta,0)\cup (p+\delta,+\infty)
\\ \cup \left( \bigcup\limits_{j=1}^3\partial \mathcal L_j^\pm \right) \setminus \left(D(r_n) \cup D_p(\delta)\cup  D_{-q}(\delta)\right)
\end{multline*}
with the orientations as illustrated in Figure \ref{fig:contour_R}.
\item For $z\in\Gamma_R$, $R$ satisfies the jump condition
$$ R_+(z)=R_-(z)J_R(z),$$
where
\begin{equation}\label{eq:jumps_R}
J_R(z)=
\begin{cases}
G(z)J_S(z)G(z)^{-1}, & z\in \Gamma_R\setminus \left(D_R \cup (-q+\delta,0)\right),  \\
L_0(z)J_S(z)L_0(z)^{-1}, & z\in \bigcup\limits_{j=2,3}\left(\partial \mathcal L_j^\pm \cap D(\delta) \setminus D(r_n) \right),  \\
I_4 -|z|^\kappa e^{n\phi_1(z)}G_-(z)E_{21}G_-(z)^{-1} ,
& z\in (-q+\delta,-\delta),\\
I_4 -|z|^\kappa e^{n\phi_1(z)}L_{0,-}(z)E_{21}L_{0,-}(z)^{-1} ,
& z\in (-\delta,0),\\
G(z)L_{p}(z)^{-1}, & z\in \partial D_{p}(\delta), \\
G(z)L_{-q}(z)^{-1}, & z\in \partial D_{-q}(\delta), \\
G(z)L_{0}(z)^{-1}, & z\in \partial D(\delta), \\
L_{0,-}(z)L_{0,+}(z)^{-1}, & z\in \partial D(r_n),
\end{cases}
\end{equation}
and where $J_S(z)$ is given in \eqref{eq:jumps_S}.
\item As $z\to \infty$, we have
$$
R(z)=I_4+\Boh(z^{-1}).
$$
\end{enumerate}
\end{rhp}

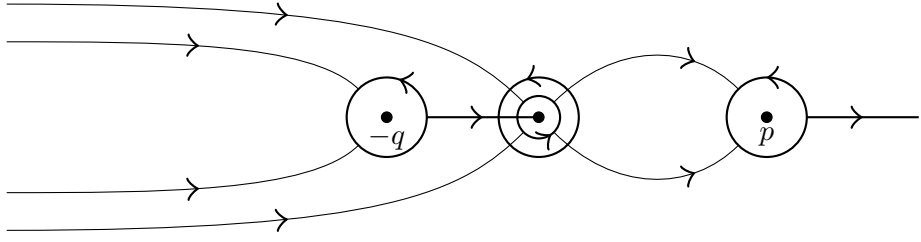
\begin{figure}[t]
\begin{center}
\begin{tikzpicture}
%
\filldraw [black] (3,3) circle (2pt);
\filldraw [black]  (5,3) circle (2pt);
\filldraw [black] (8,3) circle (2pt);
\node[below] (q) at (3,3) {$-q$};
\node[below] (p) at (8,3) {$p$};
\draw[thick,->-=.2] (3,3) circle (15pt);
\draw[thick,->-=.3] (5,3) circle (15pt);
\draw[thick,->-=.85] (5,3) circle (8pt);
\draw[thick,->-=.25] (8,3) circle (15pt);
%
%
\path (8,3) ++(135:15pt) coordinate (pplus);
\path (8,3) ++(-135:15pt) coordinate (pminus);
\path (8,3) ++(180:15pt) coordinate (pl);
\path (8,3) ++(0:15pt) coordinate (pr);
\path (5,3) ++(45:8pt) coordinate (oplusr);
\path (5,3) ++(-45:8pt) coordinate (ominusr);
\path (5,3) ++(135:8pt) coordinate (oplusl);
\path (5,3) ++(-135:8pt) coordinate (ominusl);
\path (5,3) ++(180:15pt) coordinate (ol);
\path (5,3) ++(0:15pt) coordinate (or);
\path (3,3) ++(135:15pt) coordinate (qplusl);
\path (3,3) ++(-135:15pt) coordinate (qminusl);
\path (3,3) ++(180:15pt) coordinate (ql);
\path (3,3) ++(0:15pt) coordinate (qr);
%
%
%
\draw[thick,->-=.5] (qr)--(5,3);
\draw[thick,->-=.5] (pr)--(10,3);
%
%
\draw[->-=.75] (oplusr) .. controls (6,4) and (7,4) .. (pplus);
\draw[->-=.75] (ominusr) .. controls (6,2) and (7,2) .. (pminus);
\draw[-<-=.5] (qplusl) .. controls (2,4) and (0,4) .. (-2,4);
\draw[-<-=.5] (qminusl) .. controls (2,2) and (0,2) .. (-2,2);
\draw[-<-=.5] (oplusl) .. controls (4,4) and (3,4.5) .. (-2,4.5);
\draw[-<-=.5] (ominusl) .. controls (4,2) and (3,1.5) .. (-2,1.5);
%
\end{tikzpicture}
\end{center}
\caption{The jump contours for the matrix $R$.}
\label{fig:contour_R}
\end{figure}

It comes out that the jump matrix of $R$ on each jump contour tends to the identity matrix for large $n$ with the convergence rate given in the next lemma.
\begin{lem}\label{lem:estJR}
Let $J_R(z)$ be defined in \eqref{eq:jumps_R}. There exists two positive constants $c_1,c_2$ such that, as $n\to \infty$,
\begin{equation}\label{eq:estJR}
J_R(z)=
\begin{cases}
I_4+\Boh(n^{-1}), & z\in \partial D_R \cup \partial D(r_n),
\\
I_4+\Boh(e^{-c_1n^{1/2}}), & z\in \bigcup\limits_{j=2,3}\left(\partial \mathcal L_j^\pm \cap D(\delta) \setminus D(r_n) \right),
\\
I_4+\Boh(e^{-c_2n}), &  \mbox{elsewhere on $\Gamma_R$},
\end{cases}
\end{equation}
uniformly for $z$ on the indicated contours.
\end{lem}
\begin{proof}
By \eqref{eq:jumps_R}, the first estimate of $J_R(z)$ in \eqref{eq:estJR}, that is, on the boundaries of the four disks $\partial D_R\cup \partial D(r_n)$, follows directly from \eqref{eq:matchingp}, \eqref{eq:local_param_decay_boundary} and \eqref{eq:local_param_decay_shrinking_boundary}.

For the estimate of $J_R(z)$ on $\bigcup\limits_{j=2,3}\left(\partial \mathcal L_j^\pm \cap D(\delta) \setminus D(r_n) \right)$, we first focus on the case $j=2$. From \eqref{eq:jumps_R} and \eqref{eq:jumps_S}, it follows that
\begin{equation}\label{eq:JRL2}
J_R(z)=I_4+z^{-\kappa}e^{n\phi_2(z)}L_0(z)E_{32}L_0(z)^{-1},\qquad z\in  \partial \mathcal L_2^\pm \cap D(\delta) \setminus D(r_n).
\end{equation}
Since $L_0(z)^{\pm 1}$ has at most power log singularities near the origin, the estimate of $J_R(z)$ then essentially relies on the behavior of $e^{n \phi_2 (z)}$ near $z=0$. In view of  \eqref{def:lambda_functions2}, \eqref{def:lambda_functions3}, \eqref{def:phi_functions2}, \eqref{eq:lambdaandhatlambda} and Proposition \ref{prop:hatlambda}, we have, for $z\in  \partial \mathcal L_2^\pm \cap D(\delta) \setminus D(r_n)$,
\begin{multline*}
|e^{n \phi_2 (z)}|=|e^{n(\lambda_2(z)-\lambda_3(z))}|=|e^{n(\widehat \lambda_2^{\pm}(z)-\widehat \lambda_3^{\pm}(z))}|
\\
=|e^{n((\omega^{\pm}-\omega^{\mp})f_4(0)z^{1/3}+\Boh(z^{2/3}))}|\leq e^{-cn|z|^{1/3}},
\end{multline*}
for some $c>0$. Note that $|z|>n^{-3/2}$ on the annulus $D(\delta) \setminus D(r_n)$, which together with the above estimate and \eqref{eq:JRL2} gives us
\begin{equation*}
J_R(z)= I_4+\Boh(e^{-cn^{1/2}}),  \qquad z\in \partial \mathcal L_2^\pm \cap D(\delta) \setminus D(r_n),
\end{equation*}
for large $n$. If $z\in \partial \mathcal L_3^\pm \cap D(\delta) \setminus D(r_n)$, the estimate of $J_R(z)$ can be derived in a similar manner, where one needs to explore the behavior of $e^{n\phi_3(z)}$ near $z=0$. We omit the details.

Finally, for $z$ belonging to other parts of $\Gamma_R$, we note from \eqref{eq:jumps_R} and \eqref{eq:jumps_S} that, if $z \in (p+\delta, +\infty) $,
$$J_{R}(z)=I_4+z^{\kappa}e^{-n\phi_2(z)}G(z)E_{23}G(z)^{-1}.$$
Since $G(z)^{\pm 1}$ grows at most in a power law for large $z$ (see \eqref{eq:asympt_behavior_global_param}) and $\phi_2(z)>c$ for some $c>0$ on  $(p+\delta, +\infty) $ (see \eqref{eq:inequalities_phi_off_support}),
it is immediate to conclude from the above formula that
\begin{equation*}
J_R(z)= I_4+\Boh(e^{-cn}),  \qquad z\in (p+\delta,+\infty),
\end{equation*}
for large $n$. The estimate of $J_R(z)$ on $(-q+\delta,0)\cup \left( \bigcup\limits_{j=1}^3\partial \mathcal L_j^\pm \right) \setminus D_R$ can be obtained by applying similar arguments.

A little extra effort is needed to handle the case $z\in(-\delta,0)$. Similarly as above, from \eqref{eq:jumps_R} and \eqref{eq:inequalities_phi_off_support}, it suffices to show that $L_{0}(z)E_{21}L_{0}(z)^{-1}$ has power growth in $n$. To see this, from the definition of $L_0$ given in \eqref{eq:definition_L05}, and tracing back the transformations $L_0^{(5)}\mapsto L_0^{(4)}\mapsto \hdots\mapsto L_0^{(1)}$, it is readily seen that
$$
L_0(z)=L_0^{(5)}(z)=\mathcal{A}_n(z)\widehat P(z), \qquad |z|< r_n,
$$
where the prefactor $\mathcal{A}_n$ is analytic and invertible near the origin, with $\mathcal A_n$ and $\mathcal A_n^{-1}$ having at worse power growth as $n\to\infty$, and $\widehat P$, defined in \eqref{def:hatP}, contains the Meijer-G parametrix. From the structure of $\widehat P$ (see the first identity in \eqref{eq:asymptotics_parametrix_P}) and from \eqref{eq:zeropsi} we see that, as $z\to 0$,
\begin{equation*}
\widehat P(z)\begin{pmatrix}
0 & 1 & 0 & 0
\end{pmatrix}^{T}=\Boh(1), \qquad \begin{pmatrix}
1 & 0 & 0 & 0
\end{pmatrix}\widehat P(z)^{-1}=\Boh(1).
\end{equation*}
Thus,
\begin{equation*}
L_{0}(z)E_{21}L_{0}(z)^{-1}=\mathcal{A}_n(z)\widehat P(z)\begin{pmatrix}
0 & 1 & 0 & 0
\end{pmatrix}^{T}\begin{pmatrix}
1 & 0 & 0 & 0
\end{pmatrix}\widehat P(z)^{-1}\mathcal{A}_n(z)^{-1}
\end{equation*}
has at worse power growth as $n\to\infty$.

This completes the proof of Lemma \ref{lem:estJR}.
\end{proof}

As a consequence of the above lemma, we conclude from the standard arguments in the RH analysis (cf.\cite{deift_book} and \cite[Appendix~A]{BK07}) that
\begin{equation}\label{eq:estR}
R(z)=I_4+\Boh(n^{-1}),\qquad n \to \infty,
\end{equation}
uniformly for $z\in \C\setminus \Gamma_R$.

\section{Proofs of asymptotic results}\label{sec:proofasy}

In this section, we will prove Theorems \ref{thm:limitingmeandistri} and \ref{thm:hardedge} by inverting the transformations \eqref{eq:transformations}.

\subsection{Proof of Theorem \ref{thm:limitingmeandistri}}

Let $x,y \in \Delta_2=(0,p)$ be fixed. In view of the representation of $K_n$ given in \eqref{kernel representation}, and having in mind \eqref{def:w1}--\eqref{def:w2} and the calculation \eqref{eq:relations_A_W}, we obtain from \eqref{eq:YtoX} and a straightforward calculation that
\begin{align*}
& n^2K_{n}\left(n^2x,n^2y \right)
\nonumber
\\
& = \frac{1}{2\pi i(x-y)}\begin{pmatrix}0 &0 & w_2(y) \end{pmatrix}\diag\left(A_1(y)y^{\frac{\kappa}{2}\sigma_3},A_2(y) y^{\frac{\kappa-\nu}{2}\sigma_3}\right)
\nonumber
\\
&\quad \times X_{+}(y)^{-1}X_{+}(x)\diag\left(x^{\frac{\kappa}{2}\sigma_3}A_1(x)^{-1},x^{\frac{\nu-\kappa}{2}\sigma_3}A_2(x)^{-1}\right)
\begin{pmatrix}
w_1(x)& 0 & 0
\end{pmatrix}^T
\nonumber
\\
& = \frac{1}{2\pi i(x-y)}
\begin{pmatrix}
0 &0 & 1 &
0 \end{pmatrix}
X_{+}(y)^{-1}X_{+}(x)
\begin{pmatrix}
0 & x^{\kappa} & 0 & 0
\end{pmatrix}^T.
\end{align*}
From \eqref{eq:XtoT}, this becomes
\begin{align*}
 n^2K_{n}\left(n^2x,n^2y \right)
= \frac{1}{2\pi i(x-y)}
\begin{pmatrix}
0 &0 & e^{n \lambda_{3,+}(y)} &
0 \end{pmatrix}
T_{+}(y)^{-1}T_{+}(x)
\begin{pmatrix}
0 \\ x^{\kappa}e^{-n\lambda_{2,+}(x)} \\ 0 \\ 0
\end{pmatrix}.
\end{align*}
A further appeal to \eqref{eq:TtoS} and Proposition \ref{prop:lamdaphi} yields
\begin{align}\label{kernel representation4}
& n^2K_{n}\left(n^2x,n^2y \right)
\nonumber
\\
& = \frac{1}{2\pi i(x-y)}
\begin{pmatrix}
0  & -y^{-\kappa}e^{n (\phi_{2,+}(y)+\lambda_{3,+}(y))} & e^{n \lambda_{3,+}(y)} &
0 \end{pmatrix}
S_{+}(y)^{-1}
\nonumber
\\
& \quad \times
S_{+}(x)
\begin{pmatrix}
0 & x^{\kappa}e^{-n\lambda_{2,+}(x)} & e^{n(\phi_{2,+}(x)-\lambda_{2,+}(x))}  & 0
\end{pmatrix}^T
\nonumber
\\
& = \frac{1}{2\pi i(x-y)}
\begin{pmatrix}
0  & -y^{-\kappa}e^{n \lambda_{2,+}(y)} & e^{n \lambda_{3,+}(y)} &
0 \end{pmatrix}
S_{+}(y)^{-1}
\nonumber
\\
& \quad \times S_{+}(x)
\begin{pmatrix}
0 & x^{\kappa}e^{-n\lambda_{2,+}(x)} & e^{-n\lambda_{3,+}(x)}  & 0
\end{pmatrix}^T.
\end{align}

Since both $x$ and $y$ are fixed, we may assume that $\delta$ is chosen so as that $x$ and $y$ are outside the discs around the edges $0$ and $p$. From \eqref{eq:estR} and the analyticity of $G_+$ away from $0$ and $p$, we obtain that
\begin{equation}\label{eq:TyTx}
S_{+}(y)^{-1}S_{+}(x)=I_4+\Boh(x-y), \qquad x \to y,
\end{equation}
uniformly for $x,y\in [\delta,p-\delta]$ as $n\to\infty$. Next, noticing that $\lambda_{2,\pm}(x)=\lambda_{3,\mp}(x)$ for $x\in \Delta_2$ (see Proposition \ref{prop:lamdaphi}), by taking $y \to x$, it then follows from \eqref{kernel representation4}, \eqref{eq:TyTx}, L'H\^{o}pital's rule, \eqref{def:lambda_functions3} and \eqref{eq:definition_xi_functions} that
\begin{align*}
 n^2K_{n}\left(n^2x,n^2x \right)
&= -\frac{n}{2\pi i}(\xi_{3,+}(x)-\xi_{3,-}(x))+\Boh(1)=\frac{n}{2\pi i}\left(C^{\mu_2}_{+}(x)-C^{\mu_2}_{-}(x)\right)+\Boh(1)
\nonumber
\\
&= n\frac{\ud \mu_2}{\ud x}(x)+\Boh(1),
\end{align*}
which implies that
\begin{equation*}
nK_{n}\left(n^2 x, n^2 x\right)=\frac{\ud \mu_2}{\ud x}(x)\left( 1+\Boh(n^{-1}) \right),
\end{equation*}
uniformly for $x\in (\delta,p-\delta)$ as $n\to\infty$.
Similarly, it can be shown that
\begin{equation*}
\lim_{n\to\infty}nK_{n}\left(n^2 x, n^2 x\right)=0,
\qquad x>p,
\end{equation*}
as desired.

This completes the proof of Theorem \ref{thm:limitingmeandistri} away from the endpoints $x=p$ and $x=0$. The case for $x=p$ can be handled similarly, with the Airy parametrix appearing instead of the global parametrix $G$, and with a worse error term.
\qed

\subsection{Proof of Theorem \ref{thm:hardedge}}

To prove Theorem~\ref{thm:hardedge}, let us start with $u$ and $v$ in the shrinking interval $(0,r_n)$ and trace back all the transformations \eqref{eq:transformations}.
The transformations $Y\mapsto X$ and $X\mapsto T$, given in \eqref{eq:YtoX} and \eqref{eq:XtoT} respectively, are defined globally, whereas the transformation $T\mapsto S$, defined in \eqref{eq:TtoS}, is the same on the plus side of $\Delta_2$. Thus, even for $u,v$ in the shrinking interval $(0,r_n)$, it holds
\begin{multline}\label{eq:hard_edge_kernel1}
n^2K_n(n^2u,n^2v) = \frac{1}{2\pi i(u-v)}
\begin{pmatrix}
0  & -v^{-\kappa}e^{n \lambda_{2,+}(v)} & e^{n \lambda_{3,+}(v)} &
0 \end{pmatrix} \\
\times
S_{+}(v)^{-1}S_{+}(u)
\begin{pmatrix}
0 & u^{\kappa}e^{-n\lambda_{2,+}(u)} & e^{-n\lambda_{3,+}(u)}  & 0
\end{pmatrix}^T;
\end{multline}
see \eqref{kernel representation4}. Using now the transformation \eqref{def:R} on $D(r_n)\subset D(\delta)$, we obtain
\begin{equation}\label{eq:hard_edge_relation_S_R}
S_+(v)^{-1}S_{+}(u)=L_{0,+}(v)^{-1}R(v)^{-1}R(u)L_{0,+}(u).
\end{equation}
We now scale
$$u=u_n=\frac{x}{n^3(\beta^2-\alpha^2) }\quad \mbox{and} \quad v=v_n=\frac{y}{n^3(\beta^2-\alpha^2)},$$
where $x,y$ are in fixed compact subsets of $(0,\infty)$. Note that with this scaling the points $u$ and $v$ fall inside $(0,r_n)$ and the calculations above are bona fide.
To estimate \eqref{eq:hard_edge_relation_S_R} for large $n$, we will need the following lemma, which is a refined version of Lemma~\ref{lem:key_difference_estimate}.

\begin{lem}\label{lem:key_estimate_improvement}
Suppose that $\{M_n\}$ is a sequence of matrix-valued functions satisfying the conditions of Lemma~\ref{lem:key_difference_estimate} and for which there exists a bounded sequence of constant matrices $\{\widetilde M_n\}$ for which
\begin{equation}\label{eq:estimate_Mn_improved}
M_n(z)-\widetilde M_n=\Boh(\delta_n),\qquad n\to\infty,
\end{equation}
uniformly for $z\in\partial D(2\varepsilon_n)$, where $\{\delta_n\}$ is a sequence of bounded positive numbers (possibly with $\delta_n\to 0$ but not necessarily). Then, \eqref{eq:Mnest} can be improved to
$$
M_n(z)-M_n(w)=\Boh\left( \frac{\delta_n}{\varepsilon_n}(z-w) \right),\qquad n\to\infty,
$$
uniformly for $z,w\in \overline{D(\varepsilon_n)}$.
\end{lem}
\begin{proof}
Similarly as in the proof of Lemma~\ref{lem:key_difference_estimate}, we write
$$
M_n(z)-M_n(w)=M_n(z)-\widetilde M_n-(M_n(w)-\widetilde M_n)=\frac{z-w}{2\pi i}\oint_{|t|=2\varepsilon}\frac{M_n(t)-\widetilde M_n}{t-z}\frac{\ud w}{t-w}.
$$
It remains to estimate the numerator using \eqref{eq:estimate_Mn_improved}, and the lemma follows.
\end{proof}

We start estimating $R$. The following lemma also appears in \cite[Lemma~6.5]{KM19}, although the proof has to be slightly modified to account for the jump of $R$ along $(-r_n,0)$ that appears here but not in the mentioned work.

\begin{lem}\label{lem:Rest}
The matrix $R$ satisfies
$$
R(v_n)^{-1}R(u_n)=I_4+\Boh(n^{-5/2}(x-y))
$$
uniformly for $x,y$ in compact subsets of $(0,\infty)$.
\end{lem}
\begin{proof}
For $\gamma$ being any contour for which $R$ is analytic in its interior and encircling $u_n$ and $v_n$ counter-clockwise, we write with the help of Cauchy's integral formula
$$
R(u_n)-R(v_n)=R(u_n)-I_4-(R(v_n)-I_4)=\frac{u_n-v_n}{2\pi i}\oint_\gamma \frac{R(s)-I_4}{(s-u_n)(s-v_n)} \ud s.
$$
We apply this to $\gamma$ being the boundary of  the slit disk $D(r_n/2)\setminus (-r_n/2,0]$ and obtain
\begin{equation}\label{eq:estimate_difference_R}
R(u_n)-R(v_n) = \frac{u_n-v_n}{2\pi i}\left(\oint_{|s|=\frac{r_n}{2}} \frac{R(s)-I_4}{(s-u_n)(s-v_n)} \ud s + \int_{-\frac{r_n}{2}}^{0} \frac{R_+(s)-R_-(s)}{(s-u_n)(s-v_n)} \ud s   \right).
\end{equation}
For $s\in (-r_n/2,0]$, it is readily seen from Lemma~\ref{lem:estJR} that
$$
R_+(s)-R_-(s)=R_-(s)(J_R(s)-I_4)=\Boh(e^{-c_2n}),
$$
where we have also made use of the fact that $R$ remains uniformly bounded near $0$. Moreover, since $|s-u_n|, |s-v_n|$ decay with $\Boh(n^{-3})$ along the interval $(-r_n/2,0]$, it follows that
$$
\frac{u_n-v_n}{2\pi i}\int_{-\frac{r_n}{2}}^{0} \frac{R_+(s)-R_-(s)}{(s-u_n)(s-v_n)} \ud s =\Boh\left( (x-y)e^{-cn} \right),
$$
for some constant $c>0$. The first integral in \eqref{eq:estimate_difference_R} can be estimated from \eqref{eq:estR} and using the same approach in the proof of Lemma~\ref{lem:key_estimate_improvement}, allowing us to conclude that
$$
R(u_n)-R(v_n)=\Boh((x-y)n^{-5/2}),
$$
uniformly for $x,y$ in compact subsets of $(0,\infty)$. To conclude the lemma, simply write
$$
R(v_n)^{-1}R(u_n)=I_4+R(v_n)^{-1}(R(u_n)-R(v_n))
$$
and use that $R$ remains bounded near the origin.

This completes the proof of Lemma \ref{lem:Rest}.
\end{proof}

Next, we need to estimate $L_0$, which is more cumbersome. We start by spelling it out after unraveling the transformations $L_0=L_0^{(5)}\mapsto L_0^{(1)}$, which are given in Section~\ref{sec:construction_matching}, giving us that
\begin{equation}
\label{eq:def_hat_L0}
\begin{aligned}
L_0(z)&=\widehat L_0(z)\widehat P(z),
 \\
\widehat L_0(z)&:=\widehat G(z)n^{\frac{\widehat B}{2}}\mathcal A^{(3)}(z)\mathcal A^{(2)}(z)\mathcal A^{(1)}(z)\widehat D_n(z)^{-1}n^{\frac{\widehat B}{2}}\left(I_4+T_0\right)^{-1}n^{-\widehat B},
\end{aligned}
\end{equation}
with
$$
\mathcal A^{(j)}(z)=I_4-\frac{A_1^{(j+1)}(z)-A_1^{(j+1)}(0)}{n^{3j}z},\quad j=1,2,\qquad \mathcal A^{(3)}(z)=I_4-\frac{A_2^{(3)}(z)-A_2^{(3)}(0)}{n^9z}.
$$

\begin{lem}\label{lem:L0est}
The matrix-valued function $\widehat L_0(z)$ defined in \eqref{eq:def_hat_L0} satisfies
$$
\widehat L_0(v_n)^{-1}\widehat L_0(u_n)=n^{\frac{\widehat B}{2}}\left(I_4+\Boh(n^{-3/2}(x-y))\right)n^{-\frac{\widehat B}{2}},\qquad n\to \infty,
$$
uniformly for $x,y$ in compact subsets of $(0,\infty)$.
\end{lem}
\begin{proof}
The analyticity of $\widehat G$ and its inverse near the origin (recall Proposition~\ref{lem:prefactor_global_param}) and the fact that they do not depend on $n$ gives
\begin{multline}\label{eq:estimate_difference_hatG}
\widehat G(v_n)^{-1}\widehat G(u_n)=I_4+\widehat G(v_n)^{-1}(\widehat G(u_n)-\widehat G(v_n))\\ =I_4+\Boh(u_n-v_n)=I_4+\Boh((x-y)n^{-3}),\qquad n\to\infty,
\end{multline}
uniformly for $x,y\in (0,\infty)$.

The function $\mathcal A^{(3)}$ is analytic in a neighborhood of $\overline{D(r_n)}$, and in virtue of \eqref{eq:estimate_A32},
$$
\mathcal A^{(3)}(z)-I_4=\Boh(n^{-3/2}),
$$
so from Lemma~\ref{lem:key_estimate_improvement} with $\delta_n=n^{-3/2}=\varepsilon_n$ and $\widetilde M_n=I_4$, it follows that
$$
\mathcal A^{(3)}(u_n)-\mathcal A^{(3)}(v_n)=\Boh(n^{-3}(x-y)),
$$
and consequently as in \eqref{eq:estimate_difference_hatG}
\begin{equation}\label{eq:estimate_difference_curly_A3}
\mathcal A^{(3)}(v_n)^{-1}\mathcal A^{(3)}(u_n)=I_4+\Boh(n^{-3}(x-y)).
\end{equation}
Similarly, using \eqref{eq:estimate_A31}, we find that
\begin{equation}\label{eq:estimate_difference_curly_A2}
\mathcal A^{(2)}(v_n)^{-1}\mathcal A^{(2)}(u_n)=I_4+\Boh(n^{-5/2}(x-y)).
\end{equation}
Finally, from \eqref{def:Ak2}, it is readily seen that
$$
A^{(2)}_1(z)\widehat D_n(z)^{-1}=\widehat D_n(z)^{-1}n^{\frac{\widehat B}{2}}A_1^{(1)}n^{-\frac{\widehat B}{2}}.
$$
This, together with \eqref{eq:derivative_rule_Dhat2}, implies that
\begin{align*}
\mathcal A^{(1)}(z)\widehat D_n(z)^{-1}&=
\begin{multlined}[t]
\widehat D_n(z)^{-1}-\frac{\left(\widehat D_n(z)^{-1}-\widehat D_n(0)^{-1}\right)n^{\frac{\widehat B}{2}}A_1^{(1)}n^{-\frac{\widehat B}{2}}}{n^3z}\\ +\frac{\widehat D_n(0)^{-1}n^{\frac{\widehat B}{2}}A_1^{(1)}n^{-\frac{\widehat B}{2}}\Boh(n^{3/2 }z)}{n^3z}
\end{multlined}
\\
&=\widehat D_n(z)^{-1}+\Boh(n^{-1/2}).
\end{align*}
According to Proposition~\ref{prop:hatD_estimates}, the right-hand side above is bounded, so from Lemma~\ref{lem:key_difference_estimate}, we obtain
\begin{equation}\label{eq:estimate_difference_curly_A1}
\widehat D_n(v_n)\mathcal A^{(1)}(v_n)^{-1}\mathcal A^{(1)}(u_n)\widehat D_n(u_n)^{-1}=I_4+\Boh(n^{3/2}(u_n-v_n))=I_4+\Boh(n^{-3/2}(x-y)).
\end{equation}

Moving towards the end of the proof, let us combine all the equations \eqref{eq:estimate_difference_hatG}--\eqref{eq:estimate_difference_curly_A1} into the definition \eqref{eq:def_hat_L0} of $\widehat L_0$ to obtain
\begin{equation}\label{eq:estimate_L0_L0inverse}
\widehat L_0(v_n)^{-1}\widehat L_0(u_n)=n^{\widehat B}(I_4+T_0)n^{-\frac{\widehat B}{2}} \left(I_4+\Boh(n^{-3/2}(x-y))\right)n^{\frac{\widehat B}{2}}(I_4+T_0)^{-1}n^{-\widehat B}.
\end{equation}
Now, having in mind \eqref{eq:inverse_T0}, it follows that
$$
n^{\frac{\widehat B}{2}}(I+T_0)n^{-\frac{\widehat B}{2}}=I_4+n^{\frac{\widehat B}{2}}T_0 n^{-\frac{\widehat B}{2}}\quad \mbox{with} \quad n^{\frac{\widehat B}{2}}T_0^k n^{-\frac{\widehat B}{2}}=\Boh(n^{-1/2}),\;\; k=1,2,
$$
and also
$$
n^{\frac{\widehat B}{2}}(I_4+T_0)^{-1}n^{-\frac{\widehat B}{2}}=I_4-n^{\frac{\widehat B}{2}}T_0n^{-\frac{\widehat B}{2}}+n^{\frac{\widehat B}{2}}T_0^2n^{-\frac{\widehat B}{2}}.
$$
Plugging these last two identities into \eqref{eq:estimate_L0_L0inverse} concludes the proof.
\end{proof}

Using Lemmas \ref{lem:Rest} and \ref{lem:L0est} in \eqref{eq:hard_edge_relation_S_R}, we see that
$$
S_+(v_n)^{-1}S_+(u_n)=\widehat P_+(v_n)^{-1}n^{\frac{\widehat B}{2}}\left(I_4+\Boh(n^{-3/2}(x-y))\right)n^{-\frac{\widehat B}{2}}\widehat P_+(u_n),\qquad n\to \infty,
$$
uniformly for $x,y$ in compact subsets of $(0,\infty)$. Thus, it is readily seen from \eqref{eq:hard_edge_kernel1} that
\begin{align*}
& n^2K_n\left(\frac{x}{n(\beta^2-\alpha^2)},\frac{y}{n(\beta^2-\alpha^2)}\right)=n^2K_n(n^2u_n,n^2v_n) \\
& =
\frac{n^3(\beta^2-\alpha^2)}{2\pi i (x-y)}
\begin{pmatrix}
0  & -v_n^{-\kappa}e^{n \lambda_{2,+}(v_n)} & e^{n \lambda_{3,+}(v_n)} &
0 \end{pmatrix}
\widehat P_+(v_n)^{-1}n^{\frac{\widehat B}{2}}
\\
& \quad \times
\left(I_4+\Boh(n^{-3/2}(x-y))\right)n^{-\frac{\widehat B}{2}}\widehat P_+(u_n)
\begin{pmatrix}
0 & u_n^{\kappa}e^{-n\lambda_{2,+}(u_n)} & e^{-n\lambda_{3,+}(u_n)} &  0
\end{pmatrix}^T,
\end{align*}
and then using \eqref{def:hatP},
\begin{multline}\label{eq:simplification_kernel_hard_edge}
\frac{1}{n(\beta^2-\alpha^2)}K_n\left(\frac{x}{n(\beta^2-\alpha^2)},\frac{y}{n(\beta^2-\alpha^2)}\right) =\frac{1}{2\pi i (x-y)}
\begin{pmatrix}
0 & -v_n^{-\kappa} & 1 & 0
\end{pmatrix}
\\
\times
P_+(v_n)^{-1}n^{\frac{\widehat B}{2}}\left(I_4+\Boh(n^{-3/2}(x-y))\right)n^{-\frac{\widehat B}{2}} P_+(u_n)
\begin{pmatrix}
0 & u_n^{\kappa}& 1 & 0
\end{pmatrix}^T.
\end{multline}

To simplify it further, we use \eqref{def:matrix_P}, Proposition~\ref{prop:hatlambda} and the definitions of $A$ and $\widehat B$ in \eqref{def:matrix_A} and \eqref{def:hatB}, respectively, to get
\begin{multline*}
\begin{pmatrix}
0 & -v_n^{-\kappa} & 1 & 0
\end{pmatrix}
P_+(v_n)^{-1}= n^{-\nu+2\kappa}
\begin{pmatrix}
0 & -y^{-\kappa}(\beta^2-\alpha^2)^{\kappa} & 1 & 0
\end{pmatrix} \\ \times
\diag\left(1, \left(\frac{f_4(u_n)}{3}\right)^{-A}\Psi_+(n^3\varphi(v_n))^{-1}\left(\frac{f_4(v_n)}{3}\right)^{-B}\right)n^{-\widehat B}
\end{multline*}
and
\begin{multline*}
P_+(u_n)
\begin{pmatrix}
0 & u_n^{\kappa}&1 & 0
\end{pmatrix}^T=
n^{\nu-2\kappa} n^{\widehat B}
\\ \times
\diag\left(1, \left(\frac{f_4(u_n)}{3}\right)^B \Psi_+(n^3\varphi(u_n)) \left(\frac{f_4(u_n)}{3}\right)^A \right)
\begin{pmatrix}
0 & x^\kappa(\beta^2-\alpha^2)^{-\kappa} & 1 & 0
\end{pmatrix}^T.
\end{multline*}

Moving forward, we now use Proposition~\ref{prop:hatlambda} to obtain
$$
\frac{f_4(u_n)}{3}=(\beta^2-\alpha^2)^{\frac{1}{3}}+\Boh(n^{-3})=\frac{f_4(v_n)}{3},\qquad n\to \infty,
$$
and from \eqref{eq:varphizero}
$$
n^3\varphi(u_n)=x(1+\Boh(n^{-3})),\qquad n^3\varphi(v_n)=y(1+\Boh(n^{-3})),\qquad n\to\infty,
$$
where the error terms above are uniform for $x,y$ in compact subsets of $(0,+\infty)$. Combining with the analyticity of $\Psi_+$, we thus conclude
\begin{multline*}
\begin{pmatrix}
0 & -v_n^{-\kappa} & 1 & 0
\end{pmatrix}
P_+(v_n)^{-1}=
n^{-\nu+2\kappa} (\beta^2-\alpha^2)^{\frac{2\kappa-\nu}{3}}
\begin{pmatrix}
0 & -y^{-\kappa} & 1 & 0
\end{pmatrix}
\\
\times
\diag\left(1,\Psi_+(y)^{-1}\right) (\beta^2-\alpha^2)^{-\frac{\widehat B}{3}}(I_4+\mathcal E_n(v_n))^{-1}n^{-\widehat B},
\end{multline*}
where $\{\mathcal E_n\}$ is a sequence of $4\times 4$ matrix-valued analytic functions on $D(r_n)$ with
\begin{equation}\label{eq:estimate_final_error}
\mathcal E_n(z)=\Boh(n^{-3})\quad \mbox{uniformly for } z\in D(r_n) \mbox{ as } n\to\infty,
\end{equation}
and
\begin{multline*}
P_+(u_n)
\begin{pmatrix}
0 & u_n^{\kappa}&1 & 0
\end{pmatrix}^T= n^{\nu-2\kappa}(\beta^2-\alpha^2)^{\frac{\nu-2\kappa}{3}}n^{\widehat B}(I_4+\mathcal E_n(u_n))
\\
\times
(\beta^2-\alpha^2)^{\frac{\widehat B}{3}}\diag\left(1,\Psi_+(x)\right)
\begin{pmatrix}
0 & x^{\kappa} & 1 & 0
\end{pmatrix}^T
\end{multline*}
with the same error function $\mathcal E_n$.

Inserting these last two identities into \eqref{eq:simplification_kernel_hard_edge}, we have
\begin{align*}
& \frac{1}{n(\beta^2-\alpha^2)}K_n\left(\frac{x}{n(\beta^2-\alpha^2)},\frac{y}{n(\beta^2-\alpha^2)}\right)
\\
& =
\frac{1}{2\pi i (x-y)}
\begin{pmatrix}
0 & -y^{-\kappa} & 1 & 0
\end{pmatrix}
\diag\left(1,\Psi_+(y)^{-1}\right)
(\beta^2-\alpha^2)^{-\frac{\widehat B}{3}}\left(I_4+\mathcal E_n(v_n)\right)^{-1}n^{-\frac{\widehat B}{2}}
\\
& \quad \times
\left(I_4+\Boh(n^{-\frac{3}{2}}(x-y))\right)
n^{\frac{\widehat B}{2}} \left(I_4+\mathcal E_n(u_n)\right)(\beta^2-\alpha^2)^{\frac{\widehat B}{3}}\diag(1,\Psi_+(x))
\begin{pmatrix}
0 \\ x^{\kappa} \\ 1 \\ 0
\end{pmatrix}
\\
& =
\frac{1}{2\pi i (x-y)}
\begin{pmatrix}
0 & -y^{-\kappa} & 1 & 0
\end{pmatrix}
\diag\left(1,\Psi_+(y)^{-1}\right)
(\beta^2-\alpha^2)^{-\frac{\widehat B}{3}}
\\
& \quad \times
\left(\left(I_4+\mathcal E_n(v_n)\right)^{-1}\left(I_4+\mathcal E_n(u_n)\right)+\Boh(n^{-\frac{1}{2}}(x-y))\right)
(\beta^2-\alpha^2)^{\frac{\widehat B}{3}}\diag(1,\Psi_+(x))
\begin{pmatrix}
0 \\ x^{\kappa} \\ 1 \\ 0
\end{pmatrix}.
\end{align*}

In virtue of \eqref{eq:estimate_final_error}, we can once more apply Lemma~\ref{lem:key_difference_estimate} to get that
$$
(I_4+\mathcal E_n(v_n))^{-1}(I_4+\mathcal E_n(u_n))=I_4+\Boh(n^{-3}(x-y)), \qquad \mbox{as }n\to\infty,
$$
and we finally arrive at
\begin{align*}
&\frac{1}{n(\beta^2-\alpha^2)}K_n\left(\frac{x}{n(\beta^2-\alpha^2)},\frac{y}{n(\beta^2-\alpha^2)}\right)
\\
& =
\frac{1}{2\pi i (x-y)}
\begin{pmatrix}
0 & -y^{-\kappa} & 1 & 0
\end{pmatrix}
\diag\left(1,\Psi_+(y)^{-1}\right)
\left(I_4+\Boh(n^{-\frac{1}{2}}(x-y))\right)
\\
& \quad \times
\diag(1,\Psi_+(x))
\begin{pmatrix}
0 & x^{\kappa} & 1 & 0
\end{pmatrix}^T
\\
& =
\frac{1}{2\pi i(x-y)}
\begin{pmatrix}
-y^{-\kappa} & 1 & 0
\end{pmatrix}
\Psi_+(y)^{-1}\Psi_+(x)
\begin{pmatrix}
x^{\kappa} & 1 & 0
\end{pmatrix}^T + \Boh(n^{-1/2}),
\end{align*}
where, as always, the error term is uniform for $x,y$ in compact subsets of $(0,\infty)$. Hence, we obtain that
$$
\lim_{n\to\infty }\frac{1}{n(\beta^2-\alpha^2)}K_n\left(\frac{x}{n(\beta^2-\alpha^2)},\frac{y}{n(\beta^2-\alpha^2)}\right) = K_\infty(x,y)
$$
uniformly for $x,y$ in compact subsets of $(0,\infty)$, where
$$
K_\infty(x,y)=\frac{1}{2\pi i(x-y)}
\begin{pmatrix}
-y^{-\kappa} & 1 & 0
\end{pmatrix}
\Psi_+(y)^{-1}\Psi_+(x)
\begin{pmatrix}
x^{\kappa} \\ 1 \\ 0
\end{pmatrix}.
$$

To conclude the proof of Theorem~\ref{thm:hardedge}, it remains to relate $K_\infty$ with $K_{\nu,\kappa}$ as in \eqref{def:MeijerKer}. To do so, first observe that $\Psi$ - and hence $K_\infty$ - does not depend on $\alpha$ and $\beta$, as can be seen from the RH problem \ref{rhp:MeijerG} whose conditions do not depend on $\alpha$ and $\beta$. Thus, it is enough to relate $K_\infty$ with $K_{\nu,\kappa}$ for one specific choice of $\alpha$ and $\beta$, which we take to be matching those in \eqref{def:interpolating_parameters}, that is,
$$
\beta=\frac{1}{2\tau}+\frac{1}{2},\quad \alpha=\frac{1}{2\tau}-\frac{1}{2}, \quad \mbox{so}\quad \beta^2-\alpha^2=\frac{1}{\tau},
$$
where $0<\tau<1$ is any fixed number. For this specific coupling, our model \eqref{def:coupledmatrix} coincides with the model considered by Liu \cite{Liu16}, so comparing\footnote{The correspondence between our parameters $\alpha=\alpha_{SZ}$, $\beta=\beta_{SZ}$ and $\tau=\tau_{SZ}$ and Liu's parameters $\delta_L$, $\alpha_L$ and $\mu_L$ is $\beta_{SZ}=\alpha_L$, $\alpha_{SZ}=\delta_L$ and $\tau_{SZ}=\beta_L$.} with \cite[Theorem~1.3(i) and Equation~(5.20)]{Liu16} we arrive at
$$
K_\infty(x,y)=\left(\frac{y}{x}\right)^{\kappa/2} K_{\nu,\kappa}(y,x).
$$
Alternatively, the above relation can be seen from the RH characterization of the Meijer G-kernel commented in \cite[Section 4.2.5]{BB15}.

This completes the proof of Theorem \ref{thm:hardedge}. \qed
\begin{appendices}

\section{Heuristics on the vector equilibrium problem}\label{section:heuristics_equil_probl}

In this section, we give some heuristic arguments on how to formulate the vector equilibrium problem introduced in Section \ref{subsec:vep}, which is closely related to the asymptotic analysis of the RH problem for $Y$.

Recall that the goal of the second transformation $X \to T$ is to `normalize' the large $z$ asymptotics of $X$ and to prepare for the opening of lenses.  We assume that, at this moment, it takes the following form:
\begin{equation}\label{eq:transformationT}
T(z)=\widehat C X(z) \diag(e^{n\lambda_1(z)},e^{n\lambda_2(z)},e^{n\lambda_3(z)},e^{n\lambda_4(z)}),
\end{equation}
where $\widehat C $ is a constant matrix and the $\lambda$-functions are of the form
\begin{equation}\label{eq:lambda_functions}
\begin{aligned}
\lambda_1(z) & =\int^z C^{\mu_1}(s)\ud s+V_1(z), \\
\lambda_2(z) & =\int^z C^{\mu_2}(s)\ud s-\int^z C^{\mu_1}(s)\ud s+V_2(z), \\
\lambda_3(z) & =\int^z C^{\mu_3}(s)\ud s-\int^z C^{\mu_2}(s)\ud s+V_3(z), \\
\lambda_4(z) & =-\int^z C^{\mu_3}(s)\ud s+V_4(z).
\end{aligned}
\end{equation}
In the above formulas, $C^{\mu}(z)$ is the Cauchy transform of a measure $\mu$ given in \eqref{def:Cauchytransform}, $\mu_1$, $\mu_2$ and $\mu_3$ are three measures satisfying
\begin{align}
& \supp\mu_1\subset \R_-, \quad \supp\mu_2\subset \R_+, \quad \supp\mu_3\subset\R_-, \label{eq:location_supports}
\\
& 2|\mu_1|=|\mu_2|=2|\mu_3|=1, \label{eq:total_masses}
\end{align}
and $V_1,V_2,V_3,V_4$ are four functions to be determined.

As $z\to \infty$, it is readily seen from \eqref{eq:transformationT} and \eqref{eq:Xinfty} that,
\begin{multline*}
T(z)=(I_4+\Boh(z^{-1}))B(z) \\
\times \diag\left(z^{\frac{n}{2}}e^{n(\lambda_1-2\alpha z^{\frac{1}{2}})},z^{\frac{n}{2}}e^{n(\lambda_2+2\alpha z^{\frac{1}{2}})},z^{-\frac{n}{2}}e^{n(\lambda_3+2\beta z^{\frac{1}{2}})},z^{-\frac{n}{2}}e^{n(\lambda_4-2\beta z^{\frac{1}{2}})}\right).
\end{multline*}
The normalization requirement then invokes us to expect that, as $z\to \infty$,
\begin{equation}\label{eq:asymptotics_lambdas}
\begin{aligned}
\lambda_1(z)-2\alpha z^{\frac{1}{2}}+\frac{1}{2}\log z & = \boh(1), \qquad
\lambda_2(z)+2\alpha z^{\frac{1}{2}}+\frac{1}{2}\log z  = \boh(1),
\\
\lambda_3(z)+2\beta z^{\frac{1}{2}}-\frac{1}{2}\log z & = \boh(1), \qquad
\lambda_4(z)-2\beta z^{\frac{1}{2}}-\frac{1}{2}\log z  = \boh(1).
\end{aligned}
\end{equation}
On the other hand, in view of \eqref{eq:asymptotics_cauchy_transf_log_pot}, it follows that, as $z\to \infty$,
\begin{equation*}
\begin{aligned}
\lambda_1(z) & =V_1(z)-\frac{1}{2}\log z + \Boh(z^{-1}), \qquad
\lambda_2(z)  =V_2(z)-\frac{1}{2}\log z + \Boh(z^{-1}),
\\
\lambda_3(z) & =V_3(z)+\frac{1}{2}\log z + \Boh(z^{-1}), \qquad
\lambda_4(z)  =V_4(z)+\frac{1}{2}\log z + \Boh(z^{-1}).
\end{aligned}
\end{equation*}
Comparing these asymptotics with \eqref{eq:asymptotics_lambdas}, it is easily seen that we should have
\begin{equation}\label{eq:pre_potentials}
\begin{aligned}
V_1(z)&=2\alpha z^{\frac{1}{2}}, \qquad && V_2(z)=-2\alpha z^{\frac{1}{2}}, \\
V_3(z)&=-2\beta z^{\frac{1}{2}}, \qquad && V_4(z)=2\beta z^{\frac{1}{2}}.
\end{aligned}
\end{equation}

We next come to the jump condition satisfied by $T$. Taking into account \eqref{eq:transformationT}, \eqref{eq:location_supports} and \eqref{jump for X}, it is readily seen that
\begin{equation*}
T_+(x)=T_-(x)J_{T}(x), \qquad x\in\R,
\end{equation*}
where
\begin{multline*}
J_T(x)= \diag\left(1,e^{n(\lambda_{2,+}(x)-\lambda_{2,-}(x))},e^{n(\lambda_{3,+}(x)-\lambda_{3,-}(x))},1\right)
 \\
 + x^{\frac{\nu}{2}}e^{n(\lambda_{3,+}(x)-\lambda_{2,-}(x))}E_{23}, \qquad x\in \R_+,
 \end{multline*}
 and
 \begin{multline*}
J_T(x)= \Lambda\diag\left(e^{n(\lambda_{1,+}(x)-\lambda_{1,-}(x))},e^{n(\lambda_{2,+}(x)-\lambda_{2,-}(x))},e^{n(\lambda_{3,+}(x)-\lambda_{3,-}(x))},
e^{n(\lambda_{4,+}(x)-\lambda_{4,-}(x))}\right) \\ - e^{n(\lambda_{1,+}(x)-\lambda_{2,-}(x))}E_{21}- e^{n(\lambda_{4,+}(x)-\lambda_{3,-}(x))}E_{34}, \qquad x\in \R_-,
 \end{multline*}
with
$$
\Lambda:=\diag(e^{-\pi i \kappa \sigma_3},e^{\pi i (\nu-\kappa) \sigma_3}).
$$
We now look at the non-diagonal entries of the jump matrix $J_T$. It is expected that these entries to be constant on the supports of the measures. Taking their real part, we arrive at the following conditions.
\begin{itemize}
\item $(2,3)$-entry on  $\R_+$:
\begin{align*}
2U^{\mu_2}(x)-U^{\mu_1}(x)-U^{\mu_3}(x)+\re \, (V_2(x)-V_3(x))=\ell_2;
\end{align*}

\item $(2,1)$-entry on $\R_-$:
$$
2U^{\mu_1}(x)-U^{\mu_2}(x)+\re \, (V_{1,+}(x)-V_{2,-}(x))=\ell_1;
$$

\item $(3,4)$-entry on  $\R_-$:
$$
2U^{\mu_3}(x)-U^{\mu_2}(x)+\re \,(V_{3,-}(x)-V_{4,+}(x))=\ell_3,
$$
\end{itemize}
where $\ell_j$, $j=1,2,3$, is certain constant. From \eqref{eq:pre_potentials}, we thus find that the potentials $Q_1$, $Q_2$ and $Q_3$ acting on the measures $\mu_1$, $\mu_2$ and $\mu_3$ should be
\begin{align*}
Q_1(x) & = \re\,(V_{1,+}(x)-V_{2,-}(x))= 2\alpha (\sqrt x)_+ + 2\alpha (\sqrt x)_-=0, \\
Q_2(x) & = \re\,(V_2(x)-V_3(x))= 2(\beta-\alpha)\sqrt{x}, \\
Q_3(x) & = \re\,(V_{3,-}(x)-V_{4,+}(x))=-2\beta (\sqrt{x})_- -2\beta (\sqrt x)_+ = 0,
\end{align*}
as shown in \eqref{definition_vector_energy}.

Finally, we explain the upper constraint. The fact that there is an upper constraint for $\mu_1$ but not for $\mu_2,\mu_3$ is connected to the form of the jumps: the equilibrium conditions for $\mu_1$ play a role in a lower triangular block of the jump matrix, whereas for the remaining measures the corresponding equilibrium conditions appear in an upper triangular block. In virtue of the direction of the variational inequalities for the equilibrium problem, we thus expect that associated to $\mu_1$ there should be an upper constraint, but no upper constraint should appear on the remaining measures.

To find the explicit form of the constraint, again some ansatz is needed. We expect that the functions $\lambda_1'$, $\lambda_2'$, $\lambda_3'$ and $\lambda_4'$ should all be solutions to the same algebraic equation (a.k.a. spectral curve). From the sheet structure for the associated Riemann surface, we also expect that $\lambda_1'$ is analytic across the places where $\sigma$ is active, that is, $\lambda_1'$ should be analytic across $\R_-\setminus \supp(\sigma-\mu_1)$. Hence,
$$
\lambda_{1,+}'(x)-\lambda'_{1,-}(x)=0,\qquad x\in \R_- \setminus \supp(\sigma-\mu_1).
$$
Using the explicit expression for $\lambda_1$ (see \eqref{eq:lambda_functions} and \eqref{eq:pre_potentials}) and Plemelj's formula \eqref{eq:plemelj_relations}, we can rewrite the identity above as
\begin{multline*}
\frac{1}{2\pi i}\left( C^{\mu_1}_+(x)-C^{\mu_1}_-(x) +V'_{1,+}(x)-V'_{1,-}(x)\right) \\=\frac{\ud \mu_1}{\ud x}(x)+\frac{\alpha}{2\pi i}\left((x^{-\frac{1}{2}})_+-(x^{-\frac{1}{2}})_+\right)=0,\qquad x\in \R_-\setminus \supp(\sigma-\mu_1).
\end{multline*}
Taking into account that $\mu_1=\sigma$ on $\R_-\setminus \supp(\sigma-\mu_1)$, the identity above gives us
$$
\frac{\ud \sigma}{\ud x}(x)=\frac{\alpha}{\pi \sqrt{|x|}},\qquad x\in \R_-\setminus \supp(\sigma-\mu_1),
$$
which is \eqref{def:constraint_measure}.

\end{appendices}

\section*{Acknowledgements}
Guilherme Silva thanks the hospitality of the School of Mathematical Sciences at Fudan university where most of this research was carried out, and also thanks Leslie Molag for discussions related to this work.  Lun Zhang was partially supported by National Natural Science Foundation of China under grant number 11822104, by The Program for Professor of Special Appointment (Eastern Scholar) at Shanghai Institutions of Higher Learning, and by Grant EZH1411513 from Fudan University.


\end{document}